\documentclass[manyauthors]{fundam}

 \pdfoutput=1

\publyear{24}
\papernumber{2195}
\volume{192}
\issue{3-4}

\finalVersionForARXIV

 \def \VersionWithComments{}
\usepackage[right]{lineno}

\usepackage{packages}
\usepackage{macros}
\usepackage{cleveref}

\begin{document}

\title{A Rewriting-logic-with-SMT-based  Formal Analysis and Parameter
        Synthesis Framework for Parametric Time Petri Nets}

\author{Jaime Arias\\
    LIPN, CNRS UMR 7030 \\
    Universit\'e Sorbonne Paris Nord, France \\
    jaime.arias@lipn.univ-paris13.fr
  \and
  Kyungmin Bae \\
    Pohang University of Science and Technology\\
    South Korea \\
    kmbae@postech.ac.kr
  \and
  Carlos Olarte\thanks{Address for correspondence:  LIPN, CNRS UMR 7030, Universit\'e Sorbonne Paris Nord, France.}
    \\
    LIPN, CNRS UMR 7030 \\
    Universit\'e Sorbonne Paris Nord, France \\
    carlos.olarte@lipn.univ-paris13.fr
  \and
  Peter Csaba {\"O}lveczky \\
    University of Oslo, Norway \\
    peterol@ifi.uio.no
  \and
  Laure Petrucci \\
    LIPN, CNRS UMR 7030 \\
    Universit\'e Sorbonne Paris Nord, France \\
    laure.petrucci@lipn.univ-paris13.fr
}

\maketitle

 \runninghead{J. Arias et al.}{Formal Analysis and Parameter Synthesis for Time Petri Nets}

\vspace*{-6mm}
\begin{abstract}
 This paper presents
a concrete and a symbolic rewriting
logic semantics for
\emph{parametric time Petri nets with inhibitor arcs} (PITPNs), a
   flexible model   of timed systems  where
 parameters are allowed in firing bounds.
We prove that our semantics is bisimilar to the ``standard'' semantics
of PITPNs.
This allows us to use
 the rewriting logic tool Maude,
combined with SMT solving,  to
provide sound and complete formal analyses for PITPNs.
We develop and implement a new general folding
 approach for symbolic reachability, so  that Maude-with-SMT
 reachability analysis
terminates whenever the parametric state-class graph of the PITPN is
 finite.
 Our work opens up
 the possibility of using the many formal analysis capabilities of
 Maude---including full LTL model checking,
 analysis with user-defined execution strategies,
  and even statistical model
checking---for such nets.
We illustrate this by explaining how almost all formal analysis and
parameter synthesis methods
supported by the state-of-the-art PITPN tool
\romeo{} can be performed using  Maude with SMT.  In addition,  we also support
 analysis and parameter
synthesis from \emph{parametric} initial markings, as well as  full
LTL model
checking and analysis  with
user-defined execution strategies.
Experiments  show that our methods outperform
\romeo{} in many~cases.
\end{abstract}
\begin{keywords}
parametric timed Petri nets,  
rewriting logic, Maude,
 SMT, parameter synthesis, symbolic reachability analysis
\end{keywords}


\section{Introduction}\label{sec:intro}

Time Petri nets \cite{Merlin74,DBLP:reference/crc/VernadatB07} extend Petri nets to real-time systems
  by adding time intervals to transitions. 
However, in system design we often do not  know in advance the concrete
values of key system parameters, and want to find
those parameter values  that make the system behave as desired. \emph{Parametric
  time Petri nets with inhibitor arcs}
(PITPNs)~\cite{paris-paper,DBLP:conf/formats/GrabiecTJLR10,EAGPLP13,DBLP:journals/fuin/LimeRS21}
extend time Petri
nets to the setting  where  bounds on
when transitions can fire are unknown or only partially
known.

\medskip
The modeling and formal analysis of PITPNs---including synthesizing the
values of the parameters which make the system satisfy desired
properties---are supported by the state-of-the-art tool
\romeo~\cite{romeo}, which has been applied to a number of
applications,  including oscillatory biological
systems~\cite{romeo-biology}, aerial video tracking
systems~\cite{romeo-aerial}, and distributed software
commissioning~\cite{romeo-software}.
\romeo{}
supports the analysis and parameter synthesis for reachability (is a
certain marking reachable?), liveness (will a certain marking be
reached in all behaviors?), time-bounded ``until,'' and bounded
response (will each $P$-marking
be followed by a $Q$-marking within time $\Delta$?), all from \emph{concrete}
initial markings.
\romeo{} does  not support a number of desired features,
including:
\begin{itemize}
  \item Broader set of system properties,  \eg{}  full (\ie{}
    nested) temporal logic properties.
\item  Starting  with \emph{parametric} initial
  markings and synthesizing also the initial markings that make the
  system satisfy desired properties.
 \item Analysis with user-defined execution strategies. For
      example, what  happens if I always choose to fire
      transition $t$ instead of $t'$ when they are both enabled at the
      same time?
       It is often possible to  \emph{manually
      change the model} to analyze the system under such scenarios, but
      this is arduous and error-prone.
  \item Providing a ``testbed''
  for PITPNs in which different analysis methods and algorithms can quickly be
     developed, tested, and evaluated.
     This is not well
    supported by \romeo, which is a high-performance tool with
    dedicated algorithms implemented in \texttt{C++}.
   \end{itemize}

   Rewriting   logic~\cite{Mes92,20-years}---supported by the Maude language and
    tool~\cite{maude-book}, and by Real-Time
   Maude~\cite{tacas08,wrla14} for real-time systems---is an
   expressive logic for
   distributed and real-time systems. In rewriting logic, any
   computable data type can be specified as an (algebraic)  equational
specification, and  the dynamic behaviors of a system are
specified by rewriting rules over terms (representing states). Because
of its expressiveness, Real-Time Maude has been successfully applied
to a number of large and sophisticated real-time systems---including 50-page
active networks and IETF protocols~\cite{aer-journ,sefm09},
state-of-the-art wireless sensor network algorithms involving areas,
angles, etc.~\cite{wsn-tcs}, scheduling algorithms with unbounded
queues~\cite{fase06}, airplane turning
algorithms~\cite{airplane-journ}, large cloud-based transaction
systems~\cite{kokichi,sefm14}, mobile ad hoc
networks~\cite{manets-journ}, human
multitasking~\cite{giovanna-journ}, and so on---beyond the scope of
most popular formalisms for real-time systems such as timed automata
and PITPNs.  Its
expressiveness has also made Real-Time Maude a useful semantic framework
and formal analysis backend for (subsets of) industrial
modeling
languages (e.g.,~\cite{fmoods10,musab-l,ptolemy-journ,fase12,carolyn}).

This expressiveness comes at a price: most analysis problems
are undecidable in general. Real-Time Maude uses
explicit-state analysis  where  only
\emph{some} points in time are visited.  All possible system
behaviors are therefore \emph{not} analyzed (for dense time domains),
and hence the  analysis is  unsound
in many cases~\cite{wrla06}.

This paper exploits the integration of SMT solving into Maude
  to address the
first problem above (more features for PITPNs)  and to take the second step
towards addressing the second problem (developing sound and complete
analysis methods for rewriting-logic-based real-time systems).

Maude combined with SMT solving, \eg{} as implemented in the Maude-SE
tool~\cite{maude-se},
allows us to perform \emph{symbolic rewriting}
of ``states'' $\phi\; ||\; t$, where the term $t$ is a state pattern
that
contains variables, and $\phi$ is an SMT constraint restricting
the possible values of those variables.

After  giving some necessary background to PITPNs,  rewriting
logic, and Maude-with-SMT  in Section~\ref{sec:prelim}, we provide in
Section~\ref{sec:concrete}   a ``concrete'' rewriting
logic semantics for (instantiated) PITPNs
in ``Real-Time Maude style''~\cite{rtm-journ}. In a dense-time setting,
such as for PITPNs, this model is not executable.
Section~\ref{sec:concrete-ex} shows how we can do
(in general unsound) time-sampling-based analysis where time
increases in discrete steps, of concrete nets, to quickly experiment
with different  values for the parameters. Additionally,
we show how to perform full LTL model checking on these models.

Section~\ref{sec:sym} gives a ``symbolic'' rewriting logic semantics for
\emph{parametric} PITPNs,  and shows how
to perform (sound) symbolic analysis of such nets using Maude-with-SMT.
However, existing symbolic reachability analysis methods,
including  ``folding'' of symbolic states, may fail to terminate even
when the state class graph of the PITPN is finite (and
hence \romeo{} analyses terminate). We therefore develop and
implement a new  method for ``folding'' symbolic states for
reachability analysis in Maude-with-SMT, and show that this new
reachability analysis method terminates whenever the state class graph
of the PITPN is finite.

In Sections~\ref{sec:sym} and \ref{sec:analysis} we show how a range of
 formal analyses and parameter synthesis can be performed with
Maude-with-SMT, including unbounded and time-bounded reachability
analysis. We  show in Section~\ref{sec:analysis} how all analysis
methods supported by \romeo---with one  exception:
the time bounds in some temporal formulas cannot be parameters---also
can be performed in Maude-with-SMT. In addition, we support state
properties on both markings and ``transition clocks,'' analysis and
parameter synthesis for \emph{parametric} initial markings, model
checking full (\ie{} nested) temporal logic formulas, and
analysis w.r.t.\ user-defined execution strategies, as illustrated in
Section~\ref{sec:analysis}.
Our methods are formalized/implemented in Maude itself, using Maude's
meta-programming features. This makes it very easy to
develop and prototype new analysis methods for PITPNs.

This work also constitutes the second step in our quest to develop
sound and complete formal analysis methods for dense-time real-time
systems in Real-Time Maude. One  reason
for presenting both a ``standard'' Real-Time Maude-style concrete
semantics in Section~\ref{sec:concrete} \emph{and} the symbolic
semantics in Section~\ref{sec:sym} is to explore how we can transform
Real-Time Maude models into Maude-with-SMT models for symbolic
analysis.
In our first step in this quest, we studied symbolic rewrite methods
for the much simpler parametric timed automata (PTA)~\cite{ftscs22,ftscs-journal}.
PTAs have a much
 simpler structure than PITPNs, which can have an unbounded number of
 tokens in reachable markings. 

In Section~\ref{sec:benchmarks} we benchmark both \romeo{} and
our Maude-with-SMT methods on some PITPNs. Somewhat surprisingly, in
many cases our high-level prototype ``interpreter'' outperforms \romeo{}.
We also discovered that \romeo{} answered ``maybe'' in some cases
where Maude found solutions,  where those solutions were
proved valid by running \romeo{}
with the additional constrains on the parameters found by Maude. Additionally,
 \romeo{} sometimes failed to
synthesize parameters even when solutions existed.
We also compare the performance of our previous
PITPN analyzer presented in~\cite{DBLP:conf/apn/AriasBOOPR23} with our
new version, which incorporates the elimination of existentially
quantified variables and optimizations in the
folding procedure.

Finally, Section~\ref{sec:related} discusses related work, and
Section~\ref{sec:concl} gives some concluding remarks and suggests topics
for future work.

\medskip
All  executable Maude files with
analysis commands, tools for translating \romeo{} files into Maude,  and data
from the benchmarking are available at~\cite{pitpn2maude}.

\medskip
This paper is an extended version of the conference
paper~\cite{DBLP:conf/apn/AriasBOOPR23}.
Additional
contributions include:

\begin{itemize}
 \item
     We present full proofs and explain in more depth the
     different rewrite theories and the folding procedure proposed.
     We also provide more examples on how to perform
     explicit-state analysis (\Cref{sec:concrete-ex})  and symbolic analysis
     within our framework (\Cref{sec:analysis}).

    \item The folding procedure in \cite{DBLP:conf/apn/AriasBOOPR23}
     was  implemented by relying on dedicated procedures in
     the Z3
     solver for  eliminating existentially quantified variables. By
     adapting and integrating  the Fourier-Motzkin elimination procedure described in
     \cite{ftscs-journal}, our Maude-with-SMT analysis can now be done with any
     SMT solver connected to Maude (\Cref{sec:imp:folding}). As demonstrated in
     \Cref{sec:benchmarks}, this new   procedure significantly
     outperforms its predecessor. This also allows us to
     provide  a more comprehensive
     performance comparison of Maude executed with the   SMT solvers
      Yices2, CVC4, and Z3  in
     \Cref{sec:benchmarks}.

    \item The implementation of the folding procedure in
        \cite{DBLP:conf/apn/AriasBOOPR23} used the standard search command in
        Maude, and it could  detect previously visited states only
        within the same branch of the search tree. In contrast, we now
        use the meta-level facilities of Maude to also implement a breadth-first search
        procedure with folding which  maintains a global set of already visited
        states across all branches of the search tree (\Cref{sec:imp:folding}).
        Despite the additional overhead incurred by the meta-level
        implementation, the reduction in state space is substantial in certain
        scenarios, leading to  improved performance compared to the standard
        search  command.

    \item The new \Cref{sec:analysis:methods}, and the corresponding
      Maude implementation,
        provide a user-friendly syntax for
        executing the different analysis  methods (some of them beyond
        what is supported by \romeo{}) and different implementations
        of the same analysis method, including
        reachability with/without folding, solving parameter (both
        firing bounds and markings)
        synthesis problems,
        time-bounded reachability analysis,
        synthesis of parameters when the net is executed with a user-defined strategy,
        and model checking full LTL and non-nested (T)CTL formulas. We
        have also provided
        a user-friendly
        syntax for the properties to be verified, including state properties
        on markings and ``transition clocks.''
\end{itemize}

\section{Preliminaries} \label{sec:prelim}
This section gives some necessary background to
transition systems and bisimulations \cite{model-checking},
parametric time Petri nets with inhibitor arcs \cite{paris-paper},
rewriting logic \cite{Mes92}, rewriting modulo SMT
\cite{rocha-rewsmtjlamp-2017},
and Maude~\cite{maude-book} and Maude-SE \cite{maude-se}.

\subsection{Transition systems and bisimulations}

 A \emph{transition system} $\mathcal{A}$ is a triple $(A,
a_{0}, \rightarrow_{\mathcal{A}})$,  where $A$ is a set of
\emph{states}, $a_{0}\in A$ is the \emph{initial state}, and
$\rightarrow_{\mathcal{A}}\,\subseteq A \times A$ is a \emph{transition
  relation}.
We say that $\mathcal{A}$ is \emph{finite} if the set of states
reachable by $\rightarrow_{\mathcal{A}}$ from $a_0$ is finite.
A  relation
$\sim \,\subseteq  A \times B$ is a \emph{bisimulation}~\cite{model-checking}
from $\mathcal{A}$ to $\mathcal{B} = (B, b_0, \rightarrow_{\mathcal{B}})$
iff: (i) $a_0 \sim b_0$; and (ii) for all $a,b$ s.t. $a\sim b$:
if
$a \rightarrow_{\mathcal{A}} a'$ then there is  a $b'$ s.t. $b
\rightarrow_{\mathcal{B}} b'$ and $a' \sim b'$, and,  vice versa,  if
$b \rightarrow_{\mathcal{B}} b''$,
then there is a $a''$ s.t. $a \rightarrow_{\mathcal{A}} a''$ and
$a'' \sim b''$.

\subsection{Parametric time Petri nets with inhibitor arcs}\label{sec:PTPN}

We recall  the definitions from~\cite{paris-paper}.  We denote by
$\grandn$, $\grandqplus$, and $\grandrplus$ the
natural numbers, the non-negative rational numbers,  and the non-negative real
numbers, respectively.
For sets
$A$ and $B$, we sometimes write
$B^A$ for the set $[A \rightarrow B]$ of
functions from $A$ to $B$.
Throughout this paper we assume a finite set
$\Param=\{\param_1, \dots,\param_{\ParamCard} \}$
of \emph{time parameters}.
A \emph{parameter valuation} $\pi$ is a function
$\pi : \Param \rightarrow \grandqplus$.
A (linear) \emph{inequality} over  $\Param$ is an expression
$\sum_{1 \leq i \leq \ParamCard} a_i \param_i \prec b$,
where $\prec \in \{<, \leq,=,\geq,>\}$
and $a_{i},b \in \grandr$.
%
A \emph{constraint} is a conjunction of such inequalities.
$\setP$ denotes the set of all constraints over $\Param$.
A parameter valuation~$\pi$ \emph{satisfies} a constraint $K\in
\setP$, written $\pi \models K$,
if the expression obtained by replacing each parameter~$\param$ in~$K$
with~$\pi(\param)$ evaluates to true.
%
A closed\footnote{Since we work with SMT constraints, we could also easily
  accommodate  open intervals.} interval $\interval$ of $\grandrplus$
is a $\grandqplus$-interval
if its left endpoint $\leftEP{\interval}$ belongs to $\grandqplus$ and
its right endpoint $\rightEP{\interval}$ belongs to $\grandqplus \cup
\{ \infty \}$, where the infinity value $\infty$ satisfies the usual
properties. 
We denote by $\Interval(\grandqplus)$ the set of
$\grandqplus$-intervals. 
A parametric time interval is a function $\parInterval : {\grandqplus}^\Param
\rightarrow \Interval(\grandqplus)$ that associates with each parameter
valuation a $\grandqplus$-interval.
The set of parametric time intervals
over $\Param$ is denoted $\ParInterval (\Param)$.

\begin{definition}[Parametric time Petri net with inhibitor arcs]
    A \emph{parametric time Petri net with inhibitor arcs}
    (PITPN)~\cite{paris-paper}  is
    a tuple
    \[\PN = \tuple{\Place, \Transition, \Param, \relPre{(.)},
      \relPost{(.)}, \relInhib{(.)}, \markingInit, \parIntervalStatic,
      \Kinit}\] where
    \begin{itemize}
        \item $\Place = \{ \place_1, \dots, \place_\PlaceCard \}$ is a
          non-empty finite set (of \emph{places}),
        \item $\Transition = \{ \transition_1, \dots,
          \transition_\TransitionCard \}$ is a non-empty finite set (of
          \emph{transitions}), with $P \cap T = \emptyset$,
        \item $\Param = \{ \param_1, \dots, \param_\ParamCard \}$ is a
          finite set (of \emph{parameters}),
        \item $\relPre{(.)} \in
          [\Transition \rightarrow \grandn^\Place]$ is the
          \emph{backward} \emph{incidence function},
       \item $\relPost{(.)} \in
          [\Transition \rightarrow \grandn^\Place]$ is the
          \emph{forward} \emph{incidence function},
       \item $\relInhib{(.)} \in [\Transition\rightarrow\grandn^\Place]$
        is the
          \emph{inhibition function},
        \item $\markingInit \in \grandn^\Place$ is the \emph{initial marking},
        \item $\parIntervalStatic \in
          [\Transition\rightarrow\ParInterval(\Param)]$ 
          assigns a \emph{parametric time interval} to each
          transition, and
        \item $\Kinit \in \setP$ is a satisfiable \emph{initial
            constraint} over $\Param$, and must
          ensure that \\ $\leftEP{(\parIntervalStatic
            (\transition)(\pi))} \leq \rightEP{(\parIntervalStatic
            (\transition)(\pi))}$ for all $\transition \in \Transition$
          and all parameter valuations $\pi$ satisfying $\pi \models
          \Kinit$.
        \end{itemize}
If  $\Param = \emptyset$ then $\PN$ is a (non-parametric)
\emph{time Petri net with inhibitor arcs} (ITPN).
\end{definition}

\begin{figure}[ht]
\begin{tabular}{cc}
 \includegraphics[width=0.48\textwidth]{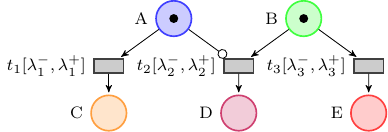}&
 \includegraphics[width=0.40\textwidth]{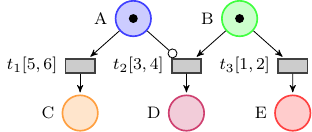}\\
(a) A PITPN~$\PN$. \label{fig:example:PITPN:param} &
(b) The ITPN~$\valuate{\PN}{\py}$.
\end{tabular}
         \caption{A PITPN and its valuation.}\label{fig:example:PITPN:valuated}
   \label{fig:example}
\end{figure}

A transition firing interval endpoint should typically either be a
non-negative rational number, the infinity value $\infty$, or a single
parameter~\cite{paris-paper}.  However, for convenience we also allow more complex
endpoints as e.g., $[2a, 2a]$ where $a$ is a parameter  (see
Fig.~\ref{fig:ex-scheduling}), and assume that
all transition firing interval endpoints can be defined as a linear
expression over the parameters.  

A \emph{marking} of  $\PN$ is an element
$\marking\in\grandn^P$, where $\marking(\place)$
is the number of tokens in place $\place$.  For a parameter valuation
$\pi$,
$\valuate{\PN}{\py}$ denotes the  ITPN where each
occurrence of $\lambda_i$ in the PITPN $\PN$  has been replaced by
$\pi(\lambda_i)$.

\begin{example}
The ITPN in Fig.~1b  corresponds
to the PITPN in  Fig.~1a where the parameters
are instantiated with
$\py = \{ \param_1^- \rightarrow 5 , \param_1^+ \rightarrow 6 ,
\param_2^- \rightarrow 3 , \param_2^+ \rightarrow 4 , \param_3^-
\rightarrow 1 , \param_3^+ \rightarrow 2 \}$.
\end{example}

The \emph{concrete semantics} of a PITPN $\PN$ is defined
in terms of concrete ITPNs $\valuate{\PN}{\py}$ where
$\pi\models\Kinit$.
We say that a transition~$\transition$ is \emph{enabled} in~$\marking$ if $\marking
\geq \relPre{\transition}$ (the number of tokens in
$\marking$ in each input place of $\transition$ is greater than or
equal to the value on the arc between this place and $t$).
A transition~$\transition$ is \emph{inhibited} if the place connected to
one of its inhibitor arcs is marked with at least as many tokens as
the weight of the inhibitor arc.
A transition~$\transition$ is \emph{active} if it is enabled and not inhibited.
The sets of enabled and inhibited transitions in marking $\marking$ are denoted
$\enabled(\marking)$ and $\inhibited(\marking)$, respectively.
A transition~$\transition$ is \emph{firable} if it has been
(continuously) enabled
for at least time~$\leftEP{\parIntervalStatic(\transition)}$, without
counting the time it has been inhibited.
A transition $\transition$ is \emph{newly enabled} by the firing of
transition $\transition_f$
in $\marking$
if it is enabled in the resulting marking
$\marking' = \marking - \relPre{\transition_f} +
\relPost{\transition_f}$,  and either $\transition$ is the fired
transition $\transition_f$ or $\transition$  was not enabled
in $\marking - \relPre{\transition_f}$:
\[\mbox{\newlyEnabled}(\transition,\marking,\transition_f) =
(\relPre{\transition} \leq
    \marking  - \relPre{\transition_f} + \relPost{\transition_f})
\land ((\transition=\transition_f)
        \lor \neg (\relPre{\transition} \leq \marking -\relPre{\transition_f})).\]
 %
$\newlyEnabled(\marking,\transition_f)$ denotes the  transitions newly enabled
by  firing  $\transition_f$ in $\marking$.

The semantics of an ITPN is defined as a transition system with states
$(\marking,\interval)$,  where
$\marking$ is a marking and
$\interval$  is a function mapping each transition enabled in
$\marking$ to
a time interval, and
two kinds of transitions:  \emph{time} transitions
where time elapses, and discrete transitions when a transition in the net is fired.

\begin{definition}[Semantics of an
  ITPN~\cite{paris-paper}] \label{def:pnet-semantics}
The dynamic behaviors of  an ITPN $\pi(\PN)$ are defined by the
transition system
$\mathcal{S_{\pi(\PN)}} = (\mathcal{A},a_0,\rightarrow)$, where:
$\mathcal{A}=\grandn^P\times [\Transition\rightarrow\Interval(\grandq)]$,
 $a_0=(\marking_0,\parIntervalStatic)$ and
    $(\marking,\interval)
    \rightarrow (\marking',\interval')$ if there exist $\delta\in
    \grandrplus$,  $\transition\in T$, and state
    $(\marking'',\interval'')$ such that
    $(\marking,\interval)\fleche{\delta} (\marking'',\interval'')$ and
    $(\marking'',\interval'') \fleche{\transition}
            (\marking',\interval')$, for the following  relations:
    \begin{itemize}
        \item the \emph{time transition relation}, defined
        $\forall\delta\in\grandrplus$ by:\\
        $(\marking,\interval)\fleche{\delta}
            (\marking,\interval')$ iff $\forall \transition \in \Transition$: \\
            $\left\{\begin{array}{l}
                \interval'(t)=\left\{\begin{array}{l}
                    \interval(\transition)\mbox{ if }
                        \transition \in \enabled(\marking) \mbox{ and }
                        \transition\in\inhibited(\marking)\\
                    \leftEP{\interval'(\transition)}=\max(0,\leftEP{\interval(\transition)}
                                       - \delta),
                    \mbox{ and } \rightEP{\interval'(\transition)} =
                                       \rightEP{\interval(\transition)}
                                       - \delta
                    \mbox{ otherwise}
                   \end{array}\right.\\
                \marking \geq\relPre(\transition) \implies
                    \rightEP{\interval'(\transition)}\geq 0
            \end{array}\right.$
        \item the \emph{discrete transition relation}, defined
        $\forall\transition_f\in\Transition$ by:
        $(\marking,\interval)\fleche{\transition_f}
            (\marking',\interval')$ iff\\
            $\left\{\begin{array}{l}
                \transition_f\in\enabled(\marking)\land
                    \transition_f\not\in\inhibited(\marking) \land
                    \marking'=\marking-\relPre{\transition_f}
                        +\relPost{\transition_f} \land
                    \leftEP{\interval(\transition_f)}=0\\
                    \forall\transition\in\Transition,
                        \interval'(\transition)=
                            \left\{\begin{array}{l}
                                \parIntervalStatic(\transition)
                                    \mbox{ if }
                                    \newlyEnabled(\transition,\marking,
                                        \transition_f) \\
                                \interval(\transition) \mbox{ otherwise}
                            \end{array}\right.
            \end{array}\right.$
    \end{itemize}
\end{definition}

\begin{example}
Consider  the ITPN in \Cref{fig:example:PITPN:valuated}.
A possible concrete firing sequence is:
\[
 \hspace*{-2mm}{   \begin{array}{lll}
(\{A,\!B\},([5,6],[3,4],[1,2]))\! \fleche{2}(\{A,\!B\},([3,4],[3,4],[0,0]))
\! \fleche{\transition_3}(\{A,\!E\},([3,4],[3,4],[0,0]))
\! \fleche{3}\\
(\{A,\!E\},([0,1],[3,4],[0,0]))
\! \fleche{\transition_1}(\{C,\!E\},([0,1],[3,4],[0,0]))
\end{array} }
\]
\end{example}

The \emph{symbolic} semantics of PITPNs is given
in~\cite{EAGPLP13} as a transition system $(\grandn^P \times \setP,
(\marking_0, K_0)$, $\Fleche{})$
on \emph{state classes}, \ie{}   pairs $\class = (\marking,\constraint)$
consisting of a marking $\marking$ and  a constraint $\constraint$
over $\Param$.
 The firing of a transition leads to
a new marking as in the concrete semantics, and also captures the new constraints
induced by the time that has passed for the transition to fire.

\begin{example} \label{example:k0}
Since the initial constraint $K_0$ of a PITPN must
ensure that each firing interval is non-empty, the initial
constraint $K_0$  of the PITPN in Fig. \ref{fig:example:PITPN:valuated}a is
$\param_1^- \leq \param_1^+
\; \land \; \param_2^- \leq \param_2^+ \; \land\; \param_3^- \leq
\param_3^+$.  Therefore,  the
 initial state class of this PITPN is $(\{A, B\}, \param_1^- \leq \param_1^+
\; \land \; \param_2^- \leq \param_2^+ \; \land\; \param_3^- \leq \param_3^+)$.
When firing transition $\transition_1$, the time spent for
$\transition_1$ to be firable is such that the other transitions
($\transition_3$ in this case) do
not miss their deadlines. So we obtain an additional inequality
$\param_1^-\leq\param_3^+$ and the new state class, obtained after firing
$\transition_1$ is $(\{\textcolor{red}{C}, B\},
\param_1^-\leq\param_1^+ \: \land \:
\param_2^-\leq\param_2^+ \: \land \:  \param_3^-\leq\param_3^+
\: \land \: \textcolor{red}{\param_1^-\leq\param_3^+})$. See~\cite{EAGPLP13}
for details.
\end{example}

A new semantics for PITPNs, where
a firing time (in the appropriate interval) is picked as soon as a
transition becomes enabled,  was recently introduced
in~\cite{DBLP:conf/apn/LeclercqLR23}.
This allows for a  simpler definition of the concrete semantics.
However, the work in~\cite{DBLP:conf/apn/LeclercqLR23}
targets a controller synthesis problem and therefore imposes additional
constraints on the model,
does not consider inhibitor arcs, and assumes nets to be safe.
We therefore consider  in this paper the much more general  definition of PITPNs
in \cite{paris-paper}.

\subsection{Rewriting with SMT and Maude}
\label{sec:rew-smt}

\paragraph{Rewrite theories.}


A \emph{rewrite theory} \cite{Mes92} is
a tuple $\mathcal{R} = (\Sigma, E, L, R)$
such that
\begin{itemize}
\itemsep=0.8pt
    \item $\Sigma$ is a signature that declares sorts, subsorts, and function symbols;
    \item $E$ is a set of 
    equations of the form $t=t' \mbox{ \textbf{if} } \psi$,
    where
    $t$ and $t'$ are terms of the same sort,
    and $\psi$ is a conjunction of equations;

    \item $L$ is a set of \emph{labels};
    and

    \item $R$ is a set of 
    rewrite rules
    of the form
    $l : q \longrightarrow r \mbox{ \textbf{if} } \psi$,
    where $l \in L$ is a label,
    $q$ and $r$ are terms of the same sort,
    and
    $\psi$ is a conjunction of equations.\footnote{The condition may
      also include rewrites, but we do not use this extra generality in
      this paper.}
\end{itemize}

$T_{\Sigma, s}$ denotes the set of ground (\ie{} not containing variables)
terms of sort $s$,
 and $T_{\Sigma}(X)_s$ denotes the set of 
 terms of sort $s$
 over a set of  sorted
 variables $X$. $T_{\Sigma}(X)$ and
 $T_{\Sigma}$ denote all terms and ground terms, respectively.
 A substitution $\sigma : X \rightarrow T_{\Sigma}(X)$
 maps each variable to a term of the same sort,
 and
  $t \sigma$
  denotes
 the term obtained
by simultaneously replacing each variable $x$ in a term $t$ with $\sigma(x)$.
The domain of a substitution
$\sigma$ is $\mathit{dom}(\sigma) = \{x \in X \mid \sigma(x) \neq x\}$,
assumed to be finite. 

\medskip
A \emph{one-step rewrite} $t \longrightarrow_{\mathcal{R}} t'$ holds
if there are
a rule $l : q \longrightarrow r \mbox{ \textbf{if} } \psi$,
a subterm $u$ of $t$,
and a substitution $\sigma$ such that
$u = q\sigma$ (modulo equations),
$t'$ is the term obtained from $t$
by replacing 
$u$ with $r\sigma$,
and $v\sigma = v'\sigma$ holds
for each 
$v = v'$ in $\psi$.
We denote by
$\longrightarrow_{\mathcal{R}}^\ast$
the reflexive-transitive closure of $\longrightarrow_{\mathcal{R}}$.

A rewrite theory $\mathcal{R}$
is called \emph{topmost}
iff
there is a sort $\mathit{State}$
at the top of one of the connected components of the subsort partial order
such that
for each rule $l : q \longrightarrow r \mbox{ \textbf{if} } \psi$,
both $q$ and $r$ have the top sort $\mathit{State}$,
and
no operator has sort $\mathit{State}$
or any of its subsorts as an argument sort.

\paragraph{Rewriting with SMT \cite{rocha-rewsmtjlamp-2017}.}
For a signature $\Sigma$ and
a set of equations $E$,
a \emph{built-in theory} $\mathcal{E}_0$
is a first-order theory with a signature $\Sigma_0 \subseteq \Sigma$,
where
(1) each sort $s$ in $\Sigma_0$ is minimal in $\Sigma$;
(2) $s \notin \Sigma_0$ for each operator $f:s_1\times \cdots \times s_n \rightarrow s$
    in $\Sigma \setminus \Sigma_0$; and
(3) $f$ has no other subsort-overloaded typing in $\Sigma_0$.
The satisfiability of a constraint in $\mathcal{E}_0$
is assumed to be decidable
using the SMT theory $\mathcal{T}_{\mathcal{E}_0}$
which is
consistent with  $(\Sigma, E)$, \ie{}
for $\Sigma_0$-terms $t_1$ and $t_2$,
if $t_1= t_2$ modulo $E$, then $\mathcal{T}_{\mathcal{E}_0} \models t_1 = t_2$.

\medskip
A \emph{constrained term}
is a pair $\phi \parallel t$ of
a constraint $\phi$ in $\mathcal{E}_0$ and  a term $t$ in $T_{\Sigma}(X_0)$
over variables $X_0 \subseteq X$ of the built-in sorts in
$\mathcal{E}_0$ \cite{rocha-rewsmtjlamp-2017,bae2019symbolic}.
A constrained term
 $\phi \parallel t$
\emph{symbolically} represents
all instances of the pattern
$t$
such that $\phi$ holds:
    $\llbracket \phi \parallel t \rrbracket
=
\{t' \mid t' = t\sigma \ \mbox{(modulo $E$) and}\  \mathcal{T}_{\mathcal{E}_0}
\models \phi\sigma \ \mbox{for ground}\ \sigma : X_0 \to T_{\Sigma_0}
\}.
$

An \emph{abstraction of built-ins}
for a $\Sigma$-term $t \in T_{\Sigma}(X)$
is a pair $(t^\circ, \sigma^\circ)$
of
a term $t^\circ \in T_{\Sigma \setminus \Sigma_0}(X)$
and
a substitution $\sigma^\circ : X_0 \to T_{\Sigma_0}(X_0)$
such that
$t = t^\circ \sigma^\circ$
and  $t^\circ$  contains no duplicate variables in $X_0$.
Any non-variable built-in subterms of $t$ are
replaced by distinct built-in variables in $t^\circ$.
$\Psi_{\sigma^\circ} = \bigwedge_{x \in \mathit{dom}(\sigma^\circ)} x = x \sigma^\circ $.
Let $\phi \parallel t$ be a constrained term
and $(t^\circ, \sigma^\circ)$ an
abstraction of built-ins for $t$.
If $\mathit{dom}(\sigma^\circ) \cap \ovars{\phi \parallel t} = \emptyset$,
then
$\llbracket \phi \parallel t \rrbracket
=
\llbracket \phi \wedge \Psi_{\sigma^\circ} \parallel t^\circ \rrbracket$
\cite{rocha-rewsmtjlamp-2017}.

Let $\mathcal{R}$ be a topmost theory
such that for each rule $l : q \longrightarrow r \mbox{ \textbf{if} }
\psi$, extra variables not occurring in the left-hand side $q$
are in $X_0$,
and
$\psi$ is a constraint in a built-in theory $\mathcal{E}_0$.
A \emph{one-step symbolic rewrite} $\phi \parallel t
\rightsquigarrow_{\mathcal{R}} \phi' \parallel t'$ holds
iff there exist
a rule $l : q \longrightarrow r \mbox{ \textbf{if} } \psi$
and
a substitution $\sigma : X \to T_{\Sigma}(X_0)$
such that
(1) $t = q\sigma$
and
$t' = r\sigma$
(modulo equations),
(2) $\mathcal{T}_{\mathcal{E}_0} \models (\phi \wedge  \psi
      \sigma) \Leftrightarrow \phi'$, and
(3) $\phi'$ is $\mathcal{T}_{\mathcal{E}_0}$-satisfiable.
We denote by
$\rightsquigarrow_{\mathcal{R}}^\ast$
the reflexive-transitive closure of $\rightsquigarrow_{\mathcal{R}}$.

\medskip
A \emph{symbolic rewrite}
on constrained terms
symbolically represents
a (possibly infinite) set of system transitions.
If $\phi_t \parallel t \rightsquigarrow^\ast \phi_u \parallel u$
is a symbolic rewrite,
then there exists a ``concrete'' rewrite  $t' \longrightarrow^\ast u'$
 with $t' \in \llbracket \phi_t \parallel t \rrbracket$
 and $u' \in \llbracket \phi_u \parallel u \rrbracket$.
Conversely,
for any concrete
rewrite $t' \longrightarrow^\ast u'$ with
$t' \in \llbracket \phi_t \parallel t \rrbracket$,
there exists
a symbolic rewrite $\phi_t \parallel t \rightsquigarrow^\ast \phi_u \parallel u$
with $u' \in \llbracket \phi_u \parallel u \rrbracket$.\vspace*{-1mm}

\paragraph{Maude.}
Maude~\cite{maude-book} is a language and tool
supporting the specification and analysis of  rewrite theories.
We summarize its syntax below:

\begin{maude}
pr R .                   --- Importing a theory R
sorts S ... Sk .         --- Declaration of sorts S1,..., Sk
subsort S1 < S2 .        --- Subsort relation
vars X1 ... Xm : S .     --- Logical variables of sort S
op f : S1 ... Sn -> S .  --- Operator S1 x ... x Sn -> S
op c : -> T .            --- Constant c of sort T
eq t = t' .              --- Equation
ceq t = t' if c .        --- Conditional equation
crl [l] : q => r if c .  --- Conditional rewrite rule
\end{maude}

\noindent Maude provides a large palette of analysis methods,  including computing
the
normal form of a term $t$ (command \lstinline[mathescape]{red $t$}),
simulation by rewriting (\lstinline[mathescape]{rew $t$}), and rewriting according
to a given rewrite
strategy (\lstinline[mathescape]{srew $t$ using $str$}).
Basic rewrite strategies include:
$r\mathtt{[}\sigma\mathtt{]}$ (apply  rule $r$ once with the optional ground
substitution $\sigma$), \code{all} (apply any of the rules once),
and \code{match $P$ s.t.\ $C$}, which  checks
whether the current term matches the pattern $P$ subject to the constraint $C$.
Compound strategies can be defined using concatenation
($\alpha\,;\,\beta$), disjunction ($\alpha\, |\, \beta$), iteration ($\alpha \,\mathtt{*}$),
$\alpha \code{ or-else } \beta$ (execute $\beta$ if $\alpha$ fails),
normalization  $\alpha\,\mathtt{!}$ (execute $\alpha$  until it cannot be further
applied), etc.

\medskip
Maude also offers explicit-state
reachability analysis from a ground
term $t$ (\lstinline[mathescape]{search [$n$,$d$] $t$ =>* $t'$}
\lstinline[mathescape]{such that $\Phi$}, where the optional
parameters $n$ and $d$ denote the maximal number of solutions to
search for and the maximal depth of the search tree, respectively)
and  model checking a linear temporal logic (LTL) formula $F$
(\lstinline[mathescape]{red modelCheck($t$,$\,F$)}). Atomic propositions
in $F$ are user-defined terms of sort \code{Prop}, and the function
\lstinline{op _|=_ : State Prop -> Bool} specifies which states satisfy a given
proposition. LTL formulas are then built from state formulas, Boolean connectives
and the temporal logic operators \texttt{[]} (``always''), \texttt{<>} (``eventually'')
and  \texttt{U} (``until'').

\medskip
For symbolic reachability analysis, the command

\begin{maude}
smt-search [$n$,$\,d$] $t$ =>* $t'$ such that $\Phi$  --- n and m are optional
\end{maude}
symbolically searches for $n$ states,
reachable from $t \in T_{\Sigma}(X_0)$ within $d$ rewrite steps,
that match the pattern $t' \in T_{\Sigma}(X)$ and satisfy
the constraint $\Phi$ in $\mathcal{E}_0$.
More precisely,
it searches for
a constrained term $\phi_u \parallel u$
such that
$\mathit{true} \parallel t \rightsquigarrow^\ast \phi_u \parallel u$
and  for some $\sigma : X \to T_{\Sigma}(X)$,
$u = t'\sigma$ (modulo equations) and
 $\phi_u \wedge \Phi\sigma$ is
$\mathcal{T}_{\mathcal{E}_0}$-satisfiable.

\medskip
Maude provides   built-in  sorts \code{Boolean}, \code{Integer}, and \code{Real} for
the SMT theories of Booleans, integers, and reals. Rational constants of sort \code{Real} are written
\code{$n$/$m$} (\eg{} \code{0/1}).
Maude-SE~\cite{maude-se} extends Maude with
additional functionality
for rewriting modulo SMT,
including
witness generation for \lstinline{smt-search}.
It uses two theory transformations
to implement
symbolic rewriting~\cite{rocha-rewsmtjlamp-2017} as ``standard'' rewriting, thus opening the
possibility of
using standard Maude's commands as \lstinline{search}
on constrained terms.
In essence,
a  rewrite
rule $l : q \longrightarrow r \mbox{ \textbf{if} } \psi$
is transformed
into
a constrained-term rule
\begin{align*}
l:
\mathtt{PHI} \parallel q^\circ
\longrightarrow
(\mathtt{PHI} \mathbin{and} \psi \mathbin{and} \Psi_{\sigma^\circ} ) \parallel r
\mbox{ \textbf{if} }
&
\mathtt{smtCheck}(\mathtt{PHI} \mathbin{and} \psi \mathbin{and} \Psi_{\sigma^\circ} )
\end{align*}

\noindent
where $\mathtt{PHI}$ is a \code{Boolean} variable,
$(q^\circ, \sigma^\circ)$ is an abstraction of built-ins for $q$,
and
\code{smtCheck} invokes the underlying  SMT solver
to check the satisfiability of an SMT condition.
This rule is executable
if the extra SMT variables in $(\ovars{r} \cup
\ovars{\psi} \cup \ovars{\Psi_{\sigma^\circ}}) \setminus \ovars{q^\circ}$
are considered constants.

\paragraph{Meta-programming.} Maude supports \emph{meta-programming}, where a
Maude module $M$ (resp., a term $t$) can be (meta-)represented as a Maude
\emph{term} $\overline{M}$ of sort \code{Module} (resp.\  as a Maude term
$\overline{t}$ of sort \code{Term}) in Maude's \code{META-LEVEL} module.
Sophisticated analysis commands and model/module transformations can then be
easily defined as ordinary Maude functions on such (meta-)terms. For this
purpose, Maude provides built-in functions such as \code{metaReduce},
\code{metaRewrite},  and \code{metaSearch}, which are the ``meta-level''
functions corresponding to equational reduction to normal form, rewriting, and
search, respectively.


\section{A rewriting logic semantics for ITPNs }\label{sec:tranformation}
\label{sec:concrete}

This section presents a rewriting logic semantics for (non-parametric)
ITPNs, using
 a (non-executable) rewrite theory
$\rtheorySem$. We provide a
bisimulation  relating the concrete
semantics of a net $\PN$ and an induced rewrite relation in $\rtheorySem$. Furthermore,   
 we discuss  variants of $\rtheorySem$
to avoid  consecutive tick steps and  to enable time-bounded
 analysis.

\subsection{Formalizing  ITPNs in Maude: The theory $\rtheorySem$}
\label{subsec:eq-theory}

We fix $\PN$ to be the ITPN $\tuple{\Place, \Transition, \emptyset, \relPre{(.)},
\relPost{(.)}, \relInhib{(.)}, \markingInit, \parIntervalStatic,
true}$, and show how ITPNs and markings of such nets can be
represented as Maude terms.

\medskip
We first  define sorts for representing transition
 labels, places, and time values  in Maude.
 The usual approach is to represent each transition $t_i$ and
 each place $p_j$ as a constant of sort \texttt{Label} and
 \texttt{Place}, respectively (\eg{} \texttt{ops \(p_1\) \(p_2\)
   ... \(p_m\) : -> Place [ctor]}).  To avoid even this simple
 parameterization and just use a single rewrite theory $\mathcal{R}_0$
 to define the semantics of  all ITPNs, we assume
 that  places and transition (labels) can be represented as
 strings. Formally, we assume that there is an injective naming
 function $\eta : P \cup T \rightarrow \mathtt{String}$.
 To improve readability, we usually do not mention $\eta$  explicitly.

\begin{maude}
  protecting STRING .   protecting RAT .
  sorts Label Place .                 --- identifiers for transitions and places
  subsorts String < Label Place .     --- we use strings for simplicity

  sorts Time TimeInf .                --- time values
  subsorts Zero PosRat < Time < TimeInf .
  op inf : -> TimeInf [ctor] .
  vars T T' T1 T2 : Time     .
  eq T <= inf = true .
\end{maude}\vspace{0.9mm}

The sort \texttt{TimeInf}
 adds an ``infinity'' value \texttt{inf} to the sort \texttt{Time} of
 time values, which are the  non-negative rational numbers (\code{PosRat}).

\medskip
The ``standard'' way of formalizing Petri nets in rewriting logic (see,
\eg{} \cite{Mes92,petri-nets-in-maude}) represents, \eg{}  a marking
with two tokens in place $p$
and three tokens in place $q$  as the Maude term
$p\; p\; q\; q\; q$. This is crucial to support \emph{concurrent}
firings of transitions in a net. However, since the semantics of PITPNs is
an \emph{interleaving} semantics, and   to support
rewriting-with-SMT-based analysis from
\emph{parametric} initial markings (Example~\ref{ex:pmarking}),
we instead represent markings as
maps from places to the number of tokens in that place,
so that the above
marking is  represented by the Maude
term  \texttt{\(\nameit(p)\)\,|->\,2 ; \(\nameit(q)\)\,|->\,3}.

\medskip
The following declarations define  
the
sort  \code{Marking} to consist of  \code{;}-separated sets of pairs \linebreak
\code{\(\nameit(p)\)\,\,|->\,\,\(n\)}.
Time intervals are represented as  terms 
\code{[\(\mathit{lower}\,\):\(\,\mathit{upper}\)]} where the upper
bound $\mathit{upper}$,
of sort \code{TimeInf}, also can be the infinity value \code{inf}.
The  Maude term
\code{  \(\;\nameit(t)\)\,: \(pre\) --> \(post\) inhibit \(inhibit\) in
  \(interval\) } represents a transition $t\in T$,
where 
\code{\(pre\)}, \code{\(post\)},  and  \code{\(inhibit\)} are markings
representing, respectively,
  $\relPre{(t)}, \relPost{(t)}, \relInhib{(t)}$; and
  \code{\(interval\)} represents the interval $J(t)$.  The
  `\texttt{inhibit}' part can be omitted if it is \texttt{empty}. A \code{Net}
  is represented as a \texttt{;}-separated set of such transitions.

\begin{maude}
vars M M1 M2 PRE POST INHIBIT INTERM-M : Marking .
vars L L' : Label .
vars I INTERVAL : Interval .
var NET : Net .

sort Marking .                --- Markings
op empty : -> Marking [ctor] .
op _|->_ : Place Nat -> Marking [ctor] .
op _;_ : Marking Marking -> Marking [ctor assoc comm id: empty] .

sort Interval .               --- Time intervals (the upper bound can be infinite)
op `[_:_`] : Time TimeInf -> Interval [ctor] .

sorts Net Transition .        --- Transitions and nets
subsort Transition < Net .
op emptyNet : -> Net [ctor] .
op _;_ : Net Net -> Net [ctor assoc comm id: emptyNet] .

op _`:_-->_inhibit_in_ : Label Marking Marking Marking Interval -> Transition [ctor] .

op _`:_-->_in_ : Label Marking Marking Interval -> Transition .
eq L : M1 --> M2 in I  =  L : M1 --> M2 inhibit empty in I .
\end{maude}

  \begin{example}
Assuming the obvious naming function $\nameit$ mapping $A$ to
\texttt{"A"}, and so on, the net in Fig.~\ref{fig:example:PITPN:valuated}b is
represented as the following term of sort \code{Net}:

\begin{maude}
"t1" : ("A" |-> 1) --> ("C" |-> 1) in [5 : 6] ;
"t2" : ("B" |-> 1) --> ("D" |-> 1) inhibit ("A" |-> 1) in [3 : 4] ;
"t3" : ("B" |-> 1) --> ("E" |-> 1) in [1 : 2]$.$
\end{maude}
\end{example}

We  define some useful operations on markings, such as
\texttt{_+_} and \texttt{_-_}:

\begin{maude}
vars N N1 N2 : Nat .    var P : Place .
op _+_ : Marking Marking -> Marking .
eq ((P |-> N1) ; M1) + ((P |-> N2) ; M2) = (P |-> N1 + N2) ; (M1 + M2) .
eq M1 + empty = M1 .
\end{maude}

\noindent (This definition  assumes that
  each place
in \code{\(\marking2\)} appears  once in \code{\(\marking1\)} and
\code{\(\marking1\:\)+\(\:\marking2\)}.) The
function \code{_-_} on markings is defined similarly.
%
The following functions  compare markings and  check whether a transition is
active in a  marking:

\begin{maude}
op _<=_ : Marking Marking -> Bool .       --- Comparing markings
eq ((P |-> N1) ; M1) <= ((P |-> N2) ; M2)  =  N1 <= N2 and (M1 <= M2) .
eq empty <= M2 = true .
ceq M1 <= empty = false if M1 =/= empty .

op active : Marking Transition -> Bool .  --- Active transition
eq active(M, L : PRE --> POST inhibit INHIBIT in INTERVAL) =
      (PRE <= M) and not inhibited(M, INHIBIT) .

op inhibited : Marking Marking -> Bool .  --- Inhibited transition
eq inhibited(M, empty) = false .
eq inhibited((P |-> N2) ; M, (P |-> N) ; INHIBIT) =
      ((N > 0) and (N2 >= N)) or inhibited(M, INHIBIT) .
\end{maude} \label{sec:states}\vspace*{-2mm}

\paragraph{Dynamics.}  We  define the dynamics of ITPNs as a Maude
``interpreter'' for such nets.
The concrete ITPN semantics in~\cite{paris-paper} dynamically adjusts
the ``time intervals'' of non-inhibited transitions when time
elapses.
Unfortunately, the definitions in~\cite{paris-paper} seem slightly
contradictory (even with non-empty firing intervals): On the one hand,
time interval end-points should be
non-negative, and only enabled transitions have intervals in the
states; on the other hand, the definition of time and discrete
transitions in~\cite{paris-paper} mentions $\forall t\in T, I'(t) =
...$ and $\marking \geq\relPre(\transition) \implies
                  \! \rightEP{\!\interval'(\transition)}\geq 0$, which
                    seems superfluous if all end-points are
                    non-negative. Taking the definition of time and
                    transition steps in~\cite{paris-paper} (our
                    Definition~\ref{def:pnet-semantics}) leads us to
                    time intervals where
the right end-points of disabled transitions
could have \emph{negative}
values.  This has some
disadvantages: (i) ``time values'' can be negative numbers; (ii) we
have counterintuitive  ``intervals'' $[0, -r]$ where the right end-point is smaller
than the left end-point; (iii) the reachable ``state spaces'' (in suitable
discretizations) could  be infinite when these negative
values could be unbounded.

\medskip
To have a simple and
well-defined semantics,   we  use
``clocks'' instead of ``decreasing intervals'';  a clock denotes  how
long the corresponding
transition has been enabled (but not inhibited). Furthermore, to
reduce the state space, the clocks of disabled transitions  are always
zero.   The resulting semantics is equivalent to the 
(most natural interpretation of the) one in~\cite{paris-paper} in a way made
precise in Theorem~\ref{thm:ground}.

The sort \code{ClockValues} denotes sets of \texttt{;}-separated  terms
\code{\(\nameit(t)\) -> \(\tau\)}, where \(t\) is the (label of the) transition and
\code{\(\tau\)}  represents
the current value of  $t$'s ``clock.''

\begin{maude}
sort ClockValues .   --- Values for clocks
op empty : -> ClockValues [ctor] .
op _->_  : Label Time -> ClockValues [ctor] .
op _;_   : ClockValues ClockValues -> ClockValues [ctor assoc comm id: empty] .

var CLOCKS : ClockValues .
\end{maude}

The states in $\rtheorySem$ are
terms $m$\,\texttt{:}\,$\mathit{clocks}$\,\texttt{:}\,$\mathit{net}$
of sort \code{State}, where $m$ represents the current marking,
$\mathit{clocks}$ the current values of the transition clocks, and
$\mathit{net}$ the representation of the Petri net:

\begin{maude}
sort State .
op _:_:_ : Marking ClockValues Net -> State [ctor] .
\end{maude}

 The following rewrite rule models the application of a transition
\code{L} in the net
\code{(\pTransitionMI{L}{PRE}{POST}{INHIBIT}{INTERVAL}) ; NET}.
 Since \texttt{_;_}
is declared to be associative and commutative, \emph{any} transition \code{L}
in the net
can be applied using this rewrite rule:

\begin{maude}
crl [applyTransition] :
     M  :
     (L -> T) ; CLOCKS  :
     ($\pTransitionMI{L}{PRE}{POST}{INHIBIT}{INTERVAL}$) ; NET
    =>
     $\highlight{(M - PRE) + POST}$ :
     $\highlight{L -> 0}$ ; updateClocks(CLOCKS, M - PRE, NET) :
     ($\pTransitionMI{L}{PRE}{POST}{INHIBIT}{INTERVAL}$) ; NET
 if $\highlight{active}$(M, $\pTransitionMI{L}{PRE}{POST}{INHIBIT}{INTERVAL}$)
    and ($\highlight{T in INTERVAL}$) .

op _in_ : Time Interval -> Bool .
eq T in [T1 : T2] = (T1 <= T) and (T <= T2) .
eq T in [T1 : inf] = T1 <= T .
\end{maude}

\noindent The  transition \code{L} is active (enabled and not inhibited)
in the  marking \code{M}
 and its clock value \code{T} is in the \code{INTERVAL}.
After performing the transition, the  marking is
\code{(M\,-\,PRE)\,+\,POST},
the clock of \code{L} is reset\footnote{Since in our semantics clocks
  of disabled transitions should be zero, we can safely set the clock of \texttt{L}
  to \texttt{0} in this rule.}
and the other clocks
    are updated using the following function:

\begin{maude}
eq updateClocks((L' -> T') ; CLOCKS, INTERM-M,
                ($\pTransitionMI{L'}{PRE}{POST}{INHIBIT}{INTERVAL}$) ; NET)
 = if PRE <= INTERM-M then (L' -> T') else (L' -> 0) fi ;
   updateClocks(CLOCKS, INTERM-M, NET) .
eq updateClocks(empty,  INTERM-M, NET) = empty .
 \end{maude}

The second rewrite rule in $\rtheorySem$ specifies how time
advances. Time can advance by \emph{any} value \code{T}, as long as
time
does not advance beyond the time when an active transition must be
taken.  The clocks are updated according to the elapsed time \code{T},
except for those transitions that are disabled or inhibited:

\begin{maude}
crl [tick] : M$\;\,$:$\;\,$CLOCKS$\;\,$:$\;\,$NET  =>  M$\;\,$:$\;\,$increaseClocks(M,$\:$CLOCKS,$\:$NET,$\:\highlight{T}$)$\;\,$:$\;\,$NET
    if  $\highlight{T}$ <= mte(M, CLOCKS, NET) [nonexec] .
\end{maude}

\noindent This rule is not executable (\code{[nonexec]}), since the
variable \code{T}, which denotes how much time
advances, only  occurs
in the right-hand side of the rule. \code{T} is
therefore \emph{not} assigned any value by the substitution matching the rule
with the state being rewritten.
This time advance  \code{T}
must be less or equal to  the smallest remaining time until the upper
bound of an  active
transition in the marking \code{M} is reached:\footnote{To increase
  readability, we sometimes replace parts of Maude code or output from
  Maude analyses by `\texttt{...}'.}

\begin{maude}
op mte : Marking ClockValues Net -> TimeInf .
eq mte(M, (L$\;$->$\;$T)$\;\,$;$\;\,$CLOCKS, (L$\;\,$:$\;\,$PRE$\;\,$-->$\;\,$POST ... in$\;\,$[T1$\,$:$\,\highlight{inf}$])$\;\,$;$\;\,$NET)
 = mte(M, CLOCKS, NET) .
eq mte(M, (L$\;$->$\;$T)$\;\,$;$\;\,$CLOCKS, (L$\;\,$:$\;\,$PRE --> ... in$\;\,$[T1$\,$:$\,\highlight{T2}$])$\;\,$;$\;\,$NET)
 = if active(M, L : ...) then min(T2 - T, mte(M, CLOCKS, NET))
   else mte(M, CLOCKS, NET) fi .
eq mte(M, empty, NET) = inf .
\end{maude}

The function \code{increaseClocks} increases the transition clocks of
the active transitions by
the elapsed time:

\begin{maude}
op increaseClocks : Marking ClockValues Net Time -> ClockValues .
eq increaseClocks(M, $\highlight{(L -> T1)}$ ; CLOCKS, (L : PRE --> ...) ; NET, $\highlight{T}$)
 = if active(M, L : PRE --> ...)
   then $\highlight{(L -> T1 + T)}$ else $\highlight{(L -> T1)}$ fi ; increaseClocks(M,$\,$CLOCKS,$\,$NET,$\,$T)$\,$.
eq increaseClocks(M, empty, NET, T) = empty .
\end{maude}

The following function $\encBase{\_}$ formalizes how markings and nets are
represented as terms, of respective sorts \texttt{Marking} and
\texttt{Net}, in rewriting logic.\footnote{$\encBase{\_}$ is parametrized
  by the naming function $\nameit$; however, we do not show this
  parameter explicitly.}

\begin{definition}
Let $\PN=\tuple{\Place, \Transition, \emptyset, \relPre{(.)},
\relPost{(.)}, \relInhib{(.)}, \markingInit, \parIntervalStatic,
\mathit{true}}$ be an ITPN.   Then $\encBase{\_} : \mathbb{N}^P \rightarrow
\mathcal{T}_{\mathcal{R}_0,\mathtt{Marking}}$ is defined by
  $\encBase{\{p_1\mapsto n_1, \ldots, p_m\mapsto n_m\}} = \nameit(p_1)
  \texttt{\,|->\;} n_1 \;\texttt{;} \,\ldots\, \texttt{;}\; \nameit(p_m)\;
  \texttt{|->} \;n_m$, where  we can omit  entries $\nameit(p_j)
  \texttt{\,|->\;} 0$. 
   The Maude representation $\encBase{\PN}$ of the net $\PN$ is
  the term $\encBase{t_1}\, \texttt{;}\, \cdots\,
\texttt{;}\, \encBase{t_n} $
of sort \code{Net}, where, for each $t_i \in T$, $\encBase{t_i}$ is
\newline \code{
$\nameit(t_i)\;$:$\;\encBase{\relPre{(t_i)}}$ -{}->
    $\encBase{\relPost{(t_i)}}$ inhibit
    $\encBase{\relInhib{(t_i)}}$ in [$\leftEP{J(t_i)}\:$:$\:\rightEP{J(t_i)}$].
}
\end{definition}

\subsection{Correctness of the semantics} \label{sec:bisimulation}

In this section we show that our rewriting logic semantics
$\mathcal{R}_0$ correctly simulates any ITPN $\PN$. More concretely, we
provide a bisimulation result relating behaviors from $a_0 =
(\marking_0,\parIntervalStatic)$
in $\PN$ with behaviors in
$\mathcal{R}_0$ starting from the initial state
$\encBase{\marking_0}$\,\code{:}\,\code{initClocks($\encBase{\PN}$)}\,\code{:}\,$\encBase{\PN}$,
where  \texttt{initClocks(\(\mathit{net}\))} is the clock valuation that
assigns the value \texttt{0} to each transition (clock) $\eta(t)$ for
each transition (label) $\eta(t)$ in $\mathit{net}$.

\medskip
Since a transition in $\PN$ consists of a delay followed by
a discrete transition, we  define a corresponding rewrite
relation $\mapsto$ combining the \code{tick} and
\code{applyTransition} rules, and prove the bisimulation for this
relation.

\begin{definition}\label{def:two-steps}
    Let $t_1, t_2, t_3$ be terms of sort \code{State} in
    $\mathcal{R}_0$. We write $t_1
    \mapsto t_3$ if
    there exists a $t_2$ such that $t_1 \longrightarrow t_2$ is a one-step
    rewrite applying the \code{tick} rule in $\mathcal{R}_0$
    and $t_2 \longrightarrow t_3$
    is a one-step rewrite applying the \code{applyTransition} rule in
    $\mathcal{R}_0$.
    Furthermore, we write $t_1 \mapsto^* t_2$ to
    indicate that there exists a
    sequence of $\mapsto$ rewrites from $t_1$ to $t_2$.
\end{definition}

The following relation relates our clock-based states with the
changing-interval-based states; the correspondence is a
straightforward function, except for the case when the upper bound of
a transition is $\infty$:

\begin{definition}  \label{def:bisim-relation}
  Let $\PN = \tuple{\Place, \Transition, \emptyset, \relPre{(.)},
    \relPost{(.)}, \relInhib{(.)}, \markingInit, \parIntervalStatic,
    \mathit{true}}$
  be an ITPN and $\mathcal{S_{\PN}} = (\mathcal{A},a_0,\rightarrow)$ be its
  concrete semantics.
   Let
  $T_{\Sigma,\texttt{State}}$ denote the set of $E$-equivalence
  classes of ground terms of sort \texttt{State} in $\mathcal{R}_0$.
  We define a relation  $\approx \, \subseteq \mathcal{A} \times T_{\Sigma,
    \texttt{State}}$, relating states in the concrete semantics
  of $\PN$ to states (of sort \texttt{State}) in
  $\mathcal{R}_0$, where for all states $(\marking,\interval) \in \mathcal{A}$,
  $(\marking,\interval) \approx m\; \code{:}
  \;\mathit{clocks}\; \code{:} \; \mathit{net}$
 if and only if $m = \encBase{M}$ and $\mathit{net} = \encBase{\PN}$
 and  for
  each transition $t\in \Transition$,
  \begin{itemize}
    \item  the value of $\eta(t)$ in $\mathit{clocks}$ is \code{0} if
      $t$ is not enabled in $M$;
    \item otherwise:
      \begin{itemize}
      \item if $\rightEP{\parIntervalStatic(t)} \not = \infty$ then
the value of clock $\eta(t)$ in $\mathit{clocks}$ is  $\rightEP{\parIntervalStatic(t)}
          - \rightEP{\interval(t)}$;
          \item otherwise, if $\leftEP{\interval(t)} > 0$ then
            $\eta(t)$ has
            the value $\leftEP{\parIntervalStatic(t)}
          - \leftEP{\interval(t)}$ in $\mathit{clocks}$; otherwise, the  value
          of $\eta(t)$ in
          $\mathit{clocks}$ could be any value $\tau\geq
          \leftEP{\parIntervalStatic(t)}$.
        \end{itemize}
        \end{itemize}
\end{definition}

\begin{theorem}
\label{thm:ground}
Let $\PN = \tuple{\Place, \Transition, \emptyset, \relPre{(.)},
  \relPost{(.)}, \relInhib{(.)}, \markingInit, \parIntervalStatic,
  true}$ be an ITPN,
 and $\mathcal{R}_0 = (\Sigma,E,L,R)$. Then,
$\approx$ is a bisimulation between the transition systems
$\mathcal{S_{\PN}} = (\mathcal{A},a_{0},\rightarrow)$ and \\$\left(T_{\Sigma
    , \texttt{State}}, (\encBase{\marking_0}
  \;\code{:}\;\mathtt{initClocks}(\encBase{\PN})\;\code{:}\;\encBase{\PN}),
\mapsto\right)$.
\end{theorem}

\begin{proof}
  By definition $a_0 = (\marking_0,\parIntervalStatic) \approx (\encBase{\marking_0}
  \;\code{:}\;\mathtt{initClocks}(\encBase{\PN})\;\code{:}\;\encBase{\PN})$,
  since all clocks are \texttt{0} in $\mathtt{initClocks}(...)$, so
  that these clocks satisfy all the constraints in
  Definition~\ref{def:bisim-relation} since  $I=J$ in the initial state.
  Hence, condition (i) for $\approx$ being a bisimulation
  is satisfied. Condition (ii) follows from the  two lemmas below.
\end{proof}

\begin{lemma}
If $\left(\marking,\interval\right) \to
\left(\marking',\interval'\right)$
and $\left(\marking,\interval\right) \approx
(\encBase{\marking} \code{:} \mathit{clocks} \code{:} \encBase{\PN})$
then there is a $\mathit{clocks'}$ such that
$(\encBase{\marking}
\code{:}\mathit{clocks}\code{:}\encBase{\PN})
\mapsto (\encBase{\marking'}\code{:}\mathit{clocks'}
\code{:}\encBase{\PN})$
and $\left(\marking',\interval'\right) \approx
(\encBase{\marking'}\code{:} \mathit{clocks'} \code{:}\linebreak  \encBase{\PN})$.
\end{lemma}

\begin{proof}
Since $\left(\marking,\interval\right) \to \left(\marking',\interval'\right)$, we
have that there exists an intermediate pair $(\marking,\interval'') \in (\Transition\cup\grandrplus)$ such that
$\left(\marking,\interval\right) \fleche{\delta} \left(\marking,\interval''\right)$~and
$\left(\marking,\interval''\right) \fleche{\transition_f} \left(\marking',\interval'\right)$.

\noindent
For the first step ($\fleche{\delta}$), since $\left(\marking,\interval\right) \fleche{\delta} \left(\marking,\interval''\right)$,
there exists a $\delta$ such that $\forall t \in \Transition$, either $\interval''(\transition) =
\interval(\transition)$ or $\leftEP{\interval''(\transition)} = \max(0,\leftEP{\interval(\transition)} - \delta)$
and $\rightEP{\interval''(\transition)} = \rightEP{\interval(\transition)} - \delta$. In both cases
we have that $\forall \transition \in \Transition, \rightEP{\interval''(\transition)}\geq 0$.
Now, letting \code{T} $= \delta$, it must be the case that
\code{T <= mte($\encBase{\marking}$, $clocks$, $\encBase{\PN}$)}. This is
because \code{mte($\encBase{\marking}$, $clocks$, $\encBase{\PN}$)} is
defined to be equal to the minimum difference between
$\rightEP{\parIntervalStatic(\transition)}$ \linebreak and~the clock value of $t$ out of all
$\transition \in \Transition$. That is, it is the maximum time that can elapse before an
enabled transition reaches the right endpoint of its interval. In other words, an upper limit
for $\delta$. Hence, the \code{tick}-rule can be applied to
$(\encBase{\marking}~{ : }~\,{clocks}~{ : }\,~\encBase{\PN})$
with all enabled clocks~having their time ad\-vanced by $\delta$.

\noindent
For the second step ($\fleche{\transition_f}$), since $\left(\marking,\interval''\right)
\fleche{\transition_f} \left(\marking',\interval'\right)$, the transition $\transition_f$
is active and $\leftEP{\interval''(\transition_f)} = 0$. Since
$\leftEP{\interval(\transition_f)} = 0$,
the clock of transition $\transition_f$ must be in the interval $[\leftEP{\parIntervalStatic(\transition_f)},
\rightEP{\parIntervalStatic(\transition_f)}]$ by definition of $\mathit{clocks}$ for $(\marking,\interval'')$.
This is precisely the condition for applying the \code{applyTransition}-rule to the
resulting state of the previous \code{tick}-rule application.
\end{proof}

\begin{lemma}
If
$(\encBase{\marking}\;\code{:}\;\mathit{clocks}\;\code{:}\;\encBase{\PN})
\mapsto b$  and \newline  $\left(\marking, \interval\right) \approx
(\encBase{\marking}\;\code{:}\;\mathit{clocks}\;\code{:}\;\encBase{\PN})$,
then there exists a state $\left(\marking', \interval'\right) \in \mathcal{A}$
such that $\left(\marking, \interval\right) \to \left(\marking', \interval'\right)$ and
$\left(\marking', \interval'\right) \approx b$. \end{lemma}

\begin{proof}
Since $(\encBase{\marking} \code{:} \mathit{clocks} \code{:} \encBase{\PN}) \mapsto b$,
we have that there exists an intermediate state \\
$(\code{M\!  : \! CLOCKS \! : \! NET}) \in T_{\Sigma , \texttt{State}}$
such that $(\encBase{\marking} \code{:} \mathit{clocks} \code{:} \encBase{\PN})
\stackrel{\texttt{tick}}{\longrightarrow} (\code{M : CLOCKS : NET})$ and
$(\code{M : CLOCKS : NET}) \stackrel{\texttt{applyTransition}}{\longrightarrow} b$.

\smallskip\noindent
For the first step $\left(\stackrel{\texttt{tick}}{\longrightarrow}\right)$, since
$(\encBase{\marking}\;\code{:}\;\mathit{clocks}\;\code{:}\;\encBase{\PN})
\stackrel{\texttt{tick}}{\longrightarrow} (\code{M : CLOCKS : NET})$, there is a
\code{T <= mte($\encBase{\marking}$, $clocks$, $\encBase{\PN}$)}.
Now, as in the previous lemma, since the \code{mte} is an upper limit for $\delta$,
we have that there exists a time transition
$\left(\marking,\interval\right) \fleche{\delta} \left(\marking,\interval''\right)$
with $\delta$ equal to the \code{T} used in the above \code{tick}-rule application so that
$(\code{M : CLOCKS : NET}) = (\encBase{\marking}\;\code{:}\;\mathit{clocks}\;\code{:}\;\encBase{\PN})$.

\smallskip\noindent
For the second step $\left(\stackrel{\texttt{applyTransition}}{\longrightarrow}\right)$, since
$(\encBase{\marking}\;\code{:}\;\mathit{clocks}\;\code{:}\;\encBase{\PN})
\stackrel{\texttt{applyTransition}}{\longrightarrow} b$, there must be a transition $\transition_f$ which
is active and whose clock is in the interval $[\leftEP{\parIntervalStatic(\transition_f)},
\rightEP{\parIntervalStatic(\transition_f)}]$. By definition of $\mathit{clocks}$ on $(\marking,\interval'')$,
This is precisely the condition for the discrete step $\transition_f$ to $(\marking,\interval'')$.
Hence, $\left(\marking,\interval''\right) \fleche{\transition_f} \left(\marking',\interval'\right)$ and
$(\encBase{\marking}\;\code{:}\;\mathit{clocks}\;\code{:}\;\encBase{\PN})
\stackrel{\texttt{applyTransition}}{\longrightarrow}
(\encBase{\marking'}\code{:}\mathit{clocks}\;\code{:}\encBase{\PN})$.
\end{proof}\vspace*{-4mm}

\subsection{Some variations of $\rtheorySem$}\label{sec:optim}

This section introduces the theories $\rtheorySemN{1}$ and
$\rtheorySemN{2}$ as  two variations of $\rtheorySem$. 
   $\rtheorySemN{1}$ avoids  consecutive applications of
  the \texttt{tick} rule. This  is  useful for \emph{symbolic}
    analysis,  since in concrete executions of $\rtheorySemN{1}$, a tick rule
    application may not advance time far enough for a transition to become
    enabled.     $\rtheorySemN{2}$
      adds a ``global clock'', denoting how much time has elapsed in the system. 
       This allows for  analyzing  time-bounded properties (can a certain state be reached
  in a certain time interval?).\vspace*{-1mm}

\subsubsection{The theory $\rtheorySemN{1}$} To avoid consecutive tick
rule applications, we  add a new component---whose value is
either \texttt{tickOk} or \texttt{tickNotOk}---to the global
state. The tick rule can only be applied when this new component of the
global state has the value \texttt{tickOk}. We therefore add a new
constructor \verb@_:_:_:_@ for these extended global states, a new
sort \texttt{TickState} with values \texttt{tickOk} and
\texttt{tickNotOk}, and modify (or add) the two rewrite rules below:

\begin{maude}
sort TickState .
ops tickOk tickNotOk : -> TickState [ctor] .
op _:_:_:_ : TickState Marking ClockValues Net -> State [ctor] .
var TS : TickState .

crl [applyTransition] :
    $\highlight{TS}$ : M : ((L -> T) ; CLOCKS) : (L : PRE --> ...) ; NET)
   =>
    $\highlight{tickOk}$ : ((M - PRE) + POST) :  ...
   if active(...) and (T in INTERVAL) .

crl [tick] : $\highlight{tickOk}$ : M : ...  =>  $\highlight{tickNotOk}$ : M : increaseClocks(...) ...
   if  T <= mte(M, CLOCKS, NET) [nonexec] .
\end{maude}

\begin{theorem}\label{th:r0r1}
    Let  $m\; \code{:} \;\mathit{clocks}\; \code{:} \; \mathit{net}$
    be a ground term of sort \code{State} in
  $\rtheorySem$. Then \vspace*{-1mm}
 \[m\; \code{:} \;\mathit{clocks}\; \code{:} \; \mathit{net}  \longrightarrow^*_{\rtheorySem} m'\; \code{:} \;\mathit{clocks'}\; \code{:} \; \mathit{net}\]
        if and only if \vspace*{-1mm}
   \[\mathtt{tickOk}\; \code{:}\; m\; \code{:} \;\mathit{clocks}\; \code{:} \; \mathit{net}
    \longrightarrow^*_{\rtheorySemN{1}}
    \mathtt{tickNotOk}\; \code{:}\; m'\; \code{:} \;\mathit{clocks'}\; \code{:} \; \mathit{net}.\]
\end{theorem}
\begin{proof}
For the ($\Leftarrow$) side, it suffices to follow in $\rtheorySem$ the
    same execution strategy as in $\rtheorySemN{1}$.
    For ($\Rightarrow$), it suffices to perform the following
    (reachability-preserving) change in the $\rtheorySem$ trace: the
    application of two consecutive \code{tick} rules with $T=t_1$ and $T=t_2$
    are replaced by a single application of \code{tick} with $T=t_1 +
    t_2$.
    This is enough to show that the same trace can be obtained in $\rtheorySemN{1}$.
\end{proof}

Although reachability is preserved, an ``arbitrary'' application of
the tick rule in  $\rtheorySemN{1}$, where time does not advance far
enough for a transition to be taken, could lead to a deadlock in
$\rtheorySemN{1}$ which does not correspond to a deadlock in
$\rtheorySemN{0}$.

\subsubsection{The theory $\rtheorySemN{2}$}  To analyze
whether a certain state can be reached in a certain time interval, and
to enable time-bounded analysis where behaviors  beyond the time
bound are not explored, we
add a new component, denoting the ``global time,''  to the global state:

\begin{maude}
op _:_:_:_$\highlight{@_}$ :  TickState Marking ClockValues Net $\highlight{Time}$ -> State [ctor] .
\end{maude}


The \code{tick} and \code{applyTransition} rules are modified as expected. For instance,
the rule \code{tick} becomes:
\begin{maude}
var GT : Time .

crl [tick] :
     tickOk    : M : CLOCKS : NET $\highlight{@ GT}$
    =>
     tickNotOk : M : increaseClocks(..., T) : NET $\highlight{@ GT + T}$
    if  T <= mte(M, CLOCKS, NET) [nonexec] .
\end{maude}

\noindent For analyses up to some
time bound $\Delta$, we  add a conjunct \code{GT\,+\,T\,<=}$\;\Delta$
to the condition of this rule
to stop executing beyond the time bound.

\medskip
Let $t$ and $t'$ be terms of sort \code{State} in $\rtheorySem$. We say that
$t'$ is reached in time $d$  from $t$, written $t
\dtransition^*_{\rtheorySem} t'$, if $t \longrightarrow^*_{\rtheorySem} t'$ and
$d$ is the sum of the values taken by the variable \code{T} in the different
applications of the rule \code{tick} in such a trace.

 \begin{theorem}\label{th:r0r2}
    \[m\; \code{:} \;\mathit{clocks}\; \code{:} \; \mathit{net} \dtransition^*_{\rtheorySem} m'\; \code{:} \;\mathit{clocks'}\; \code{:} \; \mathit{net}\]
    if and only if
    \[\mathtt{tickOk}\; \code{:}\; m\; \code{:} \;\mathit{clocks}\; \code{:} \; \mathit{net}\; \code{@}\; 0
    \longrightarrow^*_{\rtheorySemN{2}}
    \mathtt{tickNotOk}\; \code{:}\; m'\; \code{:} \;\mathit{clocks'}\; \code{:} \; \mathit{net}\; \code{@}\; d.\]
\end{theorem}

\begin{proof}
    From \Cref{th:r0r1}, we know that the \code{tickOk}/\code{tickNotOk}
    strategy can be followed in $\rtheorySemN{0}$ to produce
    an equivalent trace. Using that trace, the result follows trivially
    by noticing that applications of
    \code{tick} in $\rtheorySemN{0}$ with \code{T}$=\delta$
    ($t \dtransitionD{\delta}_{\rtheorySem} t'$)
    match applications of \code{tick} in
     $\rtheorySemN{2}$ with the same
    instance of \code{T}, thus advancing the global clock in
    exactly $\delta$ time units.
\end{proof}

\section{Explicit-state analysis of ITPNs in Maude}
\label{sec:concrete-ex}

The theories $\mathcal{R}_0$--$\mathcal{R}_2$ cannot be directly
executed in Maude, since the \texttt{tick} rule introduces a new
variable \texttt{T} in its right-hand side.
Following the Real-Time Maude~\cite{tacas08,rtm-journ}
methodology for analyzing dense-time systems, although we cannot cover
all time points, we can choose to ``sample'' system execution at \emph{some}
time points. For example, in this section we change the \texttt{tick}
rule to  increase time by \emph{one time unit} in each
application:

\begin{maude}
crl [tickOne] :
    M$\;$:$\;$CLOCKS$\;$:$\;$NET
   =>
    M$\;$:$\;$increaseClocks(M,$\,$CLOCKS,$\,$NET,$\,\highlight{1}$)$\;$:$\;$NET
   if  $\highlight{1}$ <= mte(M, CLOCKS, NET) .
\end{maude}

Analysis with such time sampling  is in general not sound and
complete, since it does not cover all possible system
behaviors: for example, if some transition's firing interval is
$[0.5,0.6]$, we could not execute that transition with this time
sampling.  Nevertheless, if all interval bounds are natural
numbers, then  ``all behaviors'' should be covered.

We can therefore  quickly prototype our specification and experiment with
different parameter values, before applying the sound and complete
symbolic analysis and parameter synthesis methods developed in the
following sections.

\begin{figure}[h]
    \begin{center}
    \includegraphics[width=0.7\textwidth]{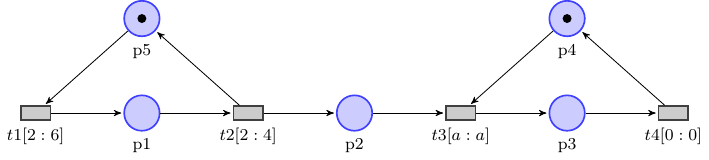}
\end{center}\vspace*{-5mm}
    \caption{A simple \texttt{production-consumption} system taken from
      \cite{Wang1998}.\label{fig:producers}}\vspace*{-2mm}
\end{figure}

The   term \code{net3($l$,$u$)} represents an instance of a  more
general version of the
 net in Fig.~\ref{fig:producers}, where the interval for
 transition $t_3$ is $[a,b]$, and where the parameters $a$ and $b$
 are instantiated
 with  values $l$ and $u$:
\begin{maude}
op net3 : Time TimeInf -> Net .
var LOWER : Time .  var UPPER : TimeInf .
eq net3(LOWER, UPPER)
 = "t1" : "p5" |-> 1               -->  "p1" |-> 1               in [2 : 6] ;
   "t2" : "p1" |-> 1               -->  "p2" |-> 1 ; "p5" |-> 1  in [2 : 4] ;
   "t3" : "p2" |-> 1 ; "p4" |-> 1  -->  "p3" |-> 1               in [LOWER : UPPER] ;
   "t4" : "p3" |-> 1               -->  "p4" |-> 1               in [0 : 0] .
\end{maude}

\noindent The initial marking in Fig.~\ref{fig:producers} is represented by the term
\texttt{init3}:
\begin{maude}
op init3 : -> Marking .
eq init3 = "p1"$\:$|->$\;$0 ; "p2"$\:$|->$\;$0 ; "p3"$\:$|->$\;$0 ; "p4"$\:$|->$\;$1 ; "p5"$\:$|->$\;$1 .
\end{maude}

We can simulate 2000 steps of the net with  different parameter
values:
\begin{maude}
Maude> $\maudeit{\blue{rew} [2000] init3 : initClocks(net3(3,5)) : net3(3,5) .}$

result State:
"p1" |-> 0 ; "p2" |-> 1 ; "p3" |-> 0 ; "p4" |-> 1 ; "p5" |-> 1 :  ...  :  ...
\end{maude}

\noindent To further analyze the system, we define a function \code{k-safe}, where
\code{k-safe($n$,$\,m$)}   holds iff the marking $m$ does not have any
place with more than $n$ tokens:
\begin{maude}
  op k-safe : Nat Marking -> Bool .

  eq k-safe(N1, empty) = true .
  eq k-safe(N1, P$\:$|->$\;$N2 ; M) = N2 <= N1 and k-safe(N1, M) .
\end{maude}

We can then quickly (in 5ms) check whether the net is 1-safe when transition
$t_3$ has interval $[3,4]$ by searching for a reachable state whose
marking \code{M} has some place holding more than one token
(\code{not k-safe(1,M)}):

\begin{maude}
Maude> $\maudeit{\blue{search} [1] init3 : initClocks(net3(3,4)) : net3(\textcolor{red}{3,4})}$
                 $\maudeit{=>*}$
                  $\maudeit{M : CLOCKS : NET \blue{such that} not k-safe(\textcolor{red}{1}, M) .}$

Solution 1 (state 27)
M --> "p1" |-> 0 ; $\highlight{"p2" |-> 2}$ ; "p3" |-> 0 ; "p4" |-> 1 ; "p5" |-> 1
CLOCKS --> "t1" -> 0 ; "t2" -> 0 ; "t3" -> 4 ; "t4" -> 0
NET --> ...
\end{maude}

\noindent The net is not 1-safe: we reached a state with two tokens in place
$p_2$.
However, the net is 1-safe if $t_3$'s interval is instead
$[2,3]$:

\begin{maude}
Maude> $\maudeit{\blue{search} [1] init3 : initClocks(net3(2,3)) : net3(\textcolor{red}{2,3})}$
                 $\maudeit{=>*}$
                  $\maudeit{M : CLOCKS : NET \textcolor{blue}{such that} not k-safe(1,\,M) .}$

No solution.
\end{maude}
\noindent Further analysis shows that \texttt{net3(3,4)} is 2-safe, but that
\texttt{net3(3,5)}
is not even 1000-safe.

\medskip
We can also analyze concrete  instances of our net by full linear
temporal logic (LTL) model checking in Maude. For example, we can
define a parametric atomic proposition
``\code{place$\;p\;$has$\;n\;$tokens}'', which holds in a state iff its
marking has exactly $n$ tokens in place $p$:

\begin{maude}
op place_has_tokens : Place Nat -> Prop [ctor] .
eq  (P$\:$|->$\;$N1 ; M : CLOCKS : NET) |= place P has N2 tokens = (N1 == N2) .
\end{maude}

Then we can check properties such as whether in \emph{each} behavior of the
system, there will be infinitely many states where $p_3$ has no tokens
\emph{and} infinitely many states where it holds one
token:\footnote{\texttt{[]}, \texttt{<>}, \texttt{O}, \texttt{/\char92}, and
  \texttt{\char126} are the Maude representations of corresponding (temporal)
  logic operators $\Box$ (``always''), $\Diamond$ (``eventually''),
  $\bigcirc$ (``next''),
  conjunction, and negation.}
\eject
\begin{maude}
Maude> $\maudeit{\textcolor{blue}{red} modelCheck(init3 : initClocks(net3(3,4)) : net3(3,4),}$
                      $\maudeit{(\textcolor{brown}{[] <>} place "p3" has 0 tokens) \textcolor{brown}{/\char92} (\textcolor{brown}{[] <>} place "p3" has 1 tokens)) .}$

result Bool: true
\end{maude}

To analyze more complex nested LTL formulas, we can check whether there
\emph{exists}\footnote{In LTL, there is behavior satisfying
$\phi$ if there is a counterexample showing that $\neg \phi$ does not hold
in all behaviors.}
  a behavior such that, from some point on, place $p_3$
  alternates between holding zero and holding one token:

\begin{maude}
Maude> $\maudeit{\textcolor{blue}{red} modelCheck(init3 : initClocks(net3(3,4)) : net3(3,4),}$
                     $\textcolor{brown}{\thicksim\:} \maudeit{(\textcolor{brown}{<> []} (place "p3" has 0 tokens \textcolor{brown}{<->} \textcolor{brown}{O} place "p3" has 1 tokens))) .}$

result Bool: true
\end{maude}

The output shows that no such behavior exists.
Furthermore, it turns out (maybe surprisingly?) that there is no
 behavior in which
we eventually reach a situation where each ``empty'' $p_3$ is
followed by a non-empty $p_3$ in at most two steps:

\begin{maude}
Maude> $\maudeit{\textcolor{blue}{red} modelCheck(init3 : initClocks(net3(3,4)) : net3(3,4),}$
                      $\textcolor{brown}{\thicksim\:}  \maudeit{(\textcolor{brown}{<> []} (place p3 has 0 tokens
                 \;\textcolor{brown}{->}}$
                              $\maudeit{ ((\textcolor{brown}{O} place p3 has 1 tokens)
			   \;  \textcolor{brown}{\char92/} (\textcolor{brown}{O O} place p3 has 1 tokens))))) .}$

result Bool: true
\end{maude}

We know that  \texttt{net3(3,4)} can reach markings with two
tokens in $p_2$; but is this inevitable (\ie{} does it happen in \emph{all}
behaviors)?

\begin{maude}
Maude> $\maudeit{\textcolor{blue}{red} modelCheck(init3 : initClocks(net3(3,4)) : net3(3,4),}$
                      $\maudeit{\textcolor{brown}{<>} place "p2" has 2 tokens) .}$

result ModelCheckResult: counterexample(...)
\end{maude}

\noindent The result is a counterexample showing a path where $p_2$ never holds
two tokens.

We also obtain a ``time sampling'' specification for time-bounded
analyses following the techniques for
$\mathcal{R}_2$; namely,  by adding a global time
component to the state:

\begin{maude}
op _:_:_$\highlight{@_}$ : Marking ClockValues Net $\highlight{Time}$ -> State [ctor] .
\end{maude}

\noindent and modifying the tick rule to increase this global clock
according to the elapsed time. Furthermore, for time-bounded analysis
  we add  a constraint ensuring that system execution does not go
  beyond the time bound $\Delta$:

\begin{maude}
crl [executableTick] :
    M : CLOCKS : NET @ $\highlight{GT}$
   =>
    M : increaseClocks(M,CLOCKS,NET,1) : NET @ $\highlight{GT + 1}$
   if  GT < $\Delta$  and   --- remove this condition for unbounded analysis
       1 <= mte(M, FT, NET) .
\end{maude}

By setting $\Delta$ to \texttt{1000}, we can simulate one behavior of
the system \texttt{net3(3,5)} up to time 1000:

\begin{maude}
Maude> $\maudeit{\blue{rew} init3 : initClocks(net3(3,5)) : net3(3,5) \textcolor{red}{@ 0} .}$

result State:
"p1" |-> 0 ; "p2" |-> 1 ; "p3" |-> 0 ; "p4" |-> 1 ; "p5" |-> 1 : ... : ... $\highlight{@ 1000}$
\end{maude}

\eject
We can then check whether \texttt{net3(3,4)} is one-safe in the time
interval $[5,10]$ by setting $\Delta $ in the tick rule to
\texttt{10}, and execute following command:

\begin{maude}
Maude> $\maudeit{\blue{search} [1] init3 : initClocks(net3(3,4)) : net3(3,4) \textcolor{red}{@ 0}}$
                 $\maudeit{=>*}$
                  $\maudeit{M : CLOCKS : NET \textcolor{red}{@ GT} \blue{such that} not k-safe(1, M) \textcolor{red}{and GT >= 5} .}$

Solution 1 (state 68)
MARKING --> "p1" |-> 0 ; $\highlight{"p2" |-> 2}$ ; "p3" |-> 0 ; "p4" |-> 1 ; "p5" |-> 1
...
$\highlight{GT --> 8}$
\end{maude}
\noindent This shows that the  non-one-safe marking can be reached in eight
time units.  It is worth noticing that this result only shows that (one of) the shortest
rewrite path(s) to a non-one-safe marking has duration \code{8}. We
cannot conclude from this, however,  that such a state cannot be reached in
shorter time, which could be possible with more rewrite steps.

\section{Parameters and symbolic executions}
\label{sec:sym}

Standard explicit-state Maude analysis of the theories
$\rtheorySem$--$\rtheorySemN{2}$ cannot be used to analyze all
possible
behaviors of PITPNs for two reasons:
\begin{enumerate}
\itemsep=0.8pt
  \item
 The rule \code{tick} introduces a new variable \code{T} in its right-hand
    side, reflecting the fact that time can
    advance by \emph{any} value  \code{T <= mte(...)}; and
    \item
    analyzing \emph{parametric} nets with \emph{uninitialized}
  parameters is
  impossible with  explicit-state Maude analysis of
  concrete states. Note, for instance, that  the condition \code{T in INTERVAL} in
  rule \texttt{applyTransition} will   never evaluate to \code{true} if \code{INTERVAL}
  is not a \emph{concrete} interval, and hence the rule will never
  be applied.
\end{enumerate}
 Maude-with-SMT analysis of \emph{symbolic} states with SMT variables can
 solve both  issues, by
symbolically representing the time advances \code{T} and the net's
uninitialized parameters. This enables 
 analysis and parameter synthesis methods  for  analyzing \emph{all}
possible behaviors in
dense-time systems with unknown parameters.  

\medskip
This section defines a  rewrite
theory $\rtheorySym$ that faithfully models PITPNs and that can be
symbolically executed using Maude-with-SMT.
We prove that (concrete)
executions in $\rtheorySemN{1}$ are captured by (symbolic) executions in
$\rtheorySym$,  and vice versa.
  We also show that standard folding
techniques \cite{DBLP:journals/jlap/Meseguer20} in rewriting modulo SMT are not
sufficient  for collapsing equivalent symbolic states in
$\rtheorySym$. We therefore
 propose a new folding technique that guarantees termination of
 reachability
 analyses in  $\rtheorySym$ when the
 state-class graph of the encoded PITPN is finite.
 We present two implementations of the folding procedure
 that we  benchmark in \Cref{sec:benchmarks}.

\subsection{The symbolic rewriting logic
  semantics}\label{subsec:sym-theory}

We define the ``symbolic'' semantics of PITPNs using the rewrite
theory $\rtheorySym$, which  is the symbolic counterpart
of $\rtheorySemN{1}$, instead of basing it on $\rtheorySemN{0}$,  since a symbolic
``tick'' step represents all
possible tick steps from a  symbolic state. We therefore  do not
 introduce deadlocks not
possible in the corresponding PITPN by disallowing multiple
consecutive  (symbolic) applications of the tick rule.
\eject

$\rtheorySym$ is obtained from $\rtheorySemN{1}$
  by replacing
the sort \code{Nat} in markings  and the sort \code{PosRat} for clock values
with the corresponding SMT sorts \code{IntExpr} and
\code{RExpr}. (The former is only needed to enable reasoning
with \emph{symbolic} initial states where the number of tokens in a
location is unknown). 
Moreover, conditions in rules (\eg{} \code{M1 <= M2})  are replaced
with the corresponding
SMT expressions of sort \code{BoolExpr}. The symbolic execution of
$\rtheorySym$ using  Maude with SMT will
accumulate and check the satisfiability of the constraints needed for a
parametric transition to happen.

\medskip
We  declare the sort \code{Time} as follows:

\begin{maude}
sorts Time TimeInf .
subsort $\highlight{RExpr}$ < Time < TimeInf .
op inf : -> TimeInf [ctor] .
\end{maude}

\noindent where \code{RExpr} is the sort for SMT reals. (We add constraints to
the rewrite rules to  guarantee that only non-negative real numbers are
considered as time values.) Besides rational constants of sort \code{Rat},
terms in  \code{RExpr} can be also SMT variables.

\medskip
Intervals are defined as in $\rtheorySem$:
\lstinline{op `[_:_`] : Time TimeInf -> Interval}.
Since \lstinline{RExpr} is a subsort of \lstinline{Time},
an interval in $\rtheorySym$
may contain SMT variables. This means that a parametric interval
$[a,b]$ in a PITPN can be represented as the term
\code{[A$\;$:$\;$B]}, where
\code{A} and \code{B} are variables of sort \code{RExpr}.

\medskip
The definition of markings, nets,  and clock values
is similar to the one in \Cref{subsec:eq-theory}. We only
need to modify the following definition for markings:

\begin{maude}
op _|->_ : Place $\highlight{IntExpr}$ -> Marking [ctor] .
\end{maude}

\noindent Hence, in a pair \code{$\nameit(p)$ |-> $e_I$}, $e_I$ is an SMT
integer expression that could be/include SMT variable(s).

\medskip
Operations on markings and intervals remain the same,  albeit with the
appropriate SMT sorts. The operators in Maude for the sorts \code{Nat} and
\code{Rat} have the same signature as those for \code{IntExpr} and
\code{RExpr}. Therefore,  the specification needs few modifications.
For instance,
the new definition of \code{M1 <= M2} is:

\begin{maude}
vars N1 N2 : $\highlight{IntExpr}$ .
op _<=_ : Marking Marking -> $\highlight{BoolExpr}$ .
eq ((P |-> N1) ; M1) <= ((P |-> N2) ; M2) = N1 $\highlight{<=}$ N2 and (M1 <= M2) .
eq empty <= M2 = true .
\end{maude}

\noindent where  \code{<=}   in \code{N1 <= N2} is a
function
\lstinline{op _<=_ : IntExpr IntExpr -> BoolExpr}.

\medskip
Symbolic states in $\rtheorySym$ are defined as follows:

\begin{maude}
sort State.
op _:_:_:_ : TickState Marking ClockValues Net -> State [ctor]
\end{maude}

The rewrite rules in $\rtheorySym$ act on
symbolic states that may contain SMT variables.
Although  these rules  are similar to those in
$\rtheorySemN{1}$,  their symbolic execution  is completely
different. Recall from \Cref{sec:prelim}
the theory transformation
to implement symbolic rewriting
in Maude-with-SMT.
In the
resulting theory $\widehat{\rtheorySym}$,
when a rule is
applied, the  variables occurring in the right-hand side but not in the
left-hand side are replaced by fresh variables,
represented as terms of the form \code{rr(id)}
of sort \lstinline{RVar} (with~\lstinline{subsort RVar < RExpr})
where$\,$  \code{id}$\,$
is a term of sort$\,$  \code{SMTVarId} (theory$\,$
\code{VAR-ID} in Maude-SE).

\eject
   \noindent \mbox{Moreover,}
rules in $\widehat{\rtheorySym}$ act on constrained terms of the form
$\phi\parallel t$, where $t$ in this case is a term of sort \code{State} and
$\phi$ is a satisfiable SMT Boolean expression (sort \code{BoolExpr}). The constraint $\phi$ is
obtained by accumulating the conditions in rules, thereby restricting the
possible values of the variables in $t$.

\medskip
The tick rewrite rule in $\rtheorySym$ is

\begin{maude}
var T : RExpr .
crl [tick] :  tickOk    : M : CLOCKS     : NET
             =>
              tickNotOk : M : increaseClocks(M, CLOCKS, NET, $\highlight{T}$) : NET
             if ($\highlight{T >= 0}$ and mte(M, CLOCKS, NET, T)) .
\end{maude}

The variable \code{T} is restricted to be a non-negative real
number and to satisfy the following \emph{predicate} \code{mte}, which gathers
the constraints to ensure that time cannot advance beyond the point in
time when an enabled transition \emph{must} fire:

\begin{maude}
var R1     : RExpr .
vars T1 T2 : Time .

op mte : Marking ClockValues Net $\highlight{RExpr}$ -> $\highlight{BoolExpr}$ .
eq mte(M, empty, NET, T) = true .
eq mte(M, (L -> R1) ; CLOCKS, (L : PRE --> ... in [T1 : $\highlight{inf}$]) ; NET , T)
 = mte(M, CLOCKS, NET, T) .
eq mte(M, (L -> $\highlight{R1}$) ; CLOCKS, (L : PRE --> ...  in [T1 : $\highlight{T2]}$) ; NET, $\highlight{T}$)
 = (active(M, L : ...) $\highlight{?\ T <= T2 - R1 :\ true}$) and mte(M, CLOCKS, NET, T) .
\end{maude}

This means that, for every transition \code{L},  if the upper bound of the
interval in \code{L} is \code{inf}, no restriction on \code{T} is added.
Otherwise, if \code{L} is active at marking \code{M}, the SMT ternary operator
\code{$c\;$?$\;e_1\;$:$\;e_2$} (checking $c$ to choose either
$e_1$ or $e_2$)
further constrains \code{T} to be less than \code{T2$\,$-$\,$R1}.  The
definition of \code{increaseClocks} also uses  this SMT
operator to represent the
new values of the clocks:

\begin{maude}
eq increaseClocks(M, (L -> R1) ; CLOCKS, (L : PRE --> ... ) ; NET, T)
 = (L -> (active(M, L : PRE ...) $\highlight{?\ R1 + T :\ R1 }$)) ;
   increaseClocks(M, CLOCKS, NET, T) .
\end{maude}\smallskip

The rule for applying a transition is defined as follows:\smallskip

\begin{maude}
crl [applyTransition] :
    TS     : M                  : ((L -> T) ; CLOCKS) : (L : PRE --> ...) ; NET)
   =>
    tickOk : ((M - PRE) + POST) :  updateClocks(...)  : (L : PRE --> ...) ; NET)
  if active(...) and (T in INTERVAL) .
\end{maude}

When applied, this rule adds new constraints asserting that
the transition \code{L} can be fired
(predicates \code{active} and \code{_in_}) and  updates
the state of the clocks:
\begin{maude}
eq updateClocks((L' -> R1)$\;\,$;$\;\,$CLOCKS, INTERM-M, (L'$\,$:$\,$PRE --> ...); NET)
 = (L ->  PRE <= INTERM-M $\highlight{?\ R1 :\ 0}$) ; updateClocks(...) .
\end{maude}

In the example below, \code{k-safe($k$,$m$)} is a predicate stating that
the marking $m$ does not have more than $k$ tokens
in any place.

\begin{example}\label{ex:producers}
Let \code{$\mathit{net}$}
and \code{$m_0$}
    be the Maude terms representing, respectively,   the PITPN and  the initial
    marking shown in \Cref{fig:producers}. The term \code{$\mathit{net}$} below
    includes an SMT variable representing the parameter $a$ in transition $t_3$:
  \eject

  \hbox{}
  \vspace*{-9mm}
\begin{maude}
op net : -> Net .
eq net =
  "t1" : ("p5" |-> 1)               -->  ("p1" |-> 1)               in  [2 : 6] ;
  "t2" : ("p1" |-> 1)               -->  ("p2" |-> 1 ; "p5" |-> 1)  in  [2 : 4] ;
  "t3" : ("p2" |-> 1 ; "p4" |-> 1)  -->  ("p3" |-> 1)               in  $\highlight{[rr("a") : rr("a")]}$ ;
  "t4" : ("p3" |-> 1)               -->  ("p4" |-> 1)               in  [0 : 0] .
\end{maude}
The Maude commands  introduced in \Cref{sec:analysis} allow us to
    answer the question whether it is possible to reach a state with a marking
    $M$ with more than one token in some place. As shown in Example
    \ref{example:finite}, Maude positively answers this
    question and the resulting accumulated constraint tells us that
    such a state
    is reachable (with 2 tokens in $p_2$) if  \code{rr("a") >= 4}.
\end{example}

Terms of sort \code{Marking} in $\rtheorySymN{1}$ may contain expressions
with parameters (\ie{} variables)  of sort \code{IntExpr}.
Let  $\Lambda_m$  denote
the set of such parameters and $\pi_m:\Lambda_m \to \grandn$
a valuation function for them.
We use $m_s$ to denote a mapping from places to \code{IntExpr}
expressions including parameter variables. Similarly, $\mathit{clocks}_s$
 denotes a mapping from transitions
to \code{RExpr} expressions (including variables).
We write
$\pi_m(\mathit{m}_s)$
to denote the ground term where the parameters in
markings are replaced by the corresponding values
$\pi_m(\lambda_i)$. Similarly for
$\pi(\mathit{clocks}_s)$, we use $\encSym{\PN}$ to
  denote   the above rewriting logic
  representation of nets in  $\rtheorySymN{1}$.

  Let $t_s$ be the constrained term  of
  sort \code{State} in $\widehat{\rtheorySymN{1}} $
  and assume that
  $\phi \parallel t_s
  \rightsquigarrow_{\widehat{\rtheorySymN{1}}} \phi' \parallel t'_s$.
  By construction, if for all $t\in \os \phi\parallel t_s \cs$
  all markings (sort \code{IntExpr}), clocks and parameters (sort
  \code{RExpr}) are non-negative
  numbers, then this is also the case for all $t' \in \os \phi' \parallel t_s'\cs$.

  The following theorem states that the symbolic semantics matches all the behaviors
resulting from a  concrete execution of $\rtheorySemN{1}$ with
arbitrary parameter valuations $\pi$ and $\pi_m$. Furthermore,
 for all symbolic executions with parameters, there exists
a corresponding concrete execution where the parameters are instantiated with
values consistent with the resulting accumulated constraint.

\begin{theorem}[Soundness and completeness]\label{th:sym-correct}
    Let $\PN$ be a  PITPN
    and $m_s$ be a marking possibly including parameters.

  \smallskip
    \noindent\textbf{(1)}
            Let $\phi$ be the constraint
                            $\bigwedge_{t\in T}(\leftEP{\parIntervalStatic(t)} \leq \rightEP{\parIntervalStatic(t)})
                            \wedge \bigwedge_{\lambda_i\in \Lambda_m}(0 \leq \lambda_i)
                            $. If
                            \[\phi\parallel
                            \mathtt{tickOk}\;\code{:}\;m_s\;\code{:}\;\mathit{clocks_s}\;\code{:}\;\encSym{\PN}
                            \rightsquigarrow^*_{\widehat{\rtheorySymN{1}}}
                            \phi'\parallel TS'\;\code{:}\;m'_s\;\code{:}\;\mathit{clocks'_s}\;\code{:}\;\encSym{\PN}
                        \]
                            \noindent then, there exist a parameter valuations $\pi$ and a parameter marking valuation  $\pi_m$ s.t.
                            \[\mathtt{tickOk}\:\code{:}\;\pi_m(m_s)\;\code{:}\;\mathit{clocks}\;\code{:}\;\encBase{\pi(\PN)}
                            \longrightarrow^*_{\rtheorySemN{1}}
                            TS'\:\code{:}\;\pi_m(m'_s)\;\code{:}\;\mathit{clocks'}\;\code{:}\;\encBase{\pi(\PN)}
                        \]
                        \noindent where the constraint
                            $\phi'\wedge\bigwedge_{\lambda_i\in\Lambda}
                            {\lambda_i} = \pi(\lambda_i)
                            \wedge \bigwedge_{\lambda_i\in\Lambda_m}
                            {\lambda_i} = \pi_m(\lambda_i)$
                            is satisfiable,
                            $\mathit{clocks}\in \os \phi \parallel \mathit{clocks}_s \cs$
                            and $\mathit{clocks'}\in \os \phi' \parallel \mathit{clocks'}_s \cs$.

    \vspace{1.8mm} \noindent  \textbf{(2)}
            Let $\pi$ be a parameter valuation and $\pi_m$ a parameter marking valuation.
            Let $\phi$ be the constraint
                            $\bigwedge_{\lambda_i\in \Lambda}(\lambda_i=\pi(\lambda_i))
                            \wedge \bigwedge_{\lambda_i\in
                              \Lambda_m}(\lambda_i=\pi_m(\lambda_i))
                            $.
            If
            \[\mathtt{tickOk}\;\code{:}\;\pi_m(m_s)\;\code{:}\;\mathit{clocks}\;\code{:}\;\encBase{\pi(\PN)}
              \longrightarrow^*_{\rtheorySemN{1}}
          TS'\;\code{:}\;m'\;\code{:}\;\mathit{clocks'}\;\code{:}\;\encBase{\pi(\PN)}\]
                 \noindent then
                   \[\phi \parallel \mathtt{tickOk}\;\code{:}\;m_s\;\code{:}\;\mathit{clocks_s}\;\code{:}\;\encSym{\PN}
                   \longrightarrow^*_{\widehat{\rtheorySymN{1}}}
               \phi' \parallel TS'\;\code{:}\;m_s'\;\code{:}\;\mathit{clocks_s}'\;\code{:}\;\encSym{\PN}\]
                            \noindent where
                            $m' \in \os \phi'\parallel m'_s\cs$,
                            $\mathit{clocks}\in \os \phi \parallel \mathit{clocks}_s \cs$
                            and $\mathit{clocks'}\in \os \phi' \parallel \mathit{clocks'}_s \cs$.
  \end{theorem}

\begin{proof}
This result is a direct consequence of
soundness and completeness of rewriting modulo SMT \cite{rocha-rewsmtjlamp-2017}.
More precisely, from  \cite{rocha-rewsmtjlamp-2017} we know that:
if $\phi_t \parallel t \rightsquigarrow^\ast \phi_u \parallel u$
then $t' \longrightarrow^\ast u'$
 for some $t' \in \llbracket \phi_t \parallel t \rrbracket$
 and $u' \in \llbracket \phi_u \parallel u \rrbracket$; and
if  $t' \longrightarrow^\ast u'$ with
$t' \in \llbracket \phi_t \parallel t \rrbracket$,
then there exists $\phi_u$ and $t_u$ s.t.
 $\phi_t \parallel t \rightsquigarrow^\ast \phi_u \parallel u$
with $u' \in \llbracket \phi_u \parallel u \rrbracket$.
\end{proof}

The symbolic counterpart $\rtheorySymN{2}$ of the theory $\rtheorySemN{2}$,
that adds a component to the state denoting the ``global time'',
can be defined similarly. We have also defined a third symbolic theory
$\rtheorySymN{3}$ that replaces the \code{applyTransition} rule in $\rtheorySymN{1}$ with the following one:

\begin{maude}
crl [applyTransition] :
    $\highlight{tickNotOk}$ : M                  : ((L -> T) ; CLOCKS) : (L : PRE --> ...) ; NET)
   =>
    tickOk    : ((M - PRE) + POST) :  updateClocks(...)  : (L : PRE --> ...) ; NET)
  if active(...) and (T in INTERVAL) .
\end{maude}

Hence, in $\rtheorySymN{3}$, two consecutive applications of
\code{applyTransition} (without
applying \code{tick} in between) are not allowed.

\subsection{A sound and complete folding method for symbolic
  reachability}\label{subsec-folding}

Reachability analysis should terminate
  for both
positive and negative queries for nets with
finite parametric state-class graphs.
However,  the state space
resulting from the (symbolic) execution of the theory ${\rtheorySym}$
is infinite even for such nets, and it
 will not terminate when the desired states are
 unreachable.
 We note that the  standard {search} commands in Maude  stop exploring from  a
symbolic state only if it has already visited  the \emph{same}
state. Due to the fresh variables created
in $\widehat{\rtheorySym}$ whenever the \texttt{tick} rule is applied, symbolic
states representing the same set
of concrete states are not the same, even though they are \emph{logically}
equivalent, as exemplified below.

\begin{example}\label{ex:infinite}
    Let $\phi$ be the constraint $0 \leq a < 4$.
    The following command, trying to show that the PITPN in
    \Cref{fig:producers} is 1-safe if the parameter $a$ satisfies $\phi$, does not
    terminate.
\begin{maude}
search $\phi$ || tickOk$\;$:$\;m_0\;$:$\;$0-clock($net$)$\;$:$\;\mathit{net}$  =>*  PHI || TICK : M : CLOCKS : NET
such that smtCheck(PHI and $\highlight{not k-safe(1,M)}$) $\,$.
\end{maude}

\noindent Furthermore, the command

\begin{maude}
search $\phi$ || tickOk$\;$:$\;m_0\;$:$\;$0-clock($net$)$\;$:$\;\mathit{net}$  =>*  PHI || TICK : M : CLOCKS : NET
such that smtCheck(PHI and $\highlight{M <=  \(m_0\) and \(m_0\) <= M}$) $\,.$
\end{maude}

\noindent searching for reachable states where
    $M = m_0$ will produce infinitely many (equivalent) solutions,
    including, \eg{} the following constraints:

\noindent\adjustbox{width=\textwidth}{
    \begin{maude}
Solution 1:  #p5-9:IntExpr === 1 and #t3-9:RExpr  + a:RExpr -  #t2-9:RExpr  <= 0/1 and ...
Solution 2: #p5-16:IntExpr === 1 and #t3-16:RExpr + a:RExpr -  #t2-16:RExpr <= 0/1 and ...
    \end{maude}
}

\noindent     where the variables created by the search procedure start with \verb@#@
    and end with a number taken from a sequence to guarantee freshness.
    Let $\phi_1\parallel t_1$ and $\phi _2\parallel t_2$ be, respectively,  the
    constrained terms found in \code{Solution 1} and \code{Solution 2}.
    In this particular output, $\phi_2\parallel t_2$ is obtained by further rewriting
    $\phi_1\parallel t_1$.
    The variables representing the state of markings and clocks
    (\eg{} \code{\#p5-9} in $t_1$ and \code{\#p5-16} in $t_2$) are clearly different,
    although they
    represent the same set of concrete values ($\os \phi_1\parallel
    t_1\cs = \os \phi_2\parallel t_2 \cs$).
    Since constrains are accumulated when a rule is applied, we note that
    $\phi_2$ equals $\phi_1 \wedge \phi_2'$ for some $\phi_2'$, and
    ${\it vars}(\phi_1\parallel t_1)\subseteq {\it vars}(\phi_2\parallel t_2)$.
\end{example}

The usual approach for collapsing equivalent symbolic states in rewriting
modulo SMT is subsumption
\cite{DBLP:journals/jlap/Meseguer20}. Essentially, we stop searching
from a symbolic state if, during the search, we have already encountered
another (``more general'') symbolic state that subsumes (``contains'') it.
More precisely, let $U = \phi_u \parallel t_u$ and
$V = \phi_v \parallel t_v$ be constrained terms.
Then
$U  \sqsubseteq  V$  if there is
a substitution $\sigma$ such that $t_u = t_v\sigma$  and
the implication $\phi_u \Rightarrow \phi_v\sigma $ holds. In that case,
$\llbracket U \rrbracket \subseteq \llbracket V \rrbracket$.
A search will not further explore a constrained term $U$ if another
constrained term $V$ with $U \sqsubseteq V$ has already been
encountered. 
It is known that such reachability analysis with folding
is sound (does not
generate spurious counterexamples~\cite{bae2013abstract})
but not necessarily complete (since
$\llbracket U \rrbracket \subseteq \llbracket V \rrbracket$
does not imply $U \sqsubseteq V$).

\begin{example}
    Let $\phi_1$ and $\phi_2$ be
    the resulting constraints in the two solutions found by the second
    \lstinline{search} command in Example \ref{ex:infinite}. Let $\sigma$ be the
    substitution that maps \code{\#p$i$-9} to \code{\#p$i$-16}
    and \code{\#t$j$-9}  to \code{\#t$j$-16} for each place $p_i$ and
    transition $t_j$.
    The SMT solver determines that  the formula $\neg (\phi_2
    \Rightarrow \phi_1\sigma)$
    is satisfiable (and therefore  $\phi_2 \Rightarrow \phi_1\sigma$ is not
    valid).
    Hence, a procedure based on checking this implication will fail
    to determine that the state in the second solution can be subsumed by the state
    found in the first solution.
    \end{example}

The satisfiability witnesses of $\neg (\phi_2 \Rightarrow \phi_1\sigma)$ can
give us some ideas on how to make the subsumption procedure more precise.
Assume that $\phi_1$ carries the information $R = T_0$ for some
clock represented by $R$ and $T_0$ is a tick variable subject to $\phi = (0 \leq
T_0 \leq 2)$. Assume also that in $\phi_2$, the value of the same clock is
$R' = T_1 + T_2$ subject to $\phi' = (\phi\wedge T_1 \geq 0 \wedge T_2\geq 0 \wedge T_1+T_2 \leq 2)$.
Let $\sigma= \{R \mapsto R'\}$. Note that  $(R'=T_1+T_2
\wedge \phi \wedge \phi')$ does not imply $(R = T_0\wedge \phi)\sigma$ (take, \eg{}
the valuation $T_1=T_2=0.5$ and $T_0=2$). The key
observation is that, even if $R$ and $R'$ are both constrained to be
in the interval $[0,2]$ (and hence represent the same state for this
clock), the assignment of $R'$ in the antecedent does not need
to coincide with the one for $R$ in the consequent of the implication.

\medskip
In the following,  we propose a subsumption relation that solves the
aforementioned problems.
Let $\phi \parallel t$ be a constrained term where $t$ is a term of sort \code{State}.
Consider the abstraction of built-ins $(t^\circ, \sigma^\circ)$ for
$t$,
where $t^\circ$ is as $t$ but it replaces the expression $e_i$ in
markings ($p_i\mapsto e_i$)
and clocks ($l_i \to e_i$)  with new fresh variables.
The substitution $\sigma^\circ$ is defined accordingly,
such that  $t = t^\circ \sigma^\circ$ (see \Cref{sec:rew-smt}).
Let $\Psi_{\sigma^\circ} = \bigwedge_{x \in
  \mathit{dom}(\sigma^\circ)} x = x \sigma^\circ $.
We use $\projnow{(\phi\parallel t)}{}$ to denote the constrained
term $\phi \wedge \Psi_{\sigma^\circ} \parallel t^\circ $.
Intuitively,  $\projnow{(\phi\parallel t)}{}$ replaces the clock values and markings
 with fresh variables and the
 Boolean expression $\Psi_{\sigma^\circ}$
 constrains those variables to take the values
 of clocks and marking in $t$. From \cite{rocha-rewsmtjlamp-2017} we can show that
 $\os \phi \parallel t \cs = \os \projnow{(\phi\parallel t)}{}\cs$.

Note that the
only variables occurring in $\projnow{(\phi\parallel t)}{}$
are those for parameters (if any)
and the fresh variables in $\mathit{dom}(\sigma^\circ)$
(representing the symbolic state of clocks and markings).
For a constrained term $\phi\parallel t$, we use
$\exists(\phi \parallel t)$ to denote the formula
$(\exists X) \phi$ where
$X = \mathit{vars}(\phi) \setminus \mathit{vars}(t)$.

 \begin{definition}[Relation $\preceq$]\label{def:rel}
 Let $U = \phi_u \parallel t_u$ and $V = \phi_v \parallel t_v$ be
 constrained terms where $t_u$ and $t_v$ are terms of sort \code{State}.
 Moreover, let $\projnow{U}{} = \phi_u'\parallel t_u'$ and
$\projnow{V}{}= \phi_v'\parallel t_v'$,
where $\ovars{t_u'} \cap \ovars{t_v'} = \emptyset$.
We define
 the relation $\preceq$ on constrained terms
 so that  $U\preceq V$
 whenever there exists a substitution $\sigma$ such that  $t_u' = t_v'\sigma$
 and the formula
 $\exists(\projnow{U}{}) \Rightarrow \exists(\projnow{V}{})\sigma$ is valid.
\end{definition}

The formula $\exists(\projnow{U}{})$ uses the existential quantification to
\emph{hide} the information about all the tick
variables created so far (in the previous time instants). We therefore obtain
a constraint representing only
the information about
 the parameters and the values of the clocks and markings ``now''.
 Moreover, if $t_u$ and $t_v$ above are both \code{tickOk} states (or both
 \code{tickNotOk} states),
  and they represent two symbolic states of the same PITPN,
 then $t_u'$ and $t_v'$ always match ($\sigma$ being
 the identity on the variables representing parameters
 and mapping the corresponding variables created
 in $\projnow{V}{}$  and $\projnow{U}{}$).

 \begin{theorem}[Soundness and completeness]\label{th:folding}
     Let $U$  
     and $V$  
     be
     constrained terms in $\widehat{\rtheorySym}$ representing two
     symbolic states of the same
     PITPN.
     Then,
     $\os U \cs \subseteq \os V \cs$ iff $U \preceq V$.
\end{theorem}
\begin{proof}
 Let
 $\projnow{U}{} = \phi_u'\parallel t_u'$ and
$\projnow{V}{}= \phi_v'\parallel t_v'$,
where $\ovars{t_u'} \cap \ovars{t_v'} = \emptyset$.
Let
$X_u = \ovars{\phi_u'} \setminus \ovars{t_u'}$
and
$X_v = \ovars{\phi_v'} \setminus \ovars{t_v'}$.
By construction,
$t_u', t_v' \in T_{\Sigma \setminus \Sigma_0}(X_0)$,
$\exists(\projnow{U}{}) = (\exists X_u) \phi_u'$,
and
$\exists(\projnow{V}{}) = (\exists X_v) \phi_v'$.
It suffices to show
$\os \projnow{U}{} \cs \subseteq \os \projnow{V}{} \cs$
iff
$U \preceq V$.

\medskip
\noindent
 ($\Rightarrow$)
Assume $\os \projnow{U}{} \cs \subseteq \os \projnow{V}{} \cs$.
Then,
$t_u'$ and $t_v'$ are $E$-unifiable (witnessed by $w \in \os
\projnow{U}{} \cs \cap \os \projnow{V}{} \cs$). Since
$t_v'$
has no duplicate variables
and
$E$ only contains structural axioms for $\Sigma \setminus \Sigma_0$,
by the matching lemma \cite[Lemma~5]{rocha-rewsmtjlamp-2017},
there exists a substitution $\sigma$
with  $t_u' = t_v' \sigma$ (equality modulo ACU).
Since
any built-in subterm of $t_u'$ is a variable in $X_0$,
$\sigma$ is a renaming substitution
$\sigma: X_0 \to X_0$
and
thus $\os \projnow{V}{} \cs = \os (\projnow{V}{})\sigma \cs$.

Suppose $\exists(\projnow{U}{}) \Rightarrow \exists(\projnow{V}{})\sigma$ is not valid,
\ie{} $((\exists X_u) \phi_u') \wedge  (\forall X_v) \neg \phi_v'\sigma$ is satisfiable.
Let $Y$
be the set of free variables
in $((\exists X_u) \phi_u') \wedge  (\forall X_v) \neg \phi_v'\sigma$.
Notice that $Y = \ovars{t_u'} = \ovars{t_v'\sigma}$.
Let $\rho : Y \to T_{\Sigma_0}$
be a ground substitution that represents
a satisfying valuation of $(\exists X_u) \phi_u' \wedge  (\forall X_v) \neg \phi_v'\sigma$.
Then, $t_u'\rho \in \os \projnow{U}{} \cs$
but $t_u'\rho = t_v' \sigma\rho \notin \os (\projnow{V}{})\sigma \cs = \os \projnow{V}{} \cs$,
which is a contradiction.

\medskip
\noindent ($\Leftarrow$)
Assume $U \preceq V$.
There exists a substitution $\sigma: X_0 \to X_0$ such that $t_u' = t_v'\sigma$
and $\exists(\projnow{U}{}) \Rightarrow \exists(\projnow{V}{})\sigma$ is valid.
Let $Y$ be the set of free variables in $\exists(\projnow{U}{}) \Rightarrow \exists(\projnow{V}{})\sigma$.
As mentioned above,
$\os \projnow{V}{} \cs = \os (\projnow{V}{})\sigma \cs$
and
$Y = \ovars{t_u'} = \ovars{t_v'\sigma}$.
Let $w \in \os \projnow{U}{} \cs$.
Then,
for some ground substitution $\rho_u$, $w = t_u'\rho_u$
and
$\phi_u'\rho_u$ holds.
From the assignments in $\rho_u \vert_Y$,
we can build a
valuation  $\mathcal{V}$ making true $\exists (\projnow{U}{})$
and, by assumption,  making also true $\exists (\projnow{V}{})\sigma$.
Hence,
 there exists a ground substitution $\rho_v$ (that agrees on the values assigned
 in $\mathcal{V}$) such that $\phi_v'\sigma \rho_v$ holds
 and $\rho_u \vert_Y \!= \rho_v \vert_Y$.
Notice that
$w\! = t_u'\rho_u = t_u' (\rho_u \vert_Y)\! = t_v'\sigma (\rho_v \vert_Y) \!= t_v'\sigma \rho_v$.
Therefore, $w \in \os (\projnow{V}{})\sigma \cs$. \vspace*{-5mm}
\end{proof}

\subsubsection{Two versions  of symbolic reachability analysis based
  on the folding procedure}\label{sec:imp:folding}

The implementation of the folding procedure in Maude requires the specification
of the subsumption relation $\preceq$ that, in turns, requires dealing with the
existentially quantified variables in $\exists(\projnow{U}{})$. At the time of
publishing the conference version of this paper
\cite{DBLP:conf/apn/AriasBOOPR23}, the only method available for checking the
satisfiability of existentially quantified formulas in Maude involved the call
to the existential quantifier elimination procedure of the SMT solver Z3.
Recently \cite{ftscs-journal}, we have implemented   the well-known
Fourier-Motzkin elimination (FME) procedure \cite{dantzig2006linear} using
equations in Maude, and it is independent
of the SMT solver connected to Maude.
(As demonstrated in \Cref{sec:benchmarks}, this change has an important
effect on the performance of the analysis).

\medskip
Building on this FME equational theory, we propose two different
implementations/versions of symbolic reachability analysis based on the folding
procedure:

\begin{enumerate}
  \item The first one is based on a simple theory transformation
that adds to the (symbolic) state the already visited (symbolic)
states \emph{in that execution branch},
and prevents a rule to be applied if the resulting state would be
subsumed by the states already seen in this branch.  This new
theory can be used directly with the standard search command in Maude. However,
this procedure cannot fold equivalent states appearing in \emph{different} branches of
the search tree.
\item The second implementation/version solves this problem by defining,
using Maude's meta-level facilities, a breath-first search procedure
that keeps a \emph{global} set of
already visited states in all the branches across the search tree.
\end{enumerate}

\paragraph{The first implementation (theory $\rtheorySymNF{1}$).}   A term $t$
in the theory $\rtheorySymNF{1}$ defined below also stores the states that have been
visited before reaching $t$. Additionally, a rewrite step from $t$ to a term $t'$
is only
possible if the state represented by $t'$ is not subsumed by
the already visited states. Hence, theory $\rtheorySymNF{1}$ can be used to
perform reachability analysis with folding as illustrated in Example \ref{example:finite}.

\medskip
 We transform
 the theory $\rtheorySymN{1}$ into a rewrite theory
$\rtheorySymNF{1}$
 that rewrites terms of the form
 $\mathit{Map} :\phi\parallel t$. The term  $\mathit{Map}$,
 of sort \code{Map} (defined in Maude's prelude \cite{maude-book}),
 is
 a set of entries of the form $m \mapsto \psi$ where $m$
is a  \code{Marking}
 and $\psi$ is a \code{BoolExpr} expression.
 $\mathit{Map}$
 stores the already visited states, and $\mathit{Map} [m]$
 is the accumulated constraint leading to the state where
 the marking is $m$.
 The theory $\rtheorySymNF{1}$ defines
 the operator
 \code{subsumed}$(\phi\parallel t~,~ \mathit{Map})$ that
 checks whether the symbolic state in the first parameter is subsumed
 by one of the states stored  in $\mathit{Map}$.

\medskip
 A rule  $l : q \longrightarrow r \mbox{ \textbf{if} } \psi$
 in $\rtheorySymN{1}$ is transformed into the following rule in $\rtheorySymNF{1}$:
\[
\begin{array}{ll}
l:~&
\mathit{Map} : \mathtt{PHI} \parallel q^\circ
\longrightarrow
(add(\phi_r\parallel r, \mathit{Map})) : \phi_r \parallel r\\
& \mbox{ \textbf{if} }
\mathtt{smtCheck}(\phi_r) \wedge \mathtt{not~subsumed}(\phi_r\parallel r, \mathit{Map})
\end{array}
\]

\noindent where $\mathtt{PHI}$ is a \code{BoolExpr} variable and $(q^\circ,
\sigma^\circ)$ is an abstraction of built-ins for $q$ and $\phi_r =
(\mathtt{PHI} \land \psi \land \Psi_{\sigma^\circ} )$. Note
that the transition happens only if the new state $\phi_r\parallel r$ is not
subsumed by an already visited state. The function $add(\phi_r\parallel r,
\mathit{Map})$ checks whether there is an entry $m \mapsto \psi$  where $m$ is the
marking in the term $r$. If this is the case, this entry is updated to $m
\mapsto \psi \vee \phi_r$. Otherwise, a new entry $m \mapsto \phi_r$ is added
to the map. In the first case, a symbolic state where the marking is $m$ was
already visited, but with constraints for the clocks and parameters that are
not implied by $\phi_r$ (and hence not subsumed). Therefore, the new state is
merged, and the marking $m$ can be reached by satisfying either $\psi$ or
$\phi_r$. 

In $\rtheorySymNF{1}$, for an initial constraint $\phi$ on the parameters,
the command\\

\noindent\lstinline[mathescape]{search [$n$,$m$]  empty : $\phi\parallel t$ =>* $Map$ : $\phi'\parallel
t'$ such that smtCheck($\phi'\wedge \Phi)$} \\

\noindent answers the question whether it is possible to reach a symbolic
state that matches $t'$ and satisfies the condition $\Phi$.
In the following, we denote by $\code{init}(net, m_0, \phi)$  the following term (of
sort \code{State}):
$\code{empty} : \phi\parallel \mathtt{tickOk}\;\code{:}\;m_0\;\code{:}\;
\code{initClocks}(net)\;\code{:}\;net$.

 \begin{example}\label{example:finite}
Consider  the  PITPN in \Cref{fig:producers}.
    Let  $m_0$ be the marking in the figure
    and $\phi$ be the constraint $0 \leq a < 4$. The command

     \begin{maude}
search init($net$, $m_0$, $\phi$) =>* MAP : PHI ||  ( TICK : M : CLOCKS : NET )
       such that smtCheck(PHI and not k-safe(1,M)) . \end{maude}

     \noindent terminates, and  returns \code{No solution}, showing that
    the net is 1-safe if $0 \leq a < 4$.
 \end{example}

The following  result shows that if the set of reachable state classes in the symbolic
semantics of $\PN$ (see \cite{EAGPLP13}) is finite, then so  is the set of
 reachable symbolic states using the new folding technique.

 \begin{corollary}\label{corollary:term}
For any PITPN $\PN$ and state class $(M,D)$,
if the transition system $(\mathcal{C}, (M,D), \Fleche{})$ is finite,
then
so  is
$\left(T_{\Sigma , \texttt{State}}, \code{init}(\PN,M,D), \symredif\right)$.
 \end{corollary}
 \begin{proof} Assume that $(\mathcal{C}, (M,D), \Fleche{})$ is a finite transition
     system and, to obtain a contradiction, that there are infinitely many
     $\symredif$-reachable states from $\code{init}(\PN,M,D)$.

     Since
     $\left(T_{\Sigma , \texttt{State}}, \code{init}(\PN,M,D), \symredif\right)$
     is finitely branching, there must be an infinite
     sequence of the form $U_0 ~\symredif~ U_1 ~\symredif~ \cdots $ where, by
     definition of $\symredif$, $U_j \not\preceq U_i$ for $i< j$. From
     \Cref{th:folding} we know that $\os U_j\cs \not\subseteq \os U_i\cs$. By
     \Cref{th:sym-correct}, this means that after each transition, more concrete
     different states are found. Hence, the set of reachable state classes cannot
     be finite, leading to a contradiction.
 \end{proof}

 \paragraph{The second implementation.}
 The above symbolic reachability analysis implementation/method is
 able to fold equivalent states only
 when they appear in the \emph{same branch} of the search tree (since
 different branches store
 different maps of already visited states). The advantage of that
 implementation is that it can be used
 directly with Maude's \lstinline{search} command, as exemplified above.

 Using the meta-programming capabilities of Maude, we have specified a
 second symbolic reachability analysis
 procedure that implements, from scratch, a breath-first search procedure where
 equivalent symbolic states are folded. This implementation maintains
 a \emph{global} $\mathit{Map}$,
 thus allowing for folding symbolic states occurring in \emph{different branches} of
 the search tree. The next level of the search tree can be easily
 computed by calling the meta-level function
 \code{metaSearch} to return the successor states.
 The states in the new frontier are checked for
 subsumption  and the
 non-visited ones are  added to the global map.
 Since equivalent states appearing in different branches can be folded,
 the resulting search space should be smaller than the one
 induced by the theory $\rtheorySymNF{1}$. However, there is an inherent
 performance penalty to pay for this second method,  due to the calls to the meta-level
 operations and for not using
 the (optimized) C++ implementation  of  Maude's \lstinline{search} command.

Section~\ref{sec:analysis:methods} introduces ``user-friendly'' commands for
using these different  symbolic reachability analysis implementations,
and also provides user-friendly command syntax for all
the analyses that can be performed with our framework.


\section{Parameter synthesis and symbolic model
  checking}\label{sec:analysis}

This section  shows how Maude-with-SMT
  can be used for a wide
range of formal analyses beyond reachability analysis. 
We show how to use Maude for solving
parameter synthesis problems and for reasoning   with
parametric initial states where the number of tokens in the different
places is not known (Section~\ref{sec:param-synthesis}),  analyzing
nets with  user-defined execution
strategies (Section~\ref{sec:strategies}), and
model checking the classes of
non-nested timed temporal logic properties supported by the
PITPN tool \romeo{} (Section~\ref{sec:model-checking}).  We 
thereby provide analysis methods
that go beyond those supported by \romeo{}, while
supporting almost all forms of analysis provided by \romeo{}.

\medskip
We provide a wide range of analysis methods, requiring different
Maude commands on slightly different transformed models.
To make all these
analysis methods easily  accessible to the PITPN user,
we have also implemented a number of ``commands'' (or operations) that
provide a user-friendly syntax/interface for the various
analysis methods. These commands/operations are summarized in
Section~\ref{sec:analysis:methods}.

\subsection{Analysis methods and properties}
\label{sec:analysis:methods}
Sections \ref{sec:concrete-ex} and \ref{sec:sym}
 present some  explicit-state and symbolic analysis methods that can
be performed with Maude's standard search and model checking commands.
Using  the different theories and the folding procedures, however,
requires that the user  imports different Maude files (each theory
is defined in its own Maude file) and deals
with possibly different representations of states (take for instance
the map in $\rtheorySymNF{1}$).

\medskip
This section introduces  a    user-friendly syntax  for properties and commands
to ease the use of our framework. The definitions
introduced below, in a single file \code{analysis.maude},  allow us to
perform the different analyses using the theories
$\rtheorySymN{1}$, $\rtheorySymN{2}$, $\rtheorySymN{3}$ and
$\rtheorySymNF{1}$,
the folding procedure described in \Cref{sec:imp:folding}, and invoking  Maude's
model checker.

\medskip
We start with the operators needed to perform the
different analyses:

\begin{maude}
--- Analysis with the symbolic theory R1S
op search-sym     [_,_] in_:_s.t._ :   Nat Nat Qid InitState Prop          -> LResult .

--- Analysis with the symbolic theory R3S
op search-sym2    [_,_] in_:_s.t._ :   Nat Nat Qid InitState Prop          -> LResult .

--- Analysis  with folding, theory R1fS
op search-folding [_,_] in_:_s.t._ :   Nat Nat Qid InitState Prop          -> LResult .

--- Analysis with the folding procedure at the meta-level
op folding [_,_] in_:_s.t._ :          Nat Nat Qid InitState Prop          -> LResult .

--- Analysis including a global clock with theory R2S
op $\texttt{search}$ [_,_] in_:_s.t._in-time_ :   Nat Nat Qid InitState Prop Interval -> LResult .

--- AG-synthesis using folding
op AG-synthesis in_:_s.t._ :                   Qid InitState Prop          -> BoolExpr .

--- Analysis with a user-defined strategy
op strat-rew [_] in_:_using_s.t._ :    Nat     Qid InitState Strategy Prop -> LResult .

--- Model checking
op A-model-check in_:_satisfies_in-time_ :     Qid InitState Formula Rat   -> Bool .
op E-model-check in_:_satisfies_in-time_ :     Qid InitState Formula Rat   -> Bool .
\end{maude}

These operators allow us to:
\begin{itemize}
  \item perform  reachability analysis and solve
parameter synthesis problems using the theories $\rtheorySymN{1}$
(\code{search-sym}) and $\rtheorySymN{3}$ (\code{search-sym2}, where applications
of rules \code{tick} and \code{applyTransition} must be interleaved);
\item perform reachability
analysis with folding using the theory $\rtheorySymNF{1}$ and the new
breath-first search procedure with a global set of already visited
states (\code{folding});
\item perform time-bounded reachability analysis
    (\code{search in-time}) using the theory $\rtheorySymN{2}$ (that adds to the state the global clock);
\item solve safety synthesis problems (\code{AG-synthesis});
  \item analyze
    a net with a user-defined strategy (\code{strat-rew}); and
    \item model check non-nested
      (T)CTL formulas (\code{A-model-check} and \code{E-model-check}).
\end{itemize}
      The parameters of these commands are explained next.

\medskip
The parameters of sort \code{Nat}  are
both optional and can be omitted in the reachability analysis
commands.  They can be used to
limit the number of solutions to be returned and the maximal depth of the
search tree. 
The parameter of sort \code{Qid} (quoted
identifiers, as \code{'MODEL})
is the name of the module with the user-defined PITPN. The initial state (sort
\code{InitState}) is a triple containing the definition of the net, the initial marking
and the initial constraint on the parameters:
\begin{maude}
op (_,_,_) : Net Marking BoolExpr -> InitState .
\end{maude}

We define the following atomic propositions,  or \emph{state
  predicates}; these will be used to define the (un)desired properties
of our (symbolic) states:

\begin{maude}
ops _>=_ _>_ _<=_ _<_ _==_ : Place IntExpr -> Prop . --- On markings
ops _>=_ _>_ _<=_ _<_ _==_ : Label RExpr   -> Prop . --- On clocks
op reach     : Marking                     -> Prop .
op k-bounded : IntExpr                     -> Prop .
op _and_     : Prop Prop                   -> Prop .
op _or_      : Prop Prop                   -> Prop .
op not_      : Prop                        -> Prop .
op diff>     : Label Label RExpr           -> Prop . --- | clock1 - clock2 | >= R
op in-time_  : Interval                    -> Prop .
\end{maude}

The atomic formula \code{p} $\bowtie$ \code{n} ($\bowtie \in\{\texttt{>=},
\texttt{>}, \texttt{<=},\texttt{<}, \texttt{==}\}$) states
that the current marking satisfies  $p\bowtie n$ where, for instance,
$p < 3$ holds in a state if  place $p$ holds less than $3$ tokens.
Similarly, \code{t} $\bowtie$ \code{r} is true in a state
where the value $c$ of the clock associated to transition $t$ satisfies $c\bowtie r$.
The predicate \code{reach($M$)} is true in a marking $M'$
if $M \leq M'$,  and
\code{k-bounded($n$)} is true in a state where each place
holds less than or equal to $n$ tokens.
For temporal properties (see \Cref{sec:model-checking} for more details),
the formula \code{in-time $\mathit{INTERVAL}$} is true
if the current value of the global clock is in  $\mathit{INTERVAL}$.
Other predicates can be built with the usual Boolean operators.

\medskip
The following operators and equations (we present only some cases)
define when a state $S$ satisfies a state property $\phi$, denoted by
$S \models \phi$:
\begin{maude}
var STATE       : State .
vars PROP PROP' : Prop  .
vars IE IE'     : IntExpr .

op _|=_   : State Prop -> Bool
eq STATE |= PROP = check(STATE, PROP) .

op check  : State Prop -> Bool .
ceq check(STATE, P >= IE) = smtCheck(constraint(STATE) and IE' >= IE)
    if   (P |-> IE' ; M) := marking(STATE) .
eq check(STATE, reach(M)) = smtCheck(constraint(STATE) and M <= marking(STATE)) .
eq check(STATE, PROP and PROP') = check(STATE, PROP) and-then check(STATE, PROP') .
\end{maude}

Since the state may contain SMT variables, a call to the SMT solver
is needed to determine whether the state entails the
given  property. The functions \code{constraint}
and \code{marking} are the expected projection functions, returning
the appropriate components of the state.

\medskip
Terms of sort \code{LResult} are lists of terms of the following sort \code{Result}:
\begin{maude}
sort Result .
op _when_ : Marking BoolExpr -> Result .

\end{maude}

\noindent where the term \code{$M$ when $\phi$} represents the fact
that the marking $M$ can be reached with accumulated
constraint $\phi$.

In the operator/command \code{strat-new}, it is possible to limit only the number of
solutions to be returned, when analyzing a system under a given user-defined
\code{Strategy}. As explained below, (T)CTL and (T)LTL temporal \code{Formula}s
can be model checked with a time bound specified by the last parameter in
the operators \code{A-model-check} and \code{E-model-check}.

\medskip
The definition of the  above  search, folding, synthesis,  and
model-checking operators use the \code{META-LEVEL} module in Maude
(for meta-programming) to:
\begin{itemize}
    \item define a new theory (a term of sort \code{Module}, a meta-representation
        of a rewrite theory) that imports both the needed
rewrite theory (\eg{} $\rtheorySymN{1}$) and the module with the
user-defined net;  
\item
    compute the meta-representation of the initial state
according to the theory to
be used;
\item invoke the needed Maude command (\eg{} \code{metaSearch}) to search for a
  solution; and
  \item simplify the resulting constraints by invoking the FME procedure,
eliminating all the SMT variables except those representing parameters. As shown
below, this last step allows for finding the constraints on the parameter
values that enable certain execution paths.
\end{itemize}

In the following subsections we exemplify the use of these operators
and how they can be used to solve interesting problems of
PITPNs.

\subsection{Parameter synthesis}  \label{sec:param-synthesis}

\emph{\EF{}-synthesis} is the problem of computing
parameter values $\pi$ such that there exists a run of $\pi(\PN)$ that reaches a state
satisfying a given state predicate $\phi$. The \emph{safety synthesis problem}
\AG{}$\neg\phi$  is the problem of computing the  parameter values
for which states satisfying $\neg \phi$ are unreachable.

\medskip
Search  (with or without folding) provides a semi-decision procedure for
solving  \EF{}-synthesis problems (which is undecidable in general).
The  synthesis problem \AG{}$\neg\phi$ is solved
 by finding all the solutions
for \EF{}$\phi$ (therefore a search procedure with
 folding is necessary to guarantee termination when possible)
 and then negating the disjunction of all these
solutions/constraints.
If the resulting constraint is unsatisfiable,
it means that there are no values for the parameters guaranteeing that
non-$\phi$-states are unreachable
and hence, there is no solution for $\AG{}\neg\phi$. In that
case, our procedure returns
the constraint \code{false}  (see Example \ref{ex:tutorial1}).

\begin{example}\label{ex:analysis1}
  Example \ref{ex:producers} shows an \EF-synthesis problem: find values
  for the parameter $a$ such that a state with at least two tokens in
  \emph{some} place can be reached. If $\phi = 0 \leq a$, the command
\begin{maude}
Maude> $\maudeit{\blue{red} search-sym  [1] in 'MODEL : (}$$net,$ $m,$ $\phi$$\maudeit{) }$ $\maudeit{s.t. not k-bounded(1) . }$

result Result: ("p1" |-> 0 ; "p2" |-> 2 ; "p3" |-> 0 ; "p4" |-> 1 ; "p5" |-> 1)
                when  4 <= rr("a") .
\end{maude}
finds a state where the place $p_2$ has two tokens.
The resulting constraint (after eliminating all the SMT variables
but not those of the parameters) determines that this is possible
when $4 \leq a$.
The same solution can be found with the commands
\code{search-sym2}, \code{search-folding}, and
\code{folding}.

\medskip
With the same model, we can synthesize the values for the parameter $a$
to reach a state where the difference between the values of the clocks in
transitions $t_1$ and $t_3$ is bigger than, for instance, $10$:

\begin{maude}
Maude> $\maudeit{\blue{red} search-sym [1] in 'MODEL : (}$$net,$ $m,$ $\phi$$\maudeit{)}$ $\maudeit{s.t.\;\:diff>("t1", "t3", 10) .}$

Result: ("p1" |-> 1 ; "p2" |-> 2 ; "p3" |-> 0 ; "p4" |-> 1 ; "p5" |-> 0)
        when 6 <= rr("a")
\end{maude}
  \end{example}

  \romeo{} only supports properties over markings. The state predicate in the
  previous example includes also conditions on the 
  values of the ``clocks'' associated to the transitions, where each
  clock denotes how long the corresponding transition has been active
  without being fired.  

\begin{figure}[h]
\vspace*{-2mm}
    \centering
    \includegraphics[width=0.49\textwidth]{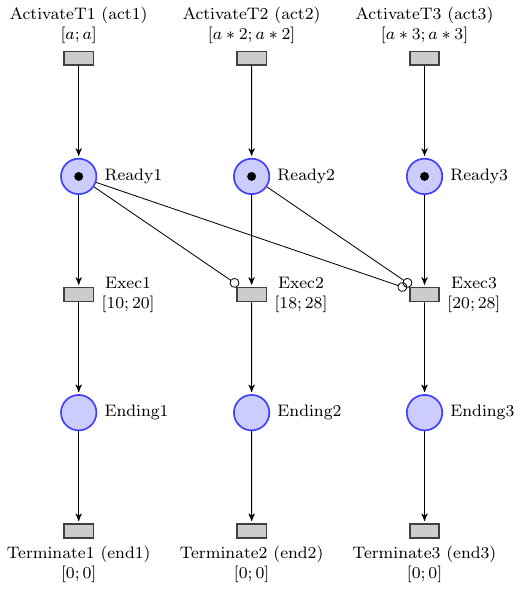}\vspace*{-1mm}
    \caption{The \texttt{scheduling} case study taken
      from~\cite{DBLP:journals/jucs/TraonouezLR09}. \label{fig:ex-scheduling}}\vspace*{-3mm}
\end{figure}

\medskip
  To solve the safety synthesis problem \AG{}$\neg\phi$,
  the \code{AG-synthesis} procedure iteratively calls the command
  \lstinline{folding} to find a state
  reachable from the initial marking $m_0$, with initial constraint $\phi_0$,
  where $\phi$ does not hold. If such state is found, with accumulated
  constraint $\phi'$, the \lstinline{folding} command is invoked again with
  initial constraint $\phi_0 \wedge \neg \phi'$. This process stops when no
  more reachable states where $\phi$ does not hold are found, thus solving the
  \AG{}$\neg\phi$ synthesis problem.

  \begin{example}  \label{ex:analysis2}
    Consider the PITPN in Fig.~\ref{fig:ex-scheduling}, taken from
    \cite{DBLP:journals/jucs/TraonouezLR09}, with 
     a parameter $a$ and three parametric transitions
    with respective firing intervals $[a:a]$, $[2a:2a],$ and $[3a:3a]$.
    \romeo{} can synthesize the values of the parameter $a$
    making the net 1-safe, subject to  initial constraint
    $\phi = 30 \leq a \leq 70$.
    The same query can be answered
    in Maude:
\begin{maude}
Maude> $\maudeit{\blue{red} AG-synthesis  in 'MODEL : (}$$net,$ $m_0,$ $\phi$$\maudeit{) s.t. k-bounded(1) . }$

result BoolExpr: 48 < rr("A") and rr("A") <= 70
\end{maude}
The first counterexample found assumes
    that $a \leq 48$. If  $a>48$, \lstinline{folding}
    does not find any state not satisfying
    \code{k-bounded(1)}. This is the same answer  found by
      \romeo{}. 
 \end{example}

Our  symbolic theories  can have  parameters
(variables of sort \code{IntExpr})  in the initial marking. This opens up the
possibility of using Maude-with-SMT to solve synthesis problems involving
parametric initial markings.  For instance, we can determine the initial markings that
make the net $k$-safe.

\begin{example}\label{ex:pmarking}
    Consider the net in \Cref{fig:producers}, with
    the initial constraint $\phi$ stating
    that $0 \leq a$ and the initial marking $m_0$
    as in the figure.
    The following command shows that the net
    is 1-safe if $0 \leq a < 4$:

\begin{maude}
Maude> $\maudeit{\blue{red} AG-synthesis  in 'MODEL : (}$$net,$ $m_0,$ $\phi\maudeit{) s.t. k-bounded(1) .}$

result ConjRelLinRExpr: rr("a") < 4 and 0 <= rr("a")
\end{maude}

Assume that we fix the parameter $a$ to be, for instance, $1$.
 We may then want to analyze
 whether the net continues to be 1-safe
even if there could be a token initially also in place $p_1$ and/or
place $p_3$ (for illustration purposes, the example
does not give an upper bound on the number of tokens initially in each of these
places).
We therefore consider the
     \emph{parametric} initial marking $m_s$
       with parameters $x_1$ and $x_3$
     denoting the number of tokens in
    places $p_1$  and $p_3$,  respectively, and the initial
    constraint $\phi'$ stating
    that $a = 1$, $ 0 \leq  x_1$,  and $0 \leq x_3$.
    The execution of the following command:
\begin{maude}
Maude> $\maudeit{\blue{red} AG-synthesis  in 'MODEL : (}$$net,$ $m_s,$ $\phi'\maudeit{) s.t.\ k-bounded(1) .}$

Result BoolExpr: ii("x1") < 1 and ii("x3") < 1 and 0 <= ii("x1") and 0 <= ii("x3")

\end{maude}
determines that the net is 1-safe only when both places $p_1$ and
$p_3$ are initially empty.
\end{example}

\begin{figure}[h]
\vspace*{-1mm}
    \centering
\includegraphics[width=.59\textwidth]{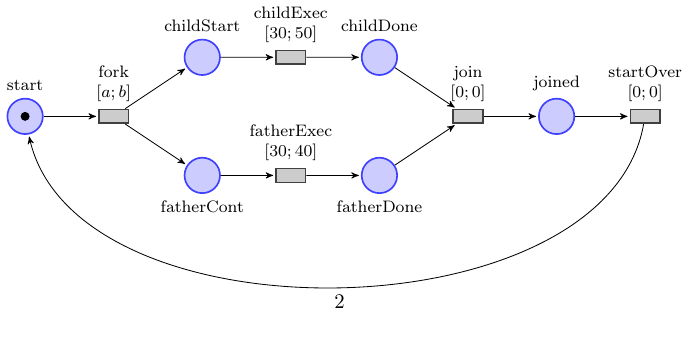}\vspace*{-5mm}
    \caption{The \texttt{tutorial} case study. \label{fig:tutorial}}\vspace*{-4mm}
\end{figure}

\begin{example}\label{ex:tutorial1}
Consider the PITPN
\texttt{tutorial}
in~\Cref{fig:tutorial}, taken from the \romeo{} website. This model
has been modified so that transition \texttt{startOver} produces two tokens,
thus leading to a non-1-safe system.
The first solution to \EF{}$(\neg \texttt{k-bounded(1)})$
 is the  initial constraint
$a \leq b$ already present in this net's $K_0$ (i.e., no further
constraints on these parameters are needed to
reach the   non-1-safe state). Hence, there is no solution for the corresponding
\AG-synthesis problem:

\begin{maude}
Maude> $\maudeit{\blue{red} AG-synthesis  in 'MODEL : (}$$net,$ $m_s,$ $\phi\maudeit{) s.t. k-bounded(1) .}$

result Bool: false
\end{maude}
\end{example}


\subsection{Analysis with user-defined strategies}
\label{sec:strategies}
Maude's strategy language  \cite{maude-manual} allows users to
define  execution strategies for their rewrite theories. This section
explains how we can analyze all possible behaviors of a PITPN allowed
by  a user-defined strategy for the net.
Such analysis  is supported by our
framework through the command
\begin{maude}
red strat-rew [n] in 'MODEL : ($\mathit{net}$, $m_0$, $\phi$) using $\mathit{str}$ s.t. $\psi$
\end{maude}

This command rewrites, using the  meta-level function \code{metaSrewrite},
the  term $(net, m_0, \phi)$ following the strategy
\lstinline[mathescape]{$str$ ; match S s.t. check(S, $\psi$)}
and the rules in $\rtheorySymNF{1}$ (to guarantee termination when possible),
and  outputs the first $n$ solutions.
The
strategy \lstinline[mathescape]{match S s.t. check(S, $\psi$)} fails
whenever
the state $S$
does not satisfy the state property $\psi$. This command therefore returns
the first $n$ reachable states following the strategy $str$ that
 satisfy $\psi$.

\begin{example}
 We  analyze the net in  Fig.~\ref{fig:producers} when all its executions
 (must) adhere to the
 following strategy \lstinline{t3-first}:
 whenever transition $t_3$  and some other
 transition are enabled
 at the same time, then $t_3$  fires first.
This execution strategy can be specified as follows using Maude's
strategy language:
 \begin{maude}
t3-first := ($\;$applyTransition[$\,$L <- "t3"$\,$] or-else all )*
\end{maude}

Starting with the initial constraint $\phi = 0 \leq a$, the execution of the command

\begin{maude}
Maude> $\maudeit{\blue{red}  strat-rew in 'MODEL : (}$$net,$ $m_0,$ $\phi$$\maudeit{) using t3-first s.t.\ k-bounded(1) .}$

result NeLResult:
  (("p1" |-> 0 ; "p2" |-> 0 ; "p3" |-> 0 ; "p4" |-> 1 ; "p5" |-> 1) when ...)
  (("p1" |-> 0 ; "p2" |-> 0 ; "p3" |-> 0 ; "p4" |-> 1 ; "p5" |-> 1) when ...)
  ...
  (("p1" |-> 1 ; "p2" |-> 0 ; "p3" |-> 0 ; "p4" |-> 1 ; "p5" |-> 0) when ...)
\end{maude}
shows that there are 12 possible reachable symbolic states when this
strategy is applied.
Furthermore, the execution of the command
\begin{maude}
Maude> $\maudeit{\blue{red}  strat-rew in 'MODEL : (}$$net,$ $m_0,$ $\phi$$\maudeit{) using t3-first s.t. not k-bounded(1) .}$

result LResult: nil
\end{maude}

\noindent
  returns
  no such non-1-safe symbolic
  states (the empty list \code{nil}). This shows that all markings reachable with
  the strategy \texttt{t3-first} are 1-safe.
  Note that this is not the case
  when the system behaviors are not restricted by such a  strategy.
  As shown in Example \ref{ex:pmarking},
  the parameter $a$ needs to be further constrained ($0 \leq a < 4$) to guarantee
  1-safety.
\end{example}

\subsection{Analyzing temporal properties}\label{sec:model-checking}

This section shows how Maude-with-SMT can be used to analyze the
 temporal properties supported by  \romeo{}~\cite{romeo},
 albeit in a few
 cases without parametric bounds in the temporal formulas.  
\romeo{} can analyze the following temporal properties:
\[
    \textbf{Q}\, \phi\,\CTLU_J\,\psi \;\mid\; \textbf{Q} \CTLF_J\,\phi \;\mid\;
\textbf{Q} \CTLG_J \,\phi \;\mid\; \phi \rightsquigarrow_{\leq b} \psi
\]where $\textbf{Q}\in \{\exists, \forall\}$ is the existential/universal path
quantifier, $\phi$ and $\psi$ are \emph{state predicates} on \emph{markings},
and $J$ is a time
interval  $[a,b]$, where $a$ and/or $b$ can be parameters and $b$ can
be $\infty$. 
For example,
$\forall  \CTLF_{[a,b]}\,\phi$ says that
in \emph{each} path from the initial state,  a marking satisfying
$\phi$ is  reachable in some time
in  $[a,b]$.
The bounded  response $\phi
\rightsquigarrow_{\leq b} \psi $ denotes the formula
$\forall\CTLG(\phi \,\Longrightarrow\,
\forall\CTLF_{[0,b]} \,\psi)$ (each $\phi$-marking \emph{must} be followed
by a $\psi$-marking  within time $b$).  

\medskip
Since queries include time bounds, we use the theory
$\rtheorySymN{2}$ (that adds a component representing the global clock)
so that the term
$\mathtt{tickOk}\;\code{:}\;m\;\code{:}\;clocks\;\code{:}\;net\;\mathtt{@\;}t$
represents a state of the system where the ``global clock'' is $t$.

\medskip
Some of the temporal formulas supported by \romeo{} can be easily verified
using reachability commands similar to the ones presented in the previous section.
The property
$\exists\CTLF_{[a,b]}\,\psi$ can be verified using the command

\begin{maude}
search [1] in 'MODEL : ($net$, $m_0$, $\phi$) s.t. $\psi$ and in-time [$a\:$:$\:b$] .
\end{maude}
where $\phi$ states that all  parameters are non-negative  numbers
(and adds the net's $K_0$ constraints, including that firing
intervals are not empty), and
 $a$ and $b$ can be variables representing parameters to be
 synthesized.

The dual property  $\forall \CTLG_{[a,b]} \:\phi$ can be
checked by analyzing
  $\exists \CTLF_{[a,b]} \:\neg\,\phi$.

\begin{example}
    Consider the PITPN in Fig.~\ref{fig:ex-scheduling} with
    parameter constraint
    $\phi = 30\leq a\leq 70$. The property
    $\exists\CTLF_{[0,b]}(\neg {  \it 1\mhyphen safe} )$
    can be verified with the following command,
    which shows that the desired property
      holds when the upper time bound $b$ in the timed temporal logic
      formula satisfies $60 \leq b$.
\begin{maude}
Maude> $\maudeit{\blue{red} search [1] in 'MODEL : (net, }$$m_0,$ $\phi$$\maudeit{ and 0 <= rr("b"))}$
           $\maudeit{s.t. not k-bounded(1) in-time [0 : rr("b")] .}$

result Result:
  ("END1" |-> 0 ; "END2" |-> 0 ; "END3" |-> 0 ; "R1" |-> 0 ; "R2" |-> 2 ; "R3" |-> 1)
   when ... 2 * rr("a") <= rr("b") and 30 <= rr("a") and rr("a") <= 48 ...
\end{maude}
  \end{example}

The bounded response $\phi \rightsquigarrow_{\leq b} \psi $ formula
can  be verified using a simple theory transformation on
$\rtheorySymN{1}$ followed by
reachability analysis. The theory transformation
adds a new constructor for the sort \code{State}
to build terms of the form $\red{C_\phi}\texttt{ : }M\texttt{ : }{\mathit Clocks}\texttt{ : } {\mathit Net}$,
where $C_\phi$ is either \code{noClock} or \code{clock($\tau$)}; the
latter represents the  time ($\tau$) since
a $\phi$-state was visited, without having been followed by a $\psi$-state.
 The rewrite rules are adjusted to
update this new component as follows.
 The new \code{tick} rule  updates
 \code{clock(T1)} to \code{clock(T1$\;$+$\;$T)} and leaves
 \code{noClock} unchanged.
The rule \code{applyTransition} is split into two rules:
\begin{maude}
crl [applyTransition] : clock(T) : M ... => NEW-TP : M' ...
 if NEW-TP := if $\mathit{STATE'}$ |= $\psi$ then noClock else clock(T) fi /\ ...

crl [applyTransition] : noClock : M ... => NEW-TP : M' ...
 if NEW-TP := if $\mathit{STATE'}$ |= $\phi$ and not $\mathit{STATE'}$ |= $\psi$
    then clock(0) else noClock fi /\ ...
\end{maude}
In the first rule, if a $\psi$-state is encountered, the
new ``$\phi$-clock''  is reset
to \code{noClock}. In the second rule, this ``$\phi$-clock'' starts running
if the new state satisfies $\phi$ but not $\psi$.
The query $\phi \rightsquigarrow_{\leq b} \psi $ can be answered
by   searching for a state where a $\phi$-state has not been followed
by a $\psi$-state before the deadline $b$:
\begin{maude}
search [1] ...  =>*  S : PHI' $\parallel $ clock(T) : ... such that T > $b$ .
\end{maude}

Reachability analysis cannot be used to analyze the other properties
supported by \romeo{} ($\textbf{Q}\, \phi\,\CTLU_J\,\psi$, and
$\forall \CTLF_J\,\phi$ and its dual
$\exists \CTLG_J \,\phi$).  
While developing a full SMT-based \emph{timed} temporal logic model
checker is future
work,  we can  combine Maude's explicit-state model checker and
SMT solving to solve these (and many other) queries. On the positive side, and
beyond \romeo, we can use full LTL.

\medskip
The timed temporal operators $\Diamond_I$, $\mathcal{U}_I$, and
$\Box_I$, written \code{<$\:I\:$>}, \code{U$\;I$}, and \code{[$\,I\,$]},
respectively,  in our framework, can be defined on top of the
(untimed) LTL temporal operators in Maude (\code{<>}, \code{[]}, and
\code{U}) as follows :
\begin{maude}
op <_>_ :  Interval Prop     -> Formula .
op _U__ : Prop Interval Prop -> Formula .
op [_]_ : Interval Prop      -> Formula .

vars PR1 PR2 : Prop .

eq < INTERVAL > PR1   = <> (PR1 /\ in-time INTERVAL) .
eq PR1 U INTERVAL PR2 = PR1 U (PR2 /\ in-time INTERVAL) .
eq [ INTERVAL ] PR1   = ~ (< INTERVAL > (~ PR1)) .
\end{maude}
\eject
For this fragment of  non-nested timed temporal logic formulas, it is possible
to model check universal and existential quantified formulas with the
following commands:

\begin{maude}
op A-model-check in_:_satisfies_in-time_ : Qid InitState Formula Rat -> Bool .
op E-model-check in_:_satisfies_in-time_ : Qid InitState Formula Rat -> Bool .
\end{maude}

Non-nested $A$-formulas
can be directly model checked by calling  Maude's
LTL model checker:

\noindent\lstinline{modelCheck(STATE, F) == true}.
For the $E$-formulas, what we need is to check
whether $\neg F$ does not hold: \lstinline{modelCheck(STATE , ~ F) =/= true}.
The ``\texttt{in time $r$}'' part in the command is optional, and it is used to
perform bounded model checking, forbidding the application
of the tick rule when  the global clock
is  beyond $r$.  This parameter is specially important
for $E$-formulas that require exploring the whole state space to
check whether $F$ does not hold.

\begin{example}\label{ex:tutorial}
Consider the PITPN \texttt{tutorial} in~\Cref{fig:tutorial}.
 Below we model check some formulas and, in comments, we explain the results:
\begin{maude}
--- All paths lead to a state where the number of tokens in place start >= 2
Maude> $\maudeit{\blue{red} A-model-check in 'MODEL : (}$$net,$ $m,$ $\phi\maudeit{) satisfies  (<> (start >= 2)) .}$

result Bool: true

--- The corresponding E-query requires a time bound to terminate
Maude> $\maudeit{\blue{red} E-model-check in 'MODEL : (}$$net,$ $m,$ $\phi\maudeit{) }$ $\maudeit{satisfies  (<> (start >= 2)) in-time 30 .}$

result Bool: true

--- 20 time units are not sufficient for producing 2 tokens
Maude> $\maudeit{\blue{red} E-model-check in 'MODEL : (}$$net,$ $m,$ $\phi\maudeit{) }$  $\maudeit{satisfies  <> (start >= 2) in-time 20 . }$

result Bool: false

--- The net is 1-safe until a state start >= 2 is reached
Maude> $\maudeit{\blue{red} E-model-check in 'MODEL : (}$$net,$  $m,$  $\phi\maudeit{) } $
           $\maudeit{satisfies  (k-bounded(1)) U [0 :  30] (start >= 2) in-time 90 .}$

result Bool: true

--- The existential property is true while the universal is not since there
--- is an execution path where 3 tokens are accumulated in fatherCont (and not in start)
Maude> $\maudeit{\blue{red} A-model-check in 'MODEL : (}$$net,$ $m,$ $\phi\maudeit{) }$
          $\maudeit{satisfies  (k-bounded(2)) U [0 :  90] (start >= 3) in-time 90 .}$

result Bool: false

Maude> $\maudeit{\blue{red} E-model-check in 'MODEL : (}$$net,$ $m,$ $\phi\maudeit{) }$
          $\maudeit{satisfies  (k-bounded(2)) U [0 :  90] (start >= 3) in-time 90 .}$

result Bool: true
\end{maude}\vspace*{-4mm}
\end{example}

\section{Benchmarking}\label{sec:benchmarks}
We have used six PITPNs to compare the performance of our
Maude-with-SMT analysis with that of \romeo{} (version 3.9.4),  and
with that of our previous
implementation presented in \cite{DBLP:conf/apn/AriasBOOPR23}. We compare the
time it takes 
to solve the synthesis problem
\EF{($p > n$)} (\ie{} place $p$ holds more than $n$ tokens), for different
places $p$ and $0 \leq n \leq 2$, and to check whether the net is
$1$-safe.

The
PITPNs used in our experiments are: the
\texttt{producer-consumer} system~\cite{Wang1998} in~\Cref{fig:producers}, the
\texttt{scheduling} system~\cite{DBLP:journals/jucs/TraonouezLR09}
in~\Cref{fig:ex-scheduling},   the \texttt{tutorial} system
in~\Cref{fig:tutorial} taken from the \romeo{} website and modified as
explained in Example \ref{ex:tutorial}, and
the systems
 \texttt{abitpro}, \texttt{train1}, and \texttt{train2}
provided as examples in the \romeo{} distribution\footnote{Other systems
are available in the current \romeo{} distribution but they
use features not supported by ``standard'' PITPNs, including
functions, updates, costs \cite{DBLP:journals/fuin/LimeRS21}, time
inhibitors \cite{DBLP:conf/apn/RouxL04}, etc.}.
For \EF-synthesis problems, we benchmark the performance of commands
implementing reachability with and without folding: \code{search-sym}
(\EF-synthesis using theory $\rtheorySymN{1}$, which does not use
folding), \code{search-sym2} (theory
$\rtheorySymN{3}$, interleaving \code{tick} and \code{applyTransition} rules),
\code{search-folding} (\EF-synthesis with folding using the theory
$\rtheorySymNF{1}$, which uses Maude's \texttt{search} command and folds
symbolic states in the same branch of the search tree) and
\code{folding} (our own meta-level implementation of breath-first
search that folds symbolic states across all branches in the search
tree). For safety synthesis problems, we use
the command \code{AG-synthesis}.

\medskip
We ran all the experiments on a Dell Precision Tower 3430 with a processor Intel
Xeon E-2136 6-cores @ 3.3GHz, 64 GiB memory, and Debian 12. Each
experiment was executed using Maude 3.3.1 connected with the SMTq solver
Yices2, and Maude-SE \cite{maude-se} in combination with
Yices2, CVC4  and Z3.
We use a timeout of 5 minutes. The reader can find the data of
all the experiments and the scripts
needed to reproduce them in~\cite{pitpn2maude}.

\begin{figure}[!ht]
\vspace*{-1mm}
  \centering
  \begin{tabular}{c}
     \includegraphics[width=0.96\textwidth]{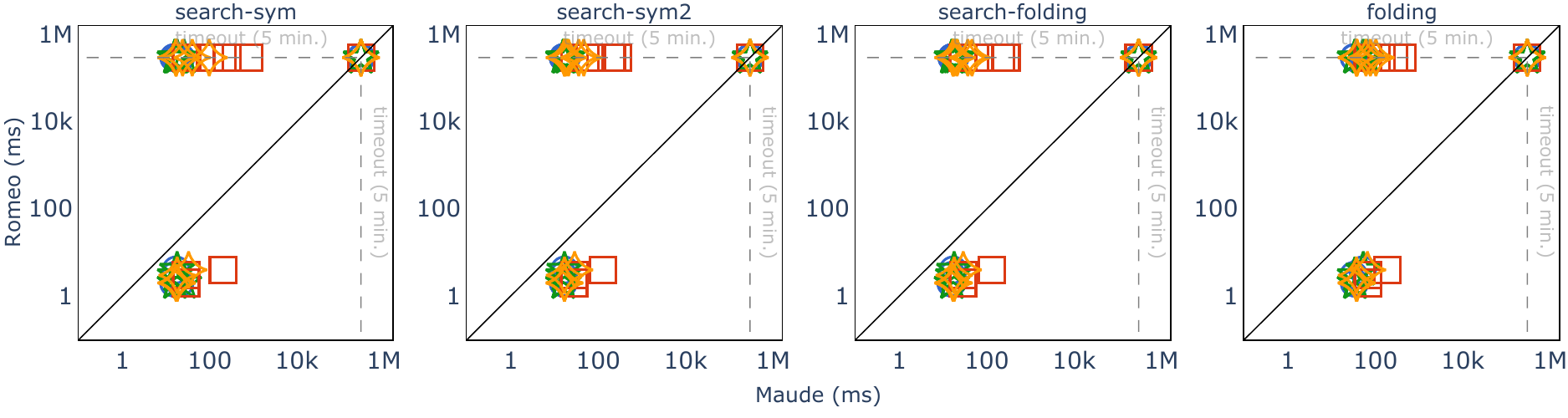}\\[-6pt]
    {(a) \texttt{producer-consumer}\label{fig:bench:A}}
   \\ [4pt] %
   \includegraphics[width=0.96\textwidth]{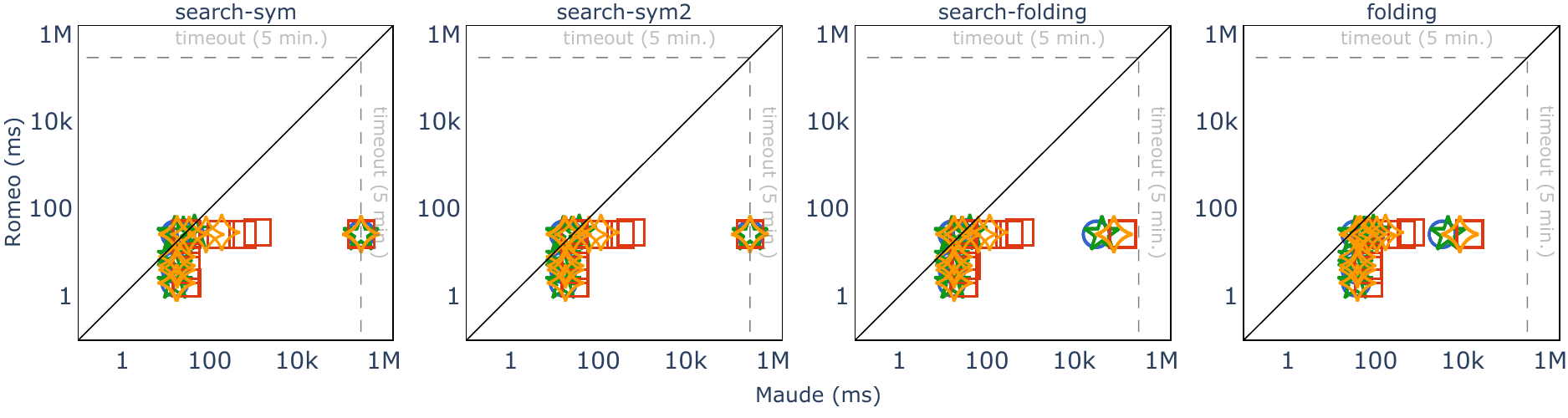}\\[-6pt]
   {(b) \texttt{scheduling}\label{fig:bench:B}}
     \\[4pt]
      \includegraphics[width=0.96\textwidth]{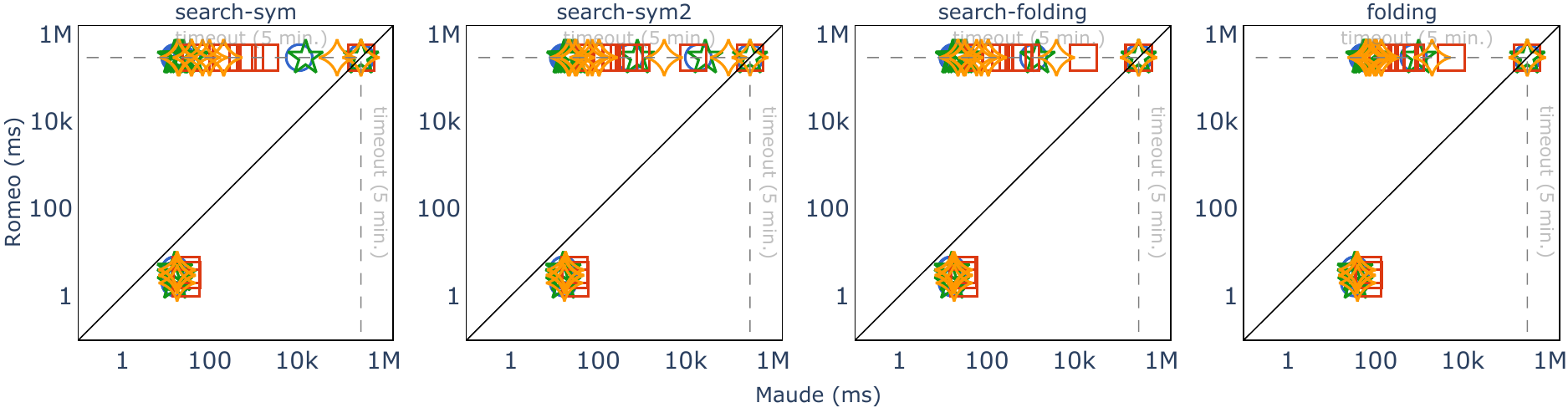}\\[-6pt]
    {(c) \texttt{tutorial}\label{fig:bench:C}}\vspace*{-2mm}
    \end{tabular}
     \caption{\EF-Synthesis. Execution times in log-scale for \romeo{} and Maude connected with
  Yices2 {(\symYices)}, Maude-SE with Yices2 {(\symSEYices)}, Z3 (\symZ), and CVC4 (\symCvc).\label{fig:bench1}}\vspace*{-2mm}
  \end{figure}%

\medskip
Figures~\ref{fig:bench1},~\ref{fig:bench2}, and \ref{fig:bench3} show
the execution times of \romeo{} and Maude in log-scale
for the six benchmarks. The diagonal line represents when the two systems
take the same time to analyze the property \EF{($p > n$)}
(\Cref{fig:bench1,fig:bench2})
and  1-safety
(\Cref{fig:bench3}). An item (or ``point'') above the diagonal line represents a problem
instance where Maude outperforms \romeo{}. Items in the horizontal
line at the top of
each figure represent instances where \romeo{} timed out and Maude was able to
complete the analysis. Items on the vertical line on the right
represent instances where Maude timed
out and \romeo{} completed the analysis.

Executing Maude with Yices2 is marginally faster than Maude-SE with Yices2,
and Maude connected with Yices2 outperforms executing Maude-SE connected to
the other two SMT solvers. For the \EF-synthesis problems in \Cref{fig:bench1,fig:bench2},
it is worth noticing  that: all  queries except
 $p_5 > 2$ (where \romeo{} also times out) could be solved in the
\code{producer-consumer} benchmark; the reduction of
the state space achieved by the folding procedures (commands
\code{search-folding} and \code{folding}) were needed to solve the
queries $\mathit{end}_3
> 2$ in \code{scheduler}, $\mathit{far}_1>2$ in \code{train1} and
\code{train2}, and $P_{12}>2$ in \code{abitpro};
and all the queries in \code{tutorial}
could be  solved by all the commands using Yices2 and CVC4. All the \AG-synthesis
problems (\Cref{fig:bench3}) could be solved by the command \code{AG-synthesis} (that uses the
implementation of folding at the meta-level).

In some reachability queries, Maude outperforms \romeo{}. More interestingly,
our approach terminates in cases where \romeo{} does not (items in the
horizontal line at the top of Figs.~\ref{fig:bench1}a and \ref{fig:bench1}c). Our
results are proven valid when injecting them into the model and running
\romeo{} with these additional constraints. This phenomenon happens when the
search order leads \romeo{} in the exploration of an infinite branch with an
unbounded marking.

\medskip
In \Cref{comp:Maude} we compare the performance of the different Maude
commands for \EF-synthesis and \AG-synthesis. For some instances,
\code{search-sym2}  is
marginally faster than  the command \code{search-sym}.
In general,
the command \code{search-folding}  is faster than
the command \code{folding}.
However, the extra
reduction of the search space (folding states in different branches of the
search tree) allows \code{folding} to solve some instances faster in the
\code{scheduling} benchmark (see the three items on the right in
\Cref{comp:Maude}b).

We have also compared the performance of the folding analysis presented in the
conference version of this paper \cite{DBLP:conf/apn/AriasBOOPR23} and the
current one. As explained in \Cref{sec:imp:folding}, the current implementation
uses the FME procedure implemented as an equational theory in Maude, while the
previous one relies on the procedure implemented in Z3. The results
are given in \Cref{fig:compare} and they clearly indicate that the new implementation
is more efficient than the previous one.

\clearpage

  \begin{figure}[!htbp]
  \vspace*{7mm}
      \centering
     \begin{tabular}{c}
    \includegraphics[width=0.98\textwidth]{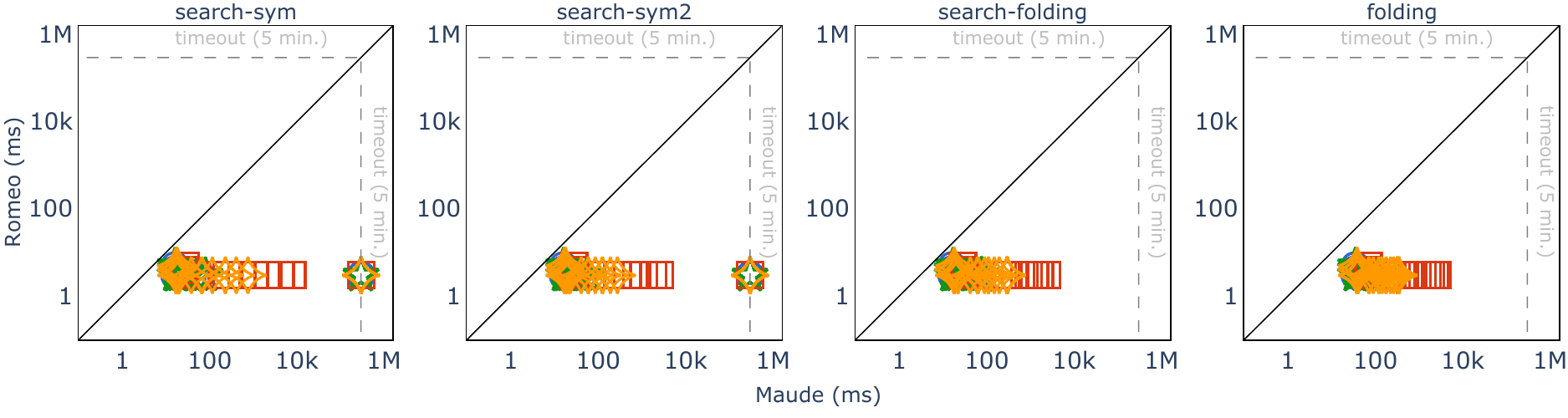}\\
    {(a) \texttt{abitpro}\label{fig:bench:D}}
     \\
    \includegraphics[width=0.98\textwidth]{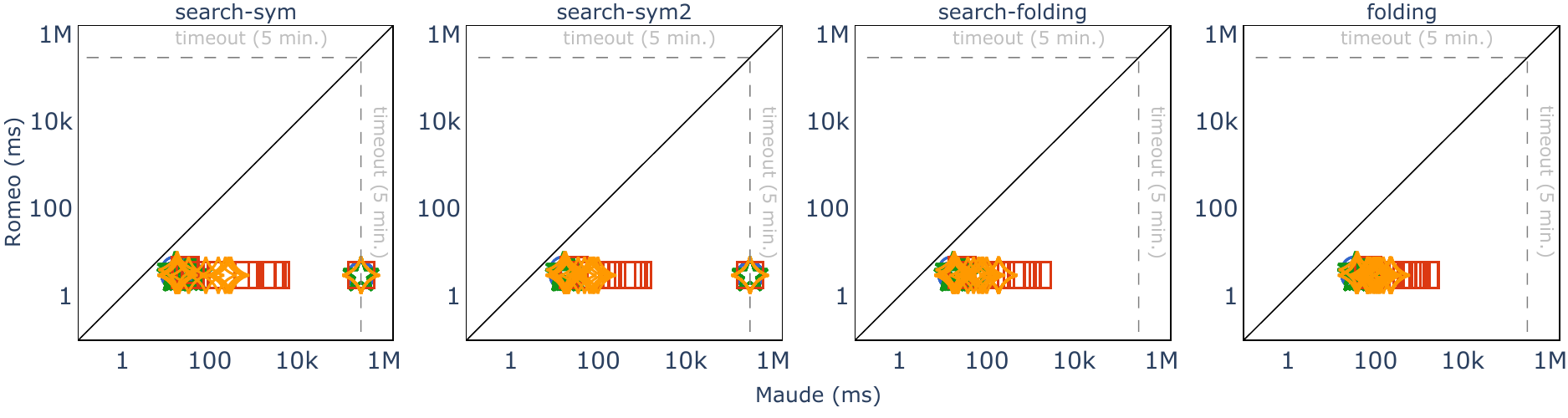}\\
    {(b) \texttt{train1}\label{fig:bench:E}}
   \\
    \includegraphics[width=0.98\textwidth]{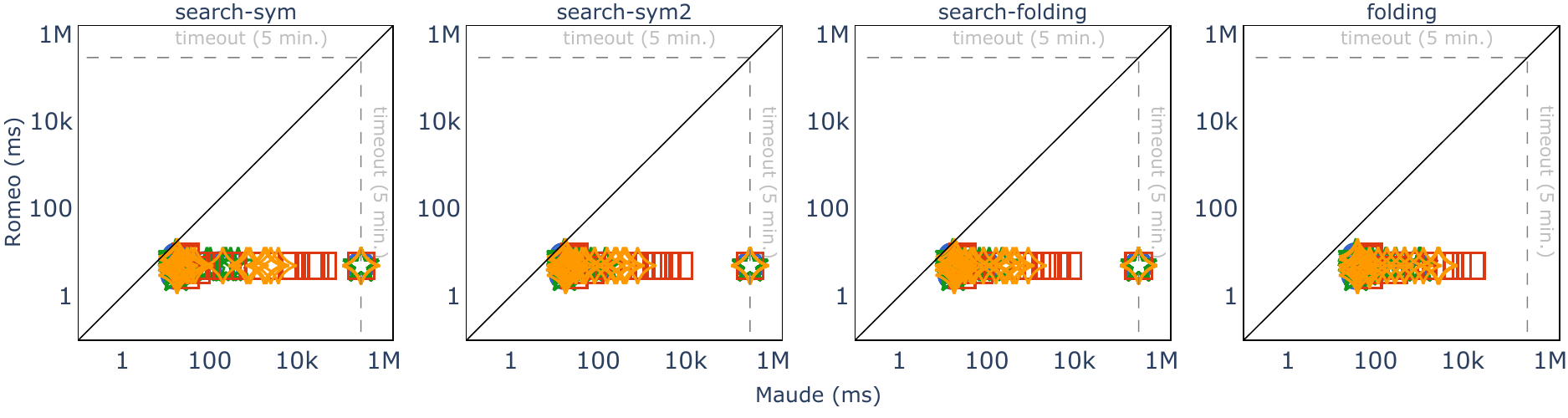}\\
    {(c) \texttt{train2}\label{fig:bench:F}}
  \end{tabular}
  \caption{\EF-Synthesis. Execution times in log-scale for \romeo{} and Maude connected with
  Yices2 {(\symYices)}, Maude-SE with Yices2 {(\symSEYices)}, Z3 (\symZ), and CVC4 (\symCvc).\label{fig:bench2}}
\end{figure}

\clearpage

\begin{figure}
    \centering
    \includegraphics[width=.8\textwidth]{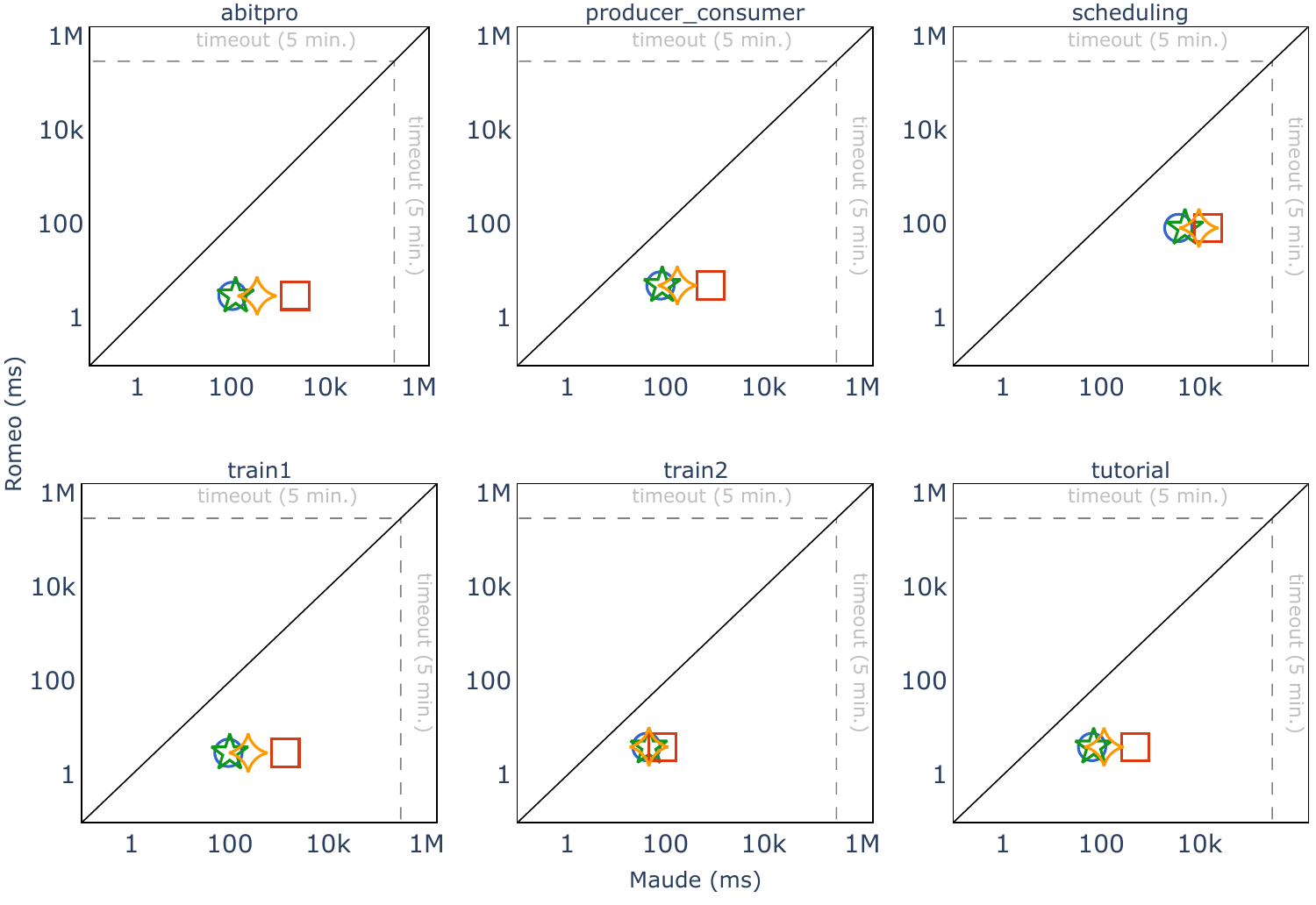}
  \caption{\AG-Synthesis. Execution times in log-scale for \romeo{} and Maude connected with
  Yices2 {(\symYices)}, Maude-SE with Yices2 {(\symSEYices)}, Z3 (\symZ), and CVC4 {(\symCvc)}.\label{fig:bench3}}
\end{figure}

\clearpage

\begin{figure}[h]
  \centering
  \begin{tabular}{c}
    \includegraphics[width=0.7\textwidth]{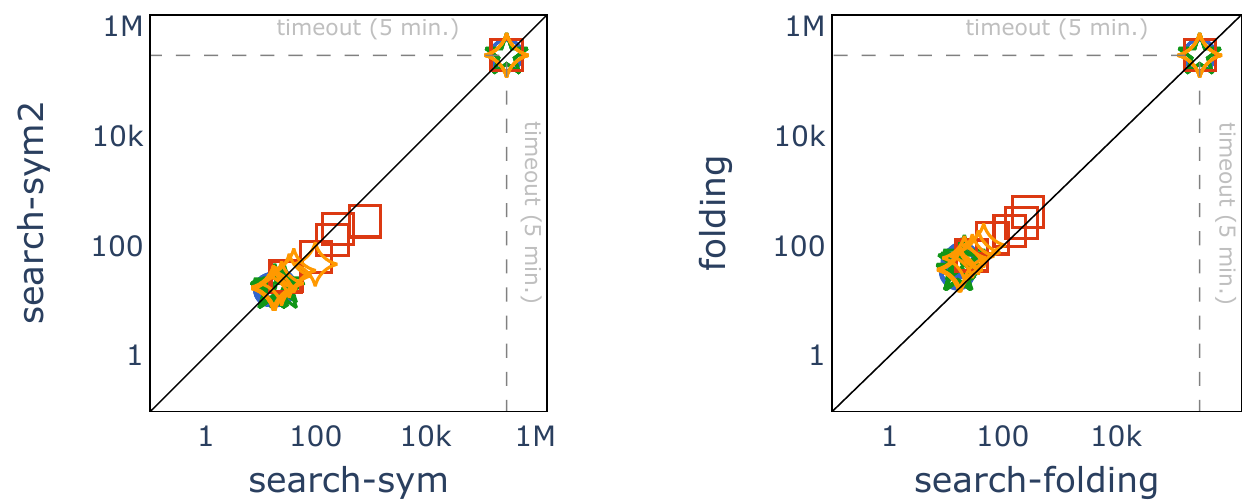}\\
   {(a) \texttt{producer-consumer}}
 \\
    \includegraphics[width=0.7\textwidth]{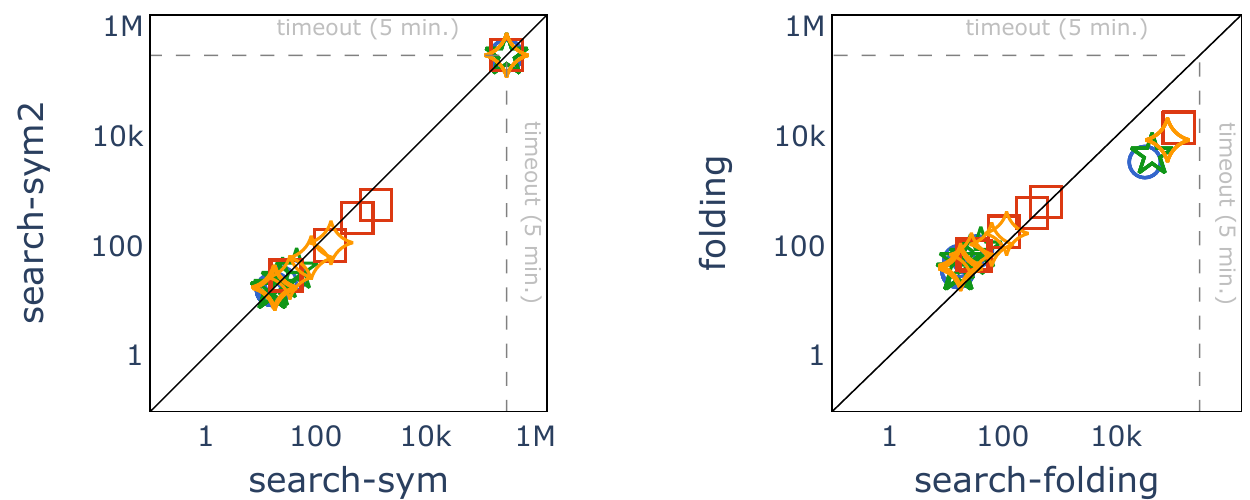}\\
  {(b) \texttt{scheduling}\label{fig:fold-sc}}
  \\
    \includegraphics[width=0.7\textwidth]{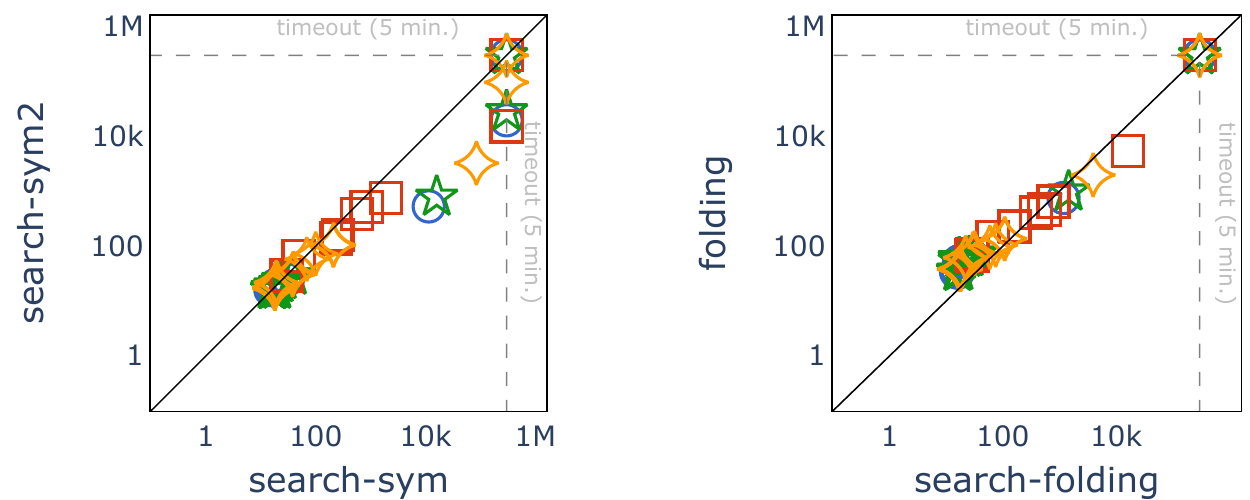}\\
     {(c) \texttt{tutorial}}
    \end{tabular}\\
  \caption{Comparison of the different Maude commands in log-scale.
Maude with Yices2 {(\symYices)}, Maude-SE with Yices2 {(\symSEYices)}, Z3 (\symZ), and CVC4 {(\symCvc)}.\label{comp:Maude}}
\end{figure}

\clearpage

\begin{figure}[!ht]
  \centering
  \includegraphics[width=.77\textwidth]{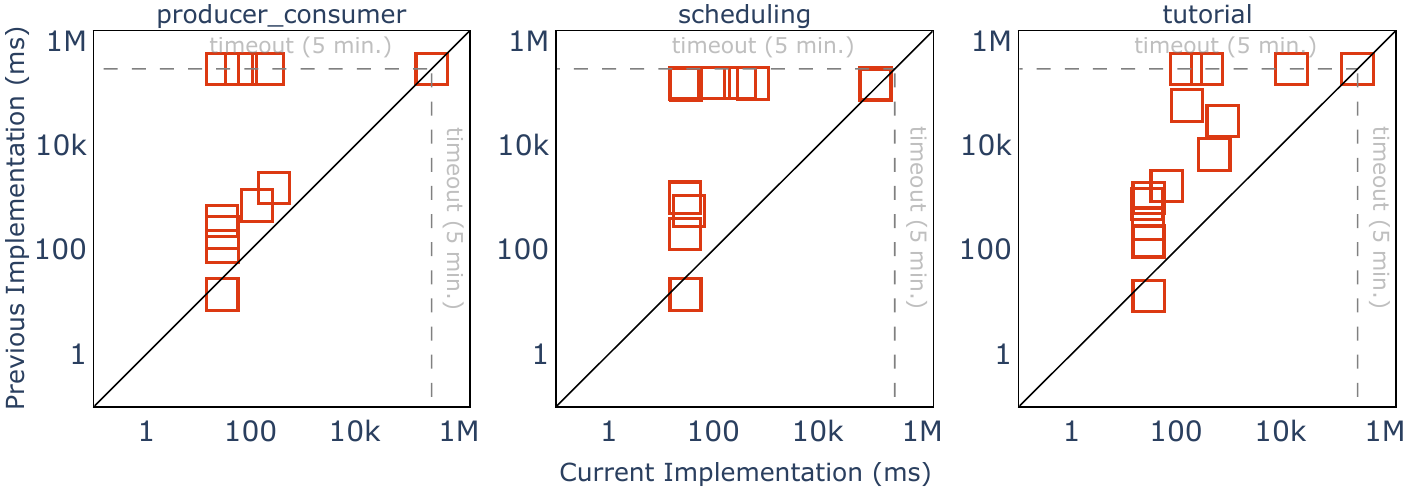}
  \caption{\label{fig:compare}Comparison of the implementation reported in this paper and the one in \cite{DBLP:conf/apn/AriasBOOPR23}. \label{comp:Maude-old}}\vspace*{-6mm}
\end{figure}



\section{Related work}   \label{sec:related}

\paragraph{Tool support for parametric time Petri nets.}
We are not aware of any  tool for analyzing parametric time(d)
Petri nets other than \romeo{}~\cite{romeo}. 

\paragraph{Petri nets in rewriting logic.} Formalizing
Petri nets algebraically~\cite{meseguer-montanari}  was one of the
inspirations behind rewriting  logic. Different kinds of Petri nets
are given a rewriting logic semantics in~\cite{petri-nets-in-maude},
and in~\cite{OlvMesTCS} for timed nets. In contrast to our paper, these papers
focus on the semantics of such nets, and do not consider execution and
analysis;  nor do they consider inhibitor arcs or parameters.
Capra~\cite{capra22,capra22b}, Padberg and Schultz~\cite{padberg},  and Barbosa et
al.~\cite{barbosa} use
Maude to formalize   dynamically
reconfigurable Petri nets (with inhibitor arcs) and I/O
Petri nets.  In
contrast to our work, these papers target untimed and non-parametric nets, and
do not focus on formal analysis, but only show examples
of standard (explicit-state) reachability analysis and LTL model checking.

Finally, we describe the differences between
this paper and its conference
version~\cite{DBLP:conf/apn/AriasBOOPR23}  in detail in the
introduction.

\paragraph{Symbolic methods for real-time systems in Maude.}
We  develop a symbolic rewrite
semantics and analysis for parametric time automata (PTA)
in~\cite{ftscs22} and \cite{ftscs-journal}.  The differences with the
current paper include:
PTAs are very simple structures
compared to PITPNs (with inhibitor arcs, no bounds on the number of
tokens in a state), so that the semantics of PITPNs is  more
sophisticated than the one for PTAs, which does not use, e.g.,
``structured'' states;
we could use  ``standard'' folding of symbolic states for PTAs compared to having to
develop a new folding mechanism  for PITPNs;
and so on.

Santiago Escobar and others recently developed a \emph{narrowing with SMT}
approach in rewriting logic to symbolically analyze parametric timed
automata extended with other features~\cite{DBLP:conf/ftscs/0001LS23}.
Essentially, rewriting
execution is
replaced by narrowing (``rewriting with logical variables'')
execution.    The key advantage of their approach
is that it allows analyzing symbolic states that
contain logical variables also of non-SMT values.   So far, we have
been able to treat all parameters in both parametric timed automata
and parametric time Petri nets, as well as parametric markings,  with
SMT variables. However, using narrowing  should widen the
scope of symbolic analyses of real-time systems, at the  cost of using
the much more complex narrowing instead of rewriting. Escobar et
al.~also develop  folding methods for merging symbolic states in
narrowing with SMT.

In addition,
a variety of real-time systems
have been formally analyzed using rewriting with SMT,
including
PLC ST programs~\cite{lee2022bounded},
virtually synchronous cyber-physical
systems~\cite{lee2022extension,hsaadl-sttt,lee2021hybrid},
and
soft agents~\cite{nigam2022automating}.
These papers differ from our work
in that
they use
guarded terms~\cite{bae2017guarded,bae2019symbolic}
for state-space reduction instead of folding,
and
do not consider parameter synthesis problems.

\section{Concluding remarks}  \label{sec:concl}

In this paper, we first provided a ``concrete'' rewriting logic
semantics for  (instantiated)
parametric time Petri nets with inhibitor arcs (PITPNs), and proved
that this semantics is bisimilar to the semantics
of such nets in~\cite{paris-paper}. However, this model is
non-executable; furthermore, explicit-state Maude analysis using
Real-Time Maude-style ``time sampling'' leads to unsound analysis for
dense-time systems such as PITPNs. We therefore \emph{systematically
  transformed} this model  into a ``symbolic'' rewriting logic model which  is
amenable to  sound and complete symbolic  analysis using Maude
combined with SMT solving. The resulting symbolic model also provides
a rewriting logic
 semantics for \emph{un-initialized} PITPNs.

\medskip
We have shown how  almost all formal analysis and parameter synthesis methods
supported by the state-of-the-art PITPN tool \romeo{} can be performed using
Maude-with-SMT. In addition, we have shown how Maude-with-SMT
can provide  additional capabilities for PITPNs, including
synthesizing also
 initial markings (and not just firing bounds) from \emph{parametric}
initial markings so that desired properties are satisfied,  full LTL
model checking, and analysis with user-defined execution
strategies.  We have developed  a new
``folding''  method for symbolic states, so that symbolic reachability
analysis using Maude-with-SMT terminates whenever  the corresponding
\romeo{} analysis terminates.

We have compared the performance of \romeo{}, our previous implementation
 in \cite{DBLP:conf/apn/AriasBOOPR23}, and the new commands
described in \Cref{sec:analysis:methods}.
The implementation of the Fourier-Motzkin elimination  procedure
directly as an equational theory allowed us to execute all the analyses with any
SMT solver connected to Maude.
The benchmarking shows that Maude combined with
Yices2 in many cases outperforms \romeo, whereas Maude combined with
Z3 and CVC4 is
significantly slower.
The results also show a considerable improvement
with respect to the implementation reported in \cite{DBLP:conf/apn/AriasBOOPR23}.
 We also experienced that \romeo{} sometimes did not find
(existing) solutions, and that the output of some executions included the message
``\emph{maybe}'' (indicating that \romeo\ has computed an
 approximation).

It is also worth pointing out that we analyze an ``interpreter'' for
PITPNs, where one rewrite theory is used to
analyze all PITPNs.  ``Compiling'' each PITPN to a different
symbolic rewrite theory (where each transition typically would be  modeled
 by a separate rewrite rule) might improve on the already
competitive performance of our interpreter; we should explore this in
future work.

Another advantage of our  high-level
implementation, using Maude and its
meta-programming features, is that it is very easy to develop, test,
and evaluate different analysis methods, maybe \emph{before} they
are efficiently implemented  in state-of-the-art tools such as \romeo.  The wealth of
analysis methods for PITPNs that we could  quickly define and
implement demonstrate  this convenience.

This paper has not only provided significant new analysis features
for PITPNs. It has
also shown that even a
model like our Real-Time Maude-inspired PITPN interpreter---with
functions, equations, and unbounded markings---can easily be turned
into a  symbolic rewrite theory for which Maude-with-SMT provides very
useful sound and complete analyses even for dense-time systems.  This
work have given us valuable insight into symbolic analysis of
real-time systems as we continue  our quest to extend ``symbolic
Real-Time Maude'' to
wider classes of real-time systems.

In future work we should also:
extend Maude's LTL model checker to a full SMT-based (with folding) timed
LTL and CTL model checker, thus covering all the analysis provided  by
\romeo{}; and develop a richer timed strategy language for controlling
the executions of PITPNs.

\paragraph{Acknowledgments.}
We would like to thank the anonymous reviewers for their very
insightful comments on an  earlier version of this paper.
Arias,
Olarte, {\"O}lveczky, and  Petrucci gratefully
acknowledge    support  from  
  the  PHC project Aurora AESIR.    Bae was supported by
  the National Research Foundation of Korea (NRF) grants  funded by
  the Korea government (MSIT) (No. 2021R1A5A1021944 and No. RS-2023-00251577).
  This work was also funded by
 the NATO Science for Peace and Security Programme through grant number
 G6133
 (project SymSafe).


\bigskip
\begin{thebibliography}{10}
\providecommand{\url}[1]{\texttt{#1}}
\providecommand{\urlprefix}{URL }
\expandafter\ifx\csname urlstyle\endcsname\relax
  \providecommand{\doi}[1]{doi:\discretionary{}{}{}#1}\else
  \providecommand{\doi}{doi:\discretionary{}{}{}\begingroup
  \urlstyle{rm}\Url}\fi
\providecommand{\eprint}[2][]{\url{#2}}

\bibitem{Merlin74}
Merlin PM.
\newblock A study of the recoverability of computing systems.
\newblock Ph.D. thesis, University of California, Irvine, CA, USA, 1974.
doi:10.5555/907383.

\bibitem{DBLP:reference/crc/VernadatB07}
Vernadat F, Berthomieu B.
\newblock State Space Abstractions for Time {P}etri Nets.
\newblock In: Son SH, Lee I, Leung JY (eds.), Handbook of Real-Time and
  Embedded Systems. Chapman and Hall/CRC, 2007.
\newblock \doi{10.1201/9781420011746.pt6}.

\bibitem{paris-paper}
Traonouez L, Lime D, Roux OH.
\newblock Parametric Model-Checking of Time {P}etri Nets with Stopwatches Using
  the State-Class Graph.
\newblock In: Formal Modeling and Analysis of Timed Systems ({FORMATS} 2008),
  volume 5215 of \emph{LNCS}. Springer, 2008 pp. 280--294.
\newblock \doi{10.1007/978-3-540-85778-5_20}.

\bibitem{DBLP:conf/formats/GrabiecTJLR10}
Grabiec B, Traonouez L, Jard C, Lime D, Roux OH.
\newblock Diagnosis Using Unfoldings of Parametric Time {P}etri Nets.
\newblock In: Formal Modeling and Analysis of Timed Systems ({FORMATS} 2010),
  volume 6246 of \emph{LNCS}. Springer, 2010 pp. 137--151.
\newblock \doi{10.1007/978-3-642-15297-9_12}.

\bibitem{EAGPLP13}
Andr\'e E, Pellegrino G, Petrucci L.
\newblock Precise Robustness Analysis of Time {Petri} Nets with Inhibitor Arcs.
\newblock In: Formal Modeling and Analysis of Timed Systems ({FORMATS}'13),
  volume 8053 of \emph{LNCS}. Springer, 2013 pp. 1--15.
  doi:10.1007/978-3-642-40229-6\_1.

\bibitem{DBLP:journals/fuin/LimeRS21}
Lime D, Roux OH, Seidner C.
\newblock Cost Problems for Parametric Time {P}etri Nets.
\newblock \emph{Fundam. Informaticae}, 2021.
\newblock \textbf{183}(1-2):97--123.
\newblock \doi{10.3233/FI-2021-2083}.

\bibitem{romeo}
Lime D, Roux OH, Seidner C, Traonouez L.
\newblock Romeo: {A} Parametric Model-Checker for {P}etri Nets with
  Stopwatches.
\newblock In: Tools and Algorithms for the Construction and Analysis of Systems
  ({TACAS} 2009), volume 5505 of \emph{LNCS}. Springer, 2009 pp. 54--57.
\newblock \doi{10.1007/978-3-642-00768-2_6}.

\bibitem{romeo-biology}
Andreychenko A, Magnin M, Inoue K.
\newblock Analyzing resilience properties in oscillatory biological systems
  using parametric model checking.
\newblock \emph{Biosystems}, 2016.
\newblock \textbf{149}:50--58.
\newblock \doi{https://doi.org/10.1016/j.biosystems.2016.09.002}.

\bibitem{romeo-aerial}
Parquier B, Rioux L, Henia R, Soulat R, Roux OH, Lime D, Andr{\'e} {\'E}.
\newblock Applying Parametric Model-Checking Techniques for Reusing Real-Time
  Critical Systems.
\newblock In: Formal Techniques for Safety-Critical Systems (FTSCS 2016),
  volume 694 of \emph{Communications in Computer and Information Science}.
  Springer, 2017 pp. 129--144.
\newblock \doi{https://doi.org/10.1007/978-3-319-53946-1_8}.

\bibitem{romeo-software}
Coullon H, Jard C, Lime D.
\newblock Integrated Model-Checking for the Design of Safe and Efficient
  Distributed Software Commissioning.
\newblock In: Integrated Formal Methods (IFM 2019), volume 11918 of
  \emph{LNCS}. Springer, Cham, 2019 pp. 120--137.
\newblock \doi{https://doi.org/10.1007/978-3-030-34968-4_7}.

\bibitem{Mes92}
Meseguer J.
\newblock Conditional Rewriting Logic as a Unified Model of Concurrency.
\newblock \emph{Theor. Comput. Sci.}, 1992.
\newblock \textbf{96}(1):73--155.
\newblock \doi{10.1016/0304-3975(92)90182-F}.

\bibitem{20-years}
Meseguer J.
\newblock Twenty years of rewriting logic.
\newblock \emph{J. Log. Algebraic Methods Program.}, 2012.
\newblock \textbf{81}(7-8):721--781.
\newblock \doi{10.1016/j.jlap.2012.06.003}.

\bibitem{maude-book}
Clavel M, Dur{\'{a}}n F, Eker S, Lincoln P, Mart{\'{\i}}{-}Oliet N, Meseguer J,
  Talcott CL.
\newblock All About Maude -- {A} High-Performance Logical Framework, volume
  4350 of \emph{LNCS}.
\newblock Springer, 2007.  doi:10.1007/978-3-540-71999-1.

\bibitem{tacas08}
{\"{O}}lveczky PC, Meseguer J.
\newblock The {R}eal-{T}ime {M}aude Tool.
\newblock In: Tools and Algorithms for the Construction and Analysis of Systems
  ({TACAS} 2008), volume 4963 of \emph{LNCS}. Springer, 2008 pp. 332--336.
\newblock \doi{10.1007/978-3-540-78800-3_23}.

\bibitem{wrla14}
{\"{O}}lveczky PC.
\newblock {Real-Time Maude} and Its Applications.
\newblock In: Rewriting Logic and Its Applications ({WRLA} 2014), volume 8663
  of \emph{LNCS}. Springer, 2014 pp. 42--79.
\newblock \doi{10.1007/978-3-319-12904-4_3}.

\bibitem{aer-journ}
{\"{O}}lveczky PC, Meseguer J, Talcott CL.
\newblock Specification and analysis of the {AER/NCA} active network protocol
  suite in {R}eal-{T}ime {M}aude.
\newblock \emph{Formal Methods Syst. Des.}, 2006.
\newblock \textbf{29}(3):253--293.
\newblock \doi{10.1007/s10703-006-0015-0}.

\bibitem{sefm09}
Lien E, {\"{O}}lveczky PC.
\newblock Formal Modeling and Analysis of an {IETF} Multicast Protocol.
\newblock In: Seventh {IEEE} International Conference on Software Engineering
  and Formal Methods ({SEFM} 2009). {IEEE} Computer Society, 2009 pp. 273--282.
\newblock \doi{10.1109/SEFM.2009.11}.

\bibitem{wsn-tcs}
{\"{O}}lveczky PC, Thorvaldsen S.
\newblock Formal modeling, performance estimation, and model checking of
  wireless sensor network algorithms in {R}eal-{T}ime {M}aude.
\newblock \emph{Theor. Comput. Sci.}, 2009.
\newblock \textbf{410}(2-3):254--280.
\newblock \doi{10.1016/j.tcs.2008.09.022}.

\bibitem{fase06}
{\"{O}}lveczky PC, Caccamo M.
\newblock Formal Simulation and Analysis of the {CASH} Scheduling Algorithm in
  {R}eal-{T}ime {M}aude.
\newblock In: Fundamental Approaches to Software Engineering ({FASE} 2006),
  volume 3922 of \emph{LNCS}. Springer, 2006 pp. 357--372.
\newblock \doi{10.1007/11693017_26}.

\bibitem{airplane-journ}
Bae K, Krisiloff J, Meseguer J, {\"{O}}lveczky PC.
\newblock Designing and verifying distributed cyber-physical systems using
  {M}ultirate {PALS:} An airplane turning control system case study.
\newblock \emph{Sci. Comput. Program.}, 2015.
\newblock \textbf{103}:13--50.
\newblock \doi{10.1016/j.scico.2014.09.011}.

\bibitem{kokichi}
Grov J, {\"{O}}lveczky PC.
\newblock Formal Modeling and Analysis of {G}oogle's {M}egastore in {Real-Time
  Maude}.
\newblock In: Specification, Algebra, and Software -- Essays Dedicated to
  Kokichi Futatsugi, volume 8373 of \emph{LNCS}. Springer, 2014 pp. 494--519.
\newblock \doi{10.1007/978-3-642-54624-2_25}.

\bibitem{sefm14}
Grov J, {\"{O}}lveczky PC.
\newblock Increasing Consistency in Multi-site Data Stores: {M}egastore-{CGC}
  and Its Formal Analysis.
\newblock In: Software Engineering and Formal Methods ({SEFM} 2014), volume
  8702 of \emph{LNCS}. Springer, 2014 pp. 159--174.
\newblock \doi{10.1007/978-3-319-10431-7_12}.

\bibitem{manets-journ}
Liu S, {\"{O}}lveczky PC, Meseguer J.
\newblock Modeling and analyzing mobile ad hoc networks in {R}eal-{T}ime
  {M}aude.
\newblock \emph{J. Log. Algebraic Methods Program.}, 2016.
\newblock \textbf{85}(1):34--66.
\newblock \doi{10.1016/J.JLAMP.2015.05.002}.

\bibitem{giovanna-journ}
Broccia G, Milazzo P, {\"{O}}lveczky PC.
\newblock Formal modeling and analysis of safety-critical human multitasking.
\newblock \emph{Innov. Syst. Softw. Eng.}, 2019.
\newblock \textbf{15}(3-4):169--190.
\newblock \doi{10.1007/s11334-019-00333-7}.

\bibitem{fmoods10}
{\"{O}}lveczky PC, Boronat A, Meseguer J.
\newblock Formal Semantics and Analysis of Behavioral {AADL} Models in
  {Real-Time Maude}.
\newblock In: Formal Techniques for Distributed Systems, Joint 12th {IFIP} {WG}
  6.1 International Conference, {FMOODS} 2010 and 30th {IFIP} {WG} 6.1 {FORTE}
  2010, volume 6117 of \emph{LNCS}. Springer, 2010 pp. 47--62.
\newblock \doi{10.1007/978-3-642-13464-7_5}.

\bibitem{musab-l}
AlTurki M, Dhurjati D, Yu D, Chander A, Inamura H.
\newblock Formal Specification and Analysis of Timing Properties in Software
  Systems.
\newblock In: Fundamental Approaches to Software Engineering ({FASE} 2009),
  volume 5503 of \emph{LNCS}. Springer, 2009 pp. 262--277.
\newblock \doi{10.1007/978-3-642-00593-0_18}.

\bibitem{ptolemy-journ}
Bae K, {\"{O}}lveczky PC, Feng TH, Lee EA, Tripakis S.
\newblock Verifying hierarchical {P}tolemy {II} discrete-event models using
  {Real-Time Maude}.
\newblock \emph{Sci. Comput. Program.}, 2012.
\newblock \textbf{77}(12):1235--1271.
\newblock \doi{10.1016/j.scico.2010.10.002}.

\bibitem{fase12}
Bae K, {\"{O}}lveczky PC, Meseguer J, Al{-}Nayeem A.
\newblock The {SynchAADL2Maude} Tool.
\newblock In: Fundamental Approaches to Software Engineering ({FASE} 2012),
  volume 7212 of \emph{LNCS}. Springer, 2012 pp. 59--62.
\newblock \doi{10.1007/978-3-642-28872-2_4}.

\bibitem{carolyn}
{\"{O}}lveczky PC.
\newblock Semantics, Simulation, and Formal Analysis of Modeling Languages for
  Embedded Systems in {Real-Time Maude}.
\newblock In: Formal Modeling: Actors, Open Systems, Biological Systems --
  Essays Dedicated to Carolyn Talcott on the Occasion of Her 70th Birthday,
  volume 7000 of \emph{LNCS}, pp. 368--402. Springer, 2011.
\newblock \doi{10.1007/978-3-642-24933-4_19}.

\bibitem{wrla06}
{\"{O}}lveczky PC, Meseguer J.
\newblock Abstraction and Completeness for {Real-Time Maude}.
\newblock In: 6th International Workshop on Rewriting Logic and its
  Applications ({WRLA} 2006), volume 174 of \emph{Electronic Notes in
  Theoretical Computer Science}. Elsevier, 2006 pp. 5--27.
\newblock \doi{10.1016/j.entcs.2007.06.005}.

\bibitem{maude-se}
Yu G, Bae K.
\newblock Maude-{SE}: a Tight Integration of {M}aude and {SMT} Solvers.
\newblock In: Preliminary proceedings of WRLA@ETAPS. 2020 pp. 220--232.
\texttt{https://wrla2020.webs.upv.es/pre-proceedings.pdf}

\bibitem{rtm-journ}
{\"{O}}lveczky PC, Meseguer J.
\newblock Semantics and pragmatics of {R}eal-{T}ime {M}aude.
\newblock \emph{High. Order Symb. Comput.}, 2007.
\newblock \textbf{20}(1-2):161--196.
\newblock \doi{10.1007/s10990-007-9001-5}.

\bibitem{ftscs22}
Arias J, Bae K, Olarte C, {\"{O}}lveczky PC, Petrucci L, R{\o}mming F.
\newblock Rewriting Logic Semantics and Symbolic Analysis for Parametric Timed
  Automata.
\newblock In: 8th {ACM} {SIGPLAN} International Workshop on Formal Techniques
  for Safety-Critical Systems ({FTSCS} 2022). {ACM}, 2022 pp. 3--15.
\newblock \doi{10.1145/3563822.3569923}.

\bibitem{ftscs-journal}
Arias J, Bae K, Olarte C, {\"{O}}lveczky PC, Petrucci L, R{\o}mming F.
\newblock Symbolic Analysis and Parameter Synthesis for Networks of Parametric
  Timed Automata with Global Variables using Maude and SMT Solving.
\newblock \emph{Science of Computer Programming}, 2024.
\newblock \textbf{233}.
\newblock \doi{10.1016/j.scico.2023.103074}.

\bibitem{DBLP:conf/apn/AriasBOOPR23}
Arias J, Bae K, Olarte C, {\"{O}}lveczky PC, Petrucci L, R{\o}mming F.
\newblock Symbolic Analysis and Parameter Synthesis for Time {P}etri Nets Using
  {M}aude and {SMT} Solving.
\newblock In: Application and Theory of Petri Nets and Concurrency (PETRI NETS
  2023), volume 13929 of \emph{LNCS}. Springer, 2023 pp. 369--392.
\newblock \doi{10.1007/978-3-031-33620-1_20}.

\bibitem{pitpn2maude}
Arias J, Bae K, Olarte C, {\"{O}}lveczky PC, Petrucci L.
\newblock {PITPN2Maude}, 2024.
\newblock
  \urlprefix\url{https://depot.lipn.univ-paris13.fr/real-time-maude/pitpn2maude-journal}.

\bibitem{model-checking}
Clarke EM, Grumberg O, Peled DA.
\newblock Model Checking.
\newblock {MIT} Press, 2001. ISBN: 9780262032704.

\bibitem{rocha-rewsmtjlamp-2017}
Rocha C, Meseguer J, Mu{\~{n}}oz CA.
\newblock Rewriting modulo {SMT} and open system analysis.
\newblock \emph{J. Log. Algebraic Methods Program.}, 2017.
\newblock \textbf{86}(1):269--297.
\newblock \doi{10.1016/j.jlamp.2016.10.001}.

\bibitem{DBLP:conf/apn/LeclercqLR23}
Leclercq L, Lime D, Roux OH.
\newblock A State Class Based Controller Synthesis Approach for Time {P}etri
  Nets.
\newblock In: Application and Theory of Petri Nets and Concurrency ({PETRI}
  {NETS} 2023), volume 13929 of \emph{LNCS}. Springer, 2023 pp. 393--414.
\newblock \doi{10.1007/978-3-031-33620-1_21}.

\bibitem{bae2019symbolic}
Bae K, Rocha C.
\newblock Symbolic state space reduction with guarded terms for rewriting
  modulo {SMT}.
\newblock \emph{Sci. Comput. Program.}, 2019.
\newblock \textbf{178}:20--42.
\newblock \doi{10.1016/j.scico.2019.03.006}.

\bibitem{petri-nets-in-maude}
Stehr M, Meseguer J, {\"{O}}lveczky PC.
\newblock Rewriting Logic as a Unifying Framework for {P}etri Nets.
\newblock In: Unifying Petri Nets, Advances in Petri Nets, volume 2128 of
  \emph{Lecture Notes in Computer Science}. Springer, 2001 pp. 250--303.
\newblock \doi{10.1007/3-540-45541-8_9}.

\bibitem{Wang1998}
Wang J.
\newblock Time {P}etri Nets.
\newblock In: Timed Petri Nets: Theory and Application, pp. 63--123. Springer,
  1998.
\newblock \doi{10.1007/978-1-4615-5537-7_4}.

\bibitem{DBLP:journals/jlap/Meseguer20}
Meseguer J.
\newblock Generalized rewrite theories, coherence completion, and symbolic
  methods.
\newblock \emph{J. Log. Algebraic Methods Program.}, 2020.
\newblock \textbf{110}.
\newblock \doi{10.1016/j.jlamp.2019.100483}.

\bibitem{bae2013abstract}
Bae K, Escobar S, Meseguer J.
\newblock Abstract Logical Model Checking of Infinite-State Systems Using
  Narrowing.
\newblock In: Rewriting Techniques and Applications ({RTA} 2013), volume~21 of
  \emph{LIPIcs}. Schloss Dagstuhl - Leibniz-Zentrum f{\"{u}}r Informatik, 2013
  pp. 81--96.
\newblock \doi{10.4230/LIPIcs.RTA.2013.81}.

\bibitem{dantzig2006linear}
Dantzig G, Thapa M.
\newblock Linear Programming 1: Introduction.
\newblock Springer Series in Operations Research and Financial Engineering.
  Springer New York, 2006.
\newblock ISBN 9780387226330.

\bibitem{DBLP:journals/jucs/TraonouezLR09}
Traonouez L, Lime D, Roux OH.
\newblock Parametric Model-Checking of Stopwatch {P}etri Nets.
\newblock \emph{J. Univers. Comput. Sci.}, 2009.
\newblock \textbf{15}(17):3273--3304.
\newblock \doi{10.3217/jucs-015-17-3273}.

\bibitem{maude-manual}
Clavel M, Dur{\'a}n F, Eker S, Escobar S, Lincoln P, Mart\'{\i}-Oliet N,
  Meseguer J, Rubio R, Talcott C.
\newblock Maude Manual (Version 3.3.1).
\newblock SRI International, 2023.
\newblock Available at \url{http://maude.cs.illinois.edu}.

\bibitem{DBLP:conf/apn/RouxL04}
Roux OH, Lime D.
\newblock Time {P}etri Nets with Inhibitor Hyperarcs. {F}ormal Semantics and
  State Space Computation.
\newblock In: Cortadella J, Reisig W (eds.), Applications and Theory of Petri
  Nets 2004, 25th International Conference, {ICATPN} 2004, volume 3099 of
  \emph{LNCS}. Springer, 2004 pp. 371--390.
\newblock \doi{10.1007/978-3-540-27793-4_21}.

\bibitem{meseguer-montanari}
Meseguer J, Montanari U.
\newblock Petri Nets Are Monoids.
\newblock \emph{Information and Computation}, 1990.
\newblock \textbf{88}(2):105--155.
\newblock \doi{10.1016/0890-5401(90)90013-8}.

\bibitem{OlvMesTCS}
{\"{O}}lveczky PC, Meseguer J.
\newblock Specification of real-time and hybrid systems in rewriting logic.
\newblock \emph{Theor. Comput. Sci.}, 2002.
\newblock \textbf{285}(2):359--405.
\newblock \doi{10.1016/S0304-3975(01)00363-2}.

\bibitem{capra22}
Capra L.
\newblock Rewriting Logic and {P}etri Nets: {A} Natural Model for
  Reconfigurable Distributed Systems.
\newblock In: Distributed Computing and Intelligent Technology ({ICDCIT} 2022),
  volume 13145 of \emph{LNCS}. Springer, 2022 pp. 140--156.
\newblock \doi{10.1007/978-3-030-94876-4_9}.

\bibitem{capra22b}
Capra L.
\newblock Canonization of Reconfigurable {PT} Nets in {M}aude.
\newblock In: Reachability Problems ({RP} 2022), volume 13608 of \emph{LNCS}.
  Springer, 2022 pp. 160--177.
\newblock \doi{10.1007/978-3-031-19135-0_11}.

\bibitem{padberg}
Padberg J, Schulz A.
\newblock Model Checking Reconfigurable {P}etri Nets with {M}aude.
\newblock In: 9th International Conference on Graph Transformation ({ICGT}
  2016), volume 9761 of \emph{LNCS}. Springer, 2016 pp. 54--70.
\newblock \doi{10.1007/978-3-319-40530-8_4}.

\bibitem{barbosa}
Barbosa PES, Barros JP, Ramalho F, Gomes L, Figueiredo J, Moutinho F, Costa A,
  Aranha A.
\newblock Sys{V}eritas: {A} Framework for Verifying {IOPT} Nets and Execution
  Semantics within Embedded Systems Design.
\newblock In: Technological Innovation for Sustainability - Second {IFIP} {WG}
  5.5/SOCOLNET Doctoral Conference on Computing, Electrical and Industrial
  Systems (DoCEIS 2011), volume 349 of \emph{{IFIP} Advances in Information and
  Communication Technology}. Springer, 2011 pp. 256--265.
\newblock \doi{10.1007/978-3-642-19170-1_28}.

\bibitem{DBLP:conf/ftscs/0001LS23}
Escobar S, L{\'{o}}pez{-}Rueda R, Sapi{\~{n}}a J.
\newblock Symbolic Analysis by Using Folding Narrowing with Irreducibility and
  {SMT} Constraints.
\newblock In: 9th {ACM} {SIGPLAN} International Workshop on Formal Techniques
  for Safety-Critical Systems ({FTSCS} 2023). {ACM}, 2023 pp. 14--25.
\newblock \doi{10.1145/3623503.3623537}.

\bibitem{lee2022bounded}
Lee J, Kim S, Bae K.
\newblock Bounded Model Checking of {PLC} {ST} Programs using Rewriting Modulo
  {SMT}.
\newblock In: 8th ACM SIGPLAN International Workshop on Formal Techniques for
  Safety-Critical Systems (FTSCS 2022). ACM, 2022 pp. 56--67.
  doi:10.1145/3563822.3568016.

\bibitem{lee2022extension}
Lee J, Bae K, {\"O}lveczky PC.
\newblock An extension of {HybridSynchAADL} and its application to
  collaborating autonomous {UAVs}.
\newblock In: Leveraging Applications of Formal Methods, Verification and
  Validation. Adaptation and Learning (ISoLA 2022), volume 13703 of
  \emph{LNCS}. Springer, 2022 pp. 47--64.
doi:10.1007/978-3-031-19759-8\_4


\bibitem{hsaadl-sttt}
Lee J, Bae K, {\"O}lveczky PC, Kim S, Kang M.
\newblock Modeling and formal analysis of virtually synchronous cyber-physical
  systems in {AADL}.
\newblock \emph{International Journal on Software Tools for Technology
  Transfer}, 2022.
\newblock \textbf{24}(6):911--948.
doi:10.1007/S10009-022-00665-Z.


\bibitem{lee2021hybrid}
Lee J, Kim S, Bae K, {\"{O}}lveczky PC.
\newblock {HybridSynchAADL}: Modeling and Formal Analysis of Virtually
  Synchronous {CPSs} in {AADL}.
\newblock In: Computer Aided Verification ({CAV} 2021), volume 12759 of
  \emph{LNCS}. Springer, 2021 pp. 491--504.
\newblock \doi{10.1007/978-3-030-81685-8_23}.

\bibitem{nigam2022automating}
Nigam V, Talcott CL.
\newblock Automating Safety Proofs About Cyber-Physical Systems Using Rewriting
  Modulo {SMT}.
\newblock In: Rewriting Logic and Its Applications (WRLA 2022), volume 13252 of
  \emph{LNCS}. Springer, 2022 pp. 212--229.
\newblock \doi{10.1007/978-3-031-12441-9_11}.

\bibitem{bae2017guarded}
Bae K, Rocha C.
\newblock Guarded terms for rewriting modulo {SMT}.
\newblock In: International Conference on Formal Aspects of Component Software
  (FACS 2017). Springer, 2017 pp. 78--97.  doi:10.1007/978-3-319-68034-7\_5.

\bibitem{DBLP:books/daglib/0023756}
Jensen K, Kristensen LM.
\newblock Coloured {P}etri Nets -- Modelling and Validation of Concurrent
  Systems.
\newblock Springer, 2009.
\newblock \doi{10.1007/b95112}.

\bibitem{pn23-techrep}
Arias J, Bae K, Olarte C, {\"{O}}lveczky PC, Petrucci L, R{\o}mming F.
\newblock Symbolic Analysis and Parameter Synthesis for Time {P}etri Nets Using
  {M}aude and {SMT} Solving, 2023.
\newblock \doi{10.48550/ARXIV.2303.08929}.

\bibitem{AD94}
Alur R, Dill DL.
\newblock A Theory of Timed Automata.
\newblock \emph{Theor. Comput. Sci.}, 1994.
\newblock \textbf{126}(2):183--235.
\newblock \doi{10.1016/0304-3975(94)90010-8}.

\bibitem{DLLMP15}
David A, Larsen KG, Legay A, Mikucionis M, Poulsen DB.
\newblock Uppaal {SMC} tutorial.
\newblock \emph{Int. J. Softw. Tools Technol. Transf.}, 2015.
\newblock \textbf{17}(4):397--415.
\newblock \doi{10.1007/s10009-014-0361-y}.

\bibitem{Larsen2015}
Larsen KG, Mikucionis M, Taankvist JH.
\newblock Safe and Optimal Adaptive Cruise Control.
\newblock In: Correct System Design -- Symposium in Honor of
  Ernst-R{\"{u}}diger Olderog on the Occasion of His 60th Birthday, volume 9360
  of \emph{LNCS}. Springer, 2015 pp. 260--277.
\newblock \doi{10.1007/978-3-319-23506-6_17}.

\bibitem{lpw:tacas98}
Lindahl M, Pettersson P, Yi W.
\newblock Formal design and analysis of a gear controller.
\newblock \emph{Int. J. Softw. Tools Technol. Transf.}, 2001.
\newblock \textbf{3}(3):353--368.
\newblock \doi{10.1007/s100090100048}.

\bibitem{ARINC22}
Choi Et, Kim Th, Jun YK, Lee S, Han M.
\newblock On-the-Fly Repairing of Atomicity Violations in {ARINC} 653 Software.
\newblock \emph{Applied Sciences}, 2022.
\newblock \textbf{12}(4).
\newblock \doi{10.3390/app12042014}.

\bibitem{7273419}
Wang Y, Wang R, Guan Y, Li X, Wei H, Zhang J.
\newblock Formal Modeling and Verification of the Safety Critical Fire-Fighting
  Control System.
\newblock In: 39th Annual Computer Software and Applications Conference,
  {COMPSAC} Workshops 2015. {IEEE} Computer Society, 2015 pp. 536--541.
\newblock \doi{10.1109/COMPSAC.2015.181}.

\bibitem{AHV93}
Alur R, Henzinger TA, Vardi MY.
\newblock Parametric real-time reasoning.
\newblock In: Kosaraju SR, Johnson DS, Aggarwal A (eds.), Proceedings of the
  Twenty-Fifth Annual {ACM} Symposium on Theory of Computing, May 16-18, 1993,
  San Diego, CA, {USA}. {ACM}, 1993 pp. 592--601.
\newblock \doi{10.1145/167088.167242}.

\bibitem{Andre21}
Andr{\'{e}} {\'{E}}.
\newblock {IMITATOR} 3: Synthesis of Timing Parameters Beyond Decidability.
\newblock In: Computer Aided Verification ({CAV} 2021), volume 12759 of
  \emph{LNCS}. Springer, 2021 pp. 552--565.
\newblock \doi{10.1007/978-3-030-81685-8_26}.

\bibitem{HRSV02}
Hune T, Romijn J, Stoelinga M, Vaandrager FW.
\newblock Linear parametric model checking of timed automata.
\newblock \emph{J. Log. Algebraic Methods Program.}, 2002.
\newblock \textbf{52-53}:183--220.
\newblock \doi{10.1016/S1567-8326(02)00037-1}.

\bibitem{KP12}
Knapik M, Penczek W.
\newblock Bounded Model Checking for Parametric Timed Automata.
\newblock \emph{Trans. Petri Nets Other Model. Concurr.}, 2012.
\newblock \textbf{5}:141--159.
\newblock \doi{10.1007/978-3-642-29072-5_6}.

\bibitem{JLR15}
Jovanovic A, Lime D, Roux OH.
\newblock Integer Parameter Synthesis for Real-Time Systems.
\newblock \emph{{IEEE} Trans. Software Eng.}, 2015.
\newblock \textbf{41}(5):445--461.
\newblock \doi{10.1109/TSE.2014.2357445}.

\bibitem{CEFX09}
Chevallier R, Encrenaz{-}Tiph{\`{e}}ne E, Fribourg L, Xu W.
\newblock Timed verification of the generic architecture of a memory circuit
  using parametric timed automata.
\newblock \emph{Formal Methods Syst. Des.}, 2009.
\newblock \textbf{34}(1):59--81.
\newblock \doi{10.1007/s10703-008-0061-x}.

\bibitem{FLMS12}
Fribourg L, Soulat R, Lesens D, Moro P.
\newblock Robustness Analysis for Scheduling Problems Using the Inverse Method.
\newblock In: 19th International Symposium on Temporal Representation and
  Reasoning ({TIME} 2012). {IEEE} Computer Society, 2012 pp. 73--80.
\newblock \doi{10.1109/TIME.2012.10}.

\bibitem{manet-journ}
Liu S, {\"{O}}lveczky PC, Meseguer J.
\newblock Modeling and analyzing mobile ad hoc networks in {R}eal-{T}ime
  {M}aude.
\newblock \emph{J. Log. Algebraic Methods Program.}, 2016.
\newblock \textbf{85}(1):34--66.
\newblock \doi{10.1016/j.jlamp.2015.05.002}.

\bibitem{cloud-chapter}
Bobba R, Grov J, Gupta I, Liu S, Meseguer J, {\"O}lveczky PC, Skeirik S.
\newblock Survivability: Design, Formal Modeling, and Validation of Cloud
  Storage Systems Using {Maude}.
\newblock In: Assured Cloud Computing, chapter~2, pp. 10--48. John Wiley \&
  Sons, 2018. \texttt{https://doi.org/10.1002/9781119428497.ch2.}

\bibitem{PTA-benchmarks}
Andr{\'{e}} {\'{E}}, Marinho D, van~de Pol J.
\newblock A Benchmarks Library for Extended Parametric Timed Automata.
\newblock In: Tests and Proofs ({TAP} 2021), volume 12740 of \emph{LNCS}.
  Springer, 2021 pp. 39--50.
\newblock \doi{10.1007/978-3-030-79379-1_3}.

\bibitem{lepri-journ}
Lepri D, {\'{A}}brah{\'{a}}m E, {\"{O}}lveczky PC.
\newblock Sound and complete timed {CTL} model checking of timed {K}ripke
  structures and real-time rewrite theories.
\newblock \emph{Sci. Comput. Program.}, 2015.
\newblock \textbf{99}:128--192.
\newblock \doi{10.1016/j.scico.2014.06.006}.

\bibitem{meseguer-membership-1998}
Meseguer J.
\newblock Membership algebra as a logical framework for equational
  specification.
\newblock In: Recent Trends in Algebraic Development Techniques (WADT'97),
  volume 1376 of \emph{LNCS}. Springer, 1997 pp. 18--61.
\newblock \doi{10.1007/3-540-64299-4_26}.

\bibitem{DBLP:conf/rta/EscobarM07}
Escobar S, Meseguer J.
\newblock Symbolic Model Checking of Infinite-State Systems Using Narrowing.
\newblock In: Baader F (ed.), Term Rewriting and Applications ({RTA} 2007),
  volume 4533 of \emph{LNCS}. Springer, 2007 pp. 153--168.
\newblock \doi{10.1007/978-3-540-73449-9_13}.

\bibitem{DBLP:conf/wrla/BaeM14}
Bae K, Meseguer J.
\newblock Infinite-State Model Checking of {LTLR} Formulas Using Narrowing.
\newblock In: Escobar S (ed.), Rewriting Logic and Its Applications ({WRLA}
  2014), volume 8663 of \emph{LNCS}. Springer, 2014 pp. 113--129.
\newblock \doi{10.1007/978-3-319-12904-4_6}.

\bibitem{Andre19STTT}
Andr{\'{e}} {\'{E}}.
\newblock What's decidable about parametric timed automata?
\newblock \emph{Int. J. Softw. Tools Technol. Transf.}, 2019.
\newblock \textbf{21}(2):203--219.
\newblock \doi{10.1007/s10009-017-0467-0}.

\bibitem{NPP18}
Nguyen HG, Petrucci L, van~de Pol J.
\newblock Layered and Collecting {NDFS} with Subsumption for Parametric Timed
  Automata.
\newblock In: 23rd International Conference on Engineering of Complex Computer
  Systems ({ICECCS} 2018). {IEEE} Computer Society, 2018 pp. 1--9.
\newblock \doi{10.1109/ICECCS2018.2018.00009}.

\bibitem{AFS13atva}
Andr{\'{e}} {\'{E}}, Fribourg L, Soulat R.
\newblock Merge and Conquer: State Merging in Parametric Timed Automata.
\newblock In: Automated Technology for Verification and Analysis ({ATVA} 2013),
  volume 8172 of \emph{LNCS}. Springer, 2013 pp. 381--396.
\newblock \doi{10.1007/978-3-319-02444-8_27}.

\bibitem{BBBC16}
Bezdek P, Benes N, Barnat J, Cern{\'{a}} I.
\newblock {LTL} Parameter Synthesis of Parametric Timed Automata.
\newblock In: Software Engineering and Formal Methods ({SEFM} 2016), volume
  9763 of \emph{LNCS}. Springer, 2016 pp. 172--187.
\newblock \doi{10.1007/978-3-319-41591-8_12}.

\bibitem{DBLP:journals/ijfcs/AndreCFE09}
Andr{\'{e}} {\'{E}}, Chatain T, Fribourg L, Encrenaz E.
\newblock An Inverse Method for Parametric Timed Automata.
\newblock \emph{Int. J. Found. Comput. Sci.}, 2009.
\newblock \textbf{20}(5):819--836.
\newblock \doi{10.1142/S0129054109006905}.

\bibitem{DBLP:journals/tse/BucciFSV04}
Bucci G, Fedeli A, Sassoli L, Vicario E.
\newblock Timed State Space Analysis of Real-Time Preemptive Systems.
\newblock \emph{{IEEE} Trans. Software Eng.}, 2004.
\newblock \textbf{30}(2):97--111.
\newblock \doi{10.1109/TSE.2004.1265815}.
\end{thebibliography}
\end{document}